\documentclass[a4paper, 11pt]{article}
\usepackage[margin=25mm]{geometry}
\usepackage{booktabs}
\usepackage{multirow}
\usepackage{footnote}
\usepackage[shortlabels]{enumitem}
\usepackage{float}
\usepackage{comment}
\usepackage{amsmath}
\usepackage{amsfonts}
\usepackage{titlesec}
\usepackage{bm}
\usepackage{lscape} 
\usepackage{graphicx,comment}
\usepackage{float}
\usepackage{url}
\usepackage{color}
\usepackage{xargs, xstring}
\usepackage{subfig}
\usepackage{verbatim}
\usepackage{setspace} 
\usepackage[round]{natbib}

\providecommand{\keywords}[1]
{
  \small	
  \textbf{\textit{Keywords---}} #1
}

\newcommand{\be}{\begin{eqnarray}}
\newcommand{\ee}{\end{eqnarray}}
\newcommand{\bee}{\begin{eqnarray*}}
\newcommand{\eee}{\end{eqnarray*}}
\newcommand{\bi}{\begin{enumerate}}
\newcommand{\ei}{\end{enumerate}}

\usepackage[noabbrev,capitalize]{cleveref}
\crefformat{equation}{(#2#1#3)}

\usepackage{algorithm}
\usepackage{bbm}
\usepackage[table]{xcolor}
\usepackage{enumitem}
\usepackage{rotating}
\usepackage{tikz}
\usetikzlibrary{shapes.geometric, arrows}
\usetikzlibrary{automata,positioning}
\usepackage{multirow}

\def\bSig\boldsymbol{\Sigma}

\newcommand{\indep}{\perp \!\!\! \perp}
\usepackage{algorithm}
\usepackage{algpseudocode}
\newtheorem{theorem}{Theorem}

\usepackage{nohyperref}
\usepackage{url}

\title{Estimating Marginal Treatment Effect in Cluster Randomized Trials with
Multi-level Missing Outcomes}
\author{Chia-Rui Chang$^{1}$ and Rui Wang$^{1,2}$  \\
        \\
        \small $^{1}$Department of Biostatistics, Harvard T. H. Chan School of Public Health, Massachusetts, USA \\
        \small $^{2}$Department of Population Medicine, Harvard Pilgrim Health Care Institute and Harvard Medical School, \\
        \small Massachusetts, USA \\
}
\date{} 

\begin{document}
\maketitle

\doublespacing
\begin{abstract}
Analyses of cluster randomized trials (CRTs) can be complicated by informative missing outcome data.  Methods such as inverse probability weighted generalized estimating equations have been proposed to account for informative missingness by weighting the observed individual outcome data in each cluster. These existing methods have focused on settings where missingness occurs at the individual level and each cluster has partially or fully observed individual outcomes. In the presence of missing clusters, e.g., all outcomes from a cluster are missing due to drop-out of the cluster, these approaches effectively ignore this cluster-level missingness and can lead to biased inference if the cluster-level missingness is informative. Informative missingness at multiple levels can also occur in CRTs with a multi-level structure where study participants are nested in subclusters such as health care providers, and the subclusters are nested in clusters such as clinics.  In this paper, we propose new estimators for estimating the marginal treatment effect in CRTs accounting for missing outcome data at multiple levels based on weighted generalized estimating equations. We show that the proposed multi-level multiply robust estimator is consistent and asymptotically normally distributed provided that one set of the propensity score models is correctly specified. We evaluate the performance of the proposed method through extensive simulation and illustrate its use with a CRT evaluating a Malaria risk-reduction intervention in rural Madagascar.\\
\end{abstract} \hspace{10pt}

\noindent 

\section{Introduction}
\label{s:intro}

Cluster-randomized trial (CRT), which randomizes all individuals in the same cluster to receive the same treatment, are commonly used in biomedical research for intervention evaluation \citep{hayes2017cluster}. Because outcomes from individuals within the same cluster are likely to be correlated, analysis of CRTs must account for this dependence within cluster. 
The generalized estimating equation (GEE) approach has often been adopted to estimate the marginal treatment effect in CRTs \citep{liang1986longitudinal}. Compared to mixed effects models, GEE targets the population marginal effect parameter and requires fewer parametric assumptions on the outcome distribution \citep{hubbard2010gee}. It renders valid inference provided that the mean model is correctly specified and is robust to misspecification of the correlation structure.  However, in the presence of informative missing outcome data, the GEE estimator based on complete data may result in biased estimates (see for example, \cite{prague2016accounting}).    



Here we consider the setting where outcome missingness depends on fully observed baseline covariates and exposure, as commonly considered in the analyses of CRT data  
(see for example, \cite{hossain2017missinga,hossain2017missingb}). Such missingness process has been termed as ``covariate-dependent missingness (CDM)", a stronger version of the missing at random (MAR) mechanism as the missingness process is assumed to be unaffected by the observed outcomes. Two common approaches to address CDM data include multilevel multiple imputation (MMI) \citep{schafer2002computational, diazordaz2016multiple, hossain2017missinga, hossain2017missingb} and inverse probability weighting (IPW) \citep{robins1995analysis}. We adopt the IPW framework, which avoids the need to correctly specify the joint distribution of clustered outcomes.

The IPW methods for handling missing individual outcomes in CRTs have been proposed \citep{prague2016accounting, chen2020stochastic}.  
These methods are developed to handle the settings where missingness occurs at the individual level and each cluster has partially or fully observed individual outcomes. In the presence of missing clusters, e.g., all outcomes from a cluster are missing due to drop-out of the cluster, these approaches ignore this cluster-level missingness and can lead to biased inference if the cluster-level missingness is informative \citep{giraudeau2009preventing}. 

Missing clusters are not uncommon in CRTs. Two systematic reviews of CRTs reported that 18\% of 132 trials and 31\% of 86 trials had missing clusters \citep{diaz2014missing, fiero2016statistical}. 
Multi-level missingness can also occur in CRTs with a multi-level structure where study participants are nested in subclusters such as households and households are nested in regions. For example, in a study to evaluate if proactive community case management (pro-CCM) is effective in reducing malaria burden in rural endemic area of Madagascar, twenty-two fokontanies (smallest administrative units) were randomized to pro-CCM or conventional integrated community case management (iCCM) \citep{ratovoson2022proactive}. The study participants were nested in households, which were nested in each fokontany. About 24\% of the study participants and 22\% of the households were lost to follow-up due to moving away, absence, death, or refusal to participate.  



In this paper, we develop new estimators for estimating the marginal treatment effect in CRTs with multi-level missing outcomes based on the GEE framework.  
We derive the multi-level weights to account for informative missingness at both the (sub)cluster- and individual-level.
We further incorporate a multiply robust estimation approach \citep{han2014multiply} to the multi-level missing outcome setting and propose a multi-level multiply robust GEE (MMR-GEE) estimator. 
This novel estimator allows analysts to specify multiple sets of propensity score (PS) models and leads to consistent estimation of the marginal treatment effect provided that one set of models is correctly specified. To address potential misclassification of the cluster-level missingness indicators, that is, a cluster with all individual outcomes missing may be misclassified as cluster drop-out, we develop an  Expectation-Maximization (EM) algorithm to estimate the parameters of the PS models \citep{dempster1977maximum}.  


The remainder of the article is organized as follows. Section \ref{ss:2.1} introduces the notation for the multi-level missingness setting and provides a review of the unweighted GEE and IPW-GEE approaches. Section \ref{ss:2.2} describes the multi-level missingness processes and the assumptions made throughout the paper. Section \ref{ss:2.3} and \ref{ss:2.4} present the multi-level inverse probability weights and our proposed MMR-GEE estimator. Section \ref{ss:2.5} establishes the theoretical properties of the MMR-GEE estimator. Section \ref{ss:2.6} addresses the misclassification issue of the observed missingness indicator at the cluster level. For notational simplicity, we anchor our presentation around two-level CRTs, where informative missingness may occur both at the cluster and at the individual level in Sections \ref{ss:2.1}-\ref{ss:2.6}. In Section \ref{ss:2.7} we present an extension to three-level CRTs where informative missingness can occur both at the subcluster and at the individual level. Results from extensive simulation studies are reported in Section \ref{s:simulation}. In Section \ref{s:applications}, we illustrate the use of the proposed methods with 
a three-level CRT: the ``Proactive Community Case Management in Rural Madagascar" study \citep{ratovoson2022proactive}. The paper is concluded with practical considerations and discussions in Section \ref{s:discussion}. 

\section{Methods}

\subsection{Notation and Models}
\label{ss:2.1}

We consider a two-arm parallel CRT with outcome $Y_{ij}$, a vector of $P$ cluster-level baseline covariates $\boldsymbol{Z}_i = (Z_{i}^{1},Z_{i}^{2},\ldots, Z_{i}^{P})^\mathrm{T}$, and a vector of $Q$ individual-level baseline covariates $\boldsymbol{X}_{ij} = (X_{ij}^{1},X_{ij}^{2},\ldots,X_{ij}^{Q})^\mathrm{T}$ for subject $j = 1,...,n_i$ in cluster $i = 1,...,M$. Let $A_i \in \{0,1\}$ be the binary treatment indicator for cluster $i$ (treated $A_i = 1$ and control $A_i = 0$); the treatment assignment probability is known and given by $P(A_i = 1) = p_A$. The vector of cluster-level covariates $\boldsymbol{Z}_i$ and matrix of individual-level covariates $\boldsymbol{X}_{i} = (\boldsymbol{X}_{i1}, \boldsymbol{X}_{i2}, \ldots, \boldsymbol{X}_{in_i})^\mathrm{T}$ are assumed to be fully observed before randomization. 
Here, we use $\boldsymbol{R}_{i} = (R_{i1}, R_{i2}, \ldots, R_{in_i})^\mathrm{T}$ to denote the vector of individual-level missingness indicator and $C_i$ to denote the cluster-level missingness indicator for outcomes $\boldsymbol{Y}_i = (Y_{i1}, Y_{i2}, \ldots, Y_{in_i})^\mathrm{T}$. $R_{ij} = 1$ when $Y_{ij}$ is observed and $R_{ij} = 0$ when $Y_{ij}$ is missing. $C_i = 0$ if the cluster $i$ drops out of the study so that no individual outcomes in that cluster can be observed, and $C_i = 1$ otherwise. 

Let $s = \sum_{i=1}^{M} C_i$ be the number of observed clusters and $m_i = \sum_{j=1}^{n_i} R_{ij}$ be the number of observed outcomes in cluster $i$. Without loss of generality, let $i = 1,...,s$ be the indexes for observed clusters and $i = s+1,...,M$ be the indexes for missing clusters; for participants in observed cluster $i$, let $j = 1,...,m_i$ be the indexes of participants whose outcomes are observed and $j = m_{i}+1,...,n_i$ be the indexes of participants whose outcomes are missing. 
See Table \ref{tab:1} for an illustration of the data structure under multi-level missingness settings. 
\begin{table}[H]
\caption{Data structure of multi-level missingness in CRTs, including outcome $Y_{ij}$, cluster-level missingness indicator $C_i$ ($C_{i} = 1$ if the cluster is observed and $C_{i} = 0$ if the cluster is missing), and individual-level missingness indicator $R_{ij}$ ($R_{ij} = 1$ if $Y_{ij}$ is observed and $R_{ij} = 0$ if $Y_{ij}$ is missing)}
\label{tab:1}
\small
\begin{center}
\rowcolors{2}{white}{gray!10}
\begin{tabular}{lllll}
\hline
Cluster & Unit  & $Y_{ij}$ & $R_{ij}$ & $C_{i}$ \\ 
\hline
1 & 1 & $Y_{11}$ & 1 & 1 \\
1 & $\vdots$ & $\vdots$ & 1 & 1 \\
1 & $m_1$ & $Y_{1m_1}$ & 1 & 1 \\
1 & $m_1 + 1$ & $Y_{1(m_1 + 1)}$ & 0 & 1   \\
1 & $\vdots$ & $\vdots$ & 0 & 1  \\
1 & $n_1$ & $Y_{1n_1}$ & 0 & 1 \\ 
\hline
$\vdots$ & $\vdots$ & $\vdots$ & $\vdots$ & 1 \\ 
\hline
s & 1 & $Y_{s1}$ & 1 & 1 \\
s & $\vdots$ & $\vdots$ & 1 & 1 \\
s & $m_s$ & $Y_{sm_s}$ & 1 & 1 \\
s & $m_s + 1$ & $Y_{s(m_s + 1)}$ & 0 & 1 \\
s & $\vdots$ & $\vdots$ & 0 & 1 \\
s & $n_s$ & $Y_{sn_s}$ & 0 & 1 \\
\hline
s+1 & 1 & $Y_{(s+1)1}$ & 0 & 0 \\
s+1 & $\vdots$ & $\vdots$ & 0 & 0  \\
s+1 & $n_{s+1}$ & $Y_{(s+1)n_{s+1}}$ & 0 & 0 \\
\hline
$\vdots$ & $\vdots$ & $\vdots$ & 0 & 0 \\
\hline
M & 1 & $Y_{M1}$ & 0 & 0 \\
M & $\vdots$ & $\vdots$ & 0 & 0 \\
M & $n_M$ & $Y_{Mn_M}$ & 0 & 0 \\
\hline
\end{tabular}
\end{center}
\end{table}
Our primary interest lies in the marginal mean model $g(\boldsymbol{\mu}_{ij}(\boldsymbol{\beta}, A_{i})) = g(\mathbb{E}[Y_{ij}| A_i]) = \beta_I + \beta_A A_i$ with link function $g(.)$, where we make inference on the parameters $\boldsymbol{\beta} = (\beta_I, \beta_A)^\mathrm{T}$. For continuous outcomes and identity link function, $\beta_A$ corresponds to the difference-in-means marginal treatment effect. When there is no missing data, an estimator of $\boldsymbol{\beta}$ can be obtained by solving the following estimating equation \citep{liang1986longitudinal}:
\begin{equation} \label{model:GEE}
    \sum_{i=1}^{M} \underbrace{\frac{\partial{\boldsymbol{\mu}_i(\boldsymbol{\beta}, A_i)}}{\partial{\boldsymbol{\beta}}}\boldsymbol{V}_{i}^{-1}(\boldsymbol{Y}_i - \boldsymbol{\mu}_i(\boldsymbol{\beta}, A_i))}_{=U_i(\boldsymbol{Y}_i,A_i; \boldsymbol{\beta})} = 0.
\end{equation}
$\frac{\partial{\boldsymbol{\mu}_i(\boldsymbol{\beta}, A_i)}}{\partial{\boldsymbol{\beta}}}$ is the design matrix with $\boldsymbol{\mu}_i(\boldsymbol{\beta}, A_i) = \{\boldsymbol{\mu}_{i1}(\boldsymbol{\beta}, A_i), \ldots, \boldsymbol{\mu}_{in_i}(\boldsymbol{\beta}, A_i)\}^\mathrm{T}$. $\boldsymbol{V}_{i} = \boldsymbol{F}_{i}^{1/2}\boldsymbol{C}(\boldsymbol{\alpha})\boldsymbol{F}_{i}^{1/2}$ is the covariance matrix with $\boldsymbol{F}_{i} = \text{diag}\{\text{var}(y_{ij})\}$. $\boldsymbol{C}(\boldsymbol{\alpha})$ is the working correlation matrix indexed by non-diagonal elements $\boldsymbol{\alpha}$. 

When outcomes are missing under CDM, fitting Model (\ref{model:GEE}) with complete data could lead to biased inference for $\boldsymbol{\beta}$ \citep{prague2016accounting}. Provided that all clusters are observed, i.e., $C_i = 1$ for all $i$, one can attempt to correct the bias through IPW-GEE \citep{robins1995analysis}: 
\begin{equation}\label{model:IPW-GEE}
    \sum_{i=1}^{M} \frac{\partial{\boldsymbol{\mu}_i(\boldsymbol{\beta}, A_i)}}{\partial{\boldsymbol{\beta}}}\boldsymbol{V}_{i}^{-1}\boldsymbol{W}_{i}(A_i, \boldsymbol{Z}_{i}, \boldsymbol{X}_{i};\boldsymbol{\theta})(\boldsymbol{Y}_i - \boldsymbol{\mu}_i(\boldsymbol{\beta}, A_i)) = 0.
\end{equation}
Model (\ref{model:IPW-GEE}) recovers population moments by reweighing the complete data according to the $n_i \times n_i$ matrix of weights $\boldsymbol{W}_{i}(A_i, \boldsymbol{Z}_{i}, \boldsymbol{X}_{i}; \boldsymbol{\theta}) = \text{diag}\left[\frac{R_{ij}}{\pi_{ij}(A_i, \boldsymbol{Z}_{i}, \boldsymbol{X}_{ij};\boldsymbol{\theta})} \right]_{j=1,\ldots,n_i}$. The conditional probability that $Y_{ij}$ is observed, also called the PS, is denoted by $\pi_{ij}(A_i, \boldsymbol{Z}_{i}, \boldsymbol{X}_{ij};\boldsymbol{\theta}) \stackrel{def}{=} P(R_{ij} = 1|A_i, \boldsymbol{Z}_{i}, \boldsymbol{X}_{ij})$. In practice, the true PS is unknown and needs to be estimated. A common approach is to postulate a logistic model that regresses the missingness indicator on the treatment indicator and baseline covariates. A consistent and asymptotically normal (CAN) estimator of $\boldsymbol{\beta}$ can be obtained  when the PS model is correctly specified. 



\subsection{Multi-level Missingness Processes and Assumptions}
\label{ss:2.2}

We consider the following multi-level missingness processes: clusters drop out or withdraw from the study after randomization, where this cluster-level missingness is induced by the model $\lambda_{i}(A_i, \boldsymbol{Z}_i; \boldsymbol{\gamma}) \stackrel{def}{=} P(C_i = 1 | A_i, \boldsymbol{Z}_i)$ with parameters $\boldsymbol{\gamma}$. For clusters that remain throughout the study, the outcomes of individual participants may be missing, where this individual-level missingness is induced by another model $\phi_{ij}(A_i, \boldsymbol{Z}_i, \boldsymbol{X}_{ij}|C_i = 1; \boldsymbol{\eta}) \stackrel{def}{=} P(R_{ij} = 1 |C_i = 1, A_i, \boldsymbol{Z}_i, \boldsymbol{X}_{ij})$ with parameters $\boldsymbol{\eta}$. In such settings, the standard IPW-GEE method based on data from the clusters with at least partially observed individual outcomes ignores the cluster-level missingness process and may lead to biased estimates of $\boldsymbol{\beta}$.
Throughout the paper, we make the following assumptions: \begin{enumerate}[label=\textbf{\arabic*})]
    \item \textbf{Non-informative cluster size:} 
    We assume that the cluster sizes are non-informative of outcomes but could affect the missingness processes.
    \item \textbf{Multi-level CDM:} The multi-level missingness processes depend on neither observed nor missing outcomes, conditional on baseline covariates and treatment: $(C_i, \boldsymbol{R}_{i}) \indep \boldsymbol{Y}_i|A_i, \boldsymbol{Z}_i, \boldsymbol{X}_i$. 
    \item \textbf{Conditional exchangeability in $\boldsymbol{C}_i$ over $\boldsymbol{X}_i$:} The cluster-level missingness process is independent of individual-level covariates, conditional on cluster-level covariates and treatment indicator: $C_i \indep \boldsymbol{X}_i \mid A_i, \boldsymbol{Z}_i$. Note that $\boldsymbol{Z}_i$ can contain summary measure of $\boldsymbol{X}_i$ (e.g. average age of a cluster, proportion of males/females)
    and we allow $C_i$ to depend on such summary measure but not on a particular $X_{ij}$. This assumption implies that $\mathrm{P}\left(C_{i}=1 \mid A_i, \boldsymbol{Z}_{i}, \boldsymbol{X}_{i}\right) = \mathrm{P}\left(C_{i}=1 \mid A_i, \boldsymbol{Z}_{i}\right)$
    \item \textbf{No covariate interference for individual-level missingness process}: For participant $j$ in observed cluster $i$, other participants' individual-level covariates ($\boldsymbol{X}_{ij^\prime}$ for $j^\prime \neq j$) cannot affect  participant $j$'s missingness processes, conditional on baseline covariates and treatment indicator: $R_{ij} \indep \boldsymbol{X}_{ij^\prime} \mid C_i = 1, A_i, \boldsymbol{Z}_i, \boldsymbol{X}_{ij}$ for all $j^\prime \neq j$. This implies that $P\left(R_{i j}=1 \mid C_{i}=1, A_i, \boldsymbol{Z}_{i}, \boldsymbol{X}_{i}\right) = P\left(R_{i j}=1 \mid C_{i}=1, A_i, \boldsymbol{Z}_{i}, \boldsymbol{X}_{ij}\right)$.
\end{enumerate}    

\subsection{Multi-level IPW-GEE (MIPW-GEE)} 
\label{ss:2.3}
We adapt weighting methods from the longitudinal drop-out setting
\citep{robins1995analysis, mitani2022accounting} to estimate $\boldsymbol{\beta}$ under the multi-level CDM setting. The conditional probability of observing $Y_{ij}$ can be expressed as:
\begin{align}\label{eq:cond_prob}
     P(R_{ij}C_{i} = 1| A_i, \boldsymbol{Z}_{i},\boldsymbol{X}_{i}) 
    &= P(R_{ij} = 1, C_{i} = 1|A_i, \boldsymbol{Z}_{i},\boldsymbol{X}_{i}) \hspace{0.2cm}  \nonumber \\
    &= P(R_{ij} = 1|C_{i} = 1,A_i, \boldsymbol{Z}_{i},\boldsymbol{X}_{i})P(C_{i} = 1|A_i, \boldsymbol{Z}_{i},\boldsymbol{X}_{i})  \nonumber  \\
    &= P(R_{ij} = 1|C_{i} = 1,A_i, \boldsymbol{Z}_{i},\boldsymbol{X}_{ij})P(C_{i} = 1|A_i, \boldsymbol{Z}_{i}),
\end{align}
where $P(C_{i} = 1|A_i,\boldsymbol{Z}_{i})$ corresponds to the cluster-level missingness process induced by $\lambda_i(A_i, \boldsymbol{Z}_{i}; \boldsymbol{\gamma})$ and $P(R_{ij} = 1|C_{i} = 1, A_i, \boldsymbol{Z}_{i}, \boldsymbol{X}_{ij})$ corresponds to the individual-level missingness process induced by $\phi_{ij}(A_i, \boldsymbol{Z}_{i}, \boldsymbol{X}_{ij} \mid C_i = 1; \boldsymbol{\eta})$. By modifying the weighting matrix of (\ref{model:IPW-GEE}), we propose the multi-level IPW-GEE (MIPW-GEE) estimator as follows:
\begin{align}\label{model:MW}
    \boldsymbol{W}_{i}(A_i, \boldsymbol{Z}_{i}, \boldsymbol{X}_{i}; \boldsymbol{\gamma}, \boldsymbol{\eta}) 
    = \text{diag}\left[\frac{R_{ij}C_{i}}{\phi_{ij}(A_i, \boldsymbol{Z}_{i}, \boldsymbol{X}_{ij} \mid C_i = 1; \boldsymbol{\eta})\lambda_{i}(A_i, \boldsymbol{Z}_{i}; \boldsymbol{\gamma}) } \right]_{j=1,\ldots, n_i}.
\end{align}
The consistency of the MIPW-GEE estimator requires the correct specification of both PS models: that is, $\phi_{ij}(A_i, \boldsymbol{Z}_i, \boldsymbol{X}_{ij}|C_i = 1; \boldsymbol{\eta}) = P(R_{ij} = 1 |C_i = 1, A_i, \boldsymbol{Z}_i, \boldsymbol{X}_{ij})$ and $\lambda_i(A_i, \boldsymbol{Z}_{i}; \boldsymbol{\gamma}) = P(C_{i} = 1|A_i, \boldsymbol{Z}_{i})$ for some $\boldsymbol{\eta}$ and
$\boldsymbol{\gamma}$.

\subsection{Multi-level Multiply Robust GEE (MMR-GEE)}
\label{ss:2.4}
The consistency of the MIPW-GEE estimator depends on the correct specification of both the cluster- and individual-level PS models. 
To protect against misspecification of thte PS models, we propose a multiply robust estimator of $\boldsymbol{\beta}$, denoted by $\hat{\boldsymbol{\beta}}_{MR}$, based on the empirical likelihood theory \citep{ owen2001empirical, qin2009empirical, han2014multiply}. The proposed multi-level multiply robust GEE (MMR-GEE) estimator allows analysts to specify multiple sets of PS models, and $\hat{\boldsymbol{\beta}}_{MR}$ will be a consistent estimator for $\boldsymbol{\beta}$ provided that  one of the cluster-level and one of the individual-level PS models are correctly specified.

The MMR-GEE estimator can be obtained by solving estimating equation (\ref{model:IPW-GEE}) with weights replaced by the multiply robust weights:
\begin{align}
\label{model:MRW}
    \boldsymbol{W}_{i}^{MR} = \text{diag}\left[
    R_{ij}C_{i}w_{ij}^{MR}\right]_{j = 1, \ldots, n_i}.
\end{align} 
We derive the multiply robust weights by extending the method from the independent data setting in \cite{han2013estimation} to the clustered data setting where informative missingness can occur at multiple levels.  
Let $\mathcal{P}_{1} = \{\phi_{ij}^{k}(\boldsymbol{\eta}^k): k = 1,...,K \}$ denote the set of $K$ postulated individual-level PS models for $\phi_{i j}\left(A_{i}, \boldsymbol{Z}_{i}, \boldsymbol{X}_{i j} \mid C_{i}=1 ; \boldsymbol{\eta}\right)$ and $\mathcal{P}_{2} = \{\lambda_{i}^{\ell}(\boldsymbol{\gamma}^{\ell}): \ell = 1,...,L \}$ denote the set of $L$ postulated cluster-level PS models for $\lambda_{i}\left(A_{i}, \boldsymbol{Z}_{i} ; \gamma\right)$, where $\boldsymbol{\eta}^k$ and $\boldsymbol{\gamma}^{\ell}$ are vectors of parameters for the $k$th and $\ell$th models. Let $\hat{\boldsymbol{\eta}}^{k}$ and $\hat{\boldsymbol{\gamma}}^{\ell}$ be the estimators for $\boldsymbol{\eta}^k$ and $\boldsymbol{\gamma}^{\ell}$, respectively. Now define $\chi^{k,l}(\hat{\boldsymbol{\eta}}^k, \hat{\boldsymbol{\gamma}}^{\ell}) = \left(\sum_{i=1}^{M} n_i\right)^{-1}\sum_{i=1}^{M}\sum_{j=1}^{n_i} \phi_{ij}^{k}(\hat{\boldsymbol{\eta}}^k)\lambda_{i}^{\ell}(\hat{\boldsymbol{\gamma}}^{\ell})$.
The multiply robust weights for individuals with observed outcomes $i = 1,...,s$, $j = 1,...,m_i$ can be obtained from solving a constrained optimization problem:
\begin{align}\label{model:MRW_CO}
    \hat{w}_{ij}^{MR} = \underset{w_{ij}}{\text{argmax}}\prod_{i=1}^{s}\prod_{j=1}^{m_i} w_{ij},
\end{align}
subject to the following constraints:
\begin{align*}
    w_{ij} \geq 0, \hspace{0.5cm} \sum_{i=1}^{s}\sum_{j=1}^{m_i}w_{ij} = 1, \hspace{0.5cm} \sum_{i=1}^{s}\sum_{j=1}^{m_i} w_{ij}\phi_{ij}^{k}(\boldsymbol{\hat{\eta}}^k)\lambda_{i}^{\ell}(\boldsymbol{\hat{\gamma}}^{\ell})  = \chi^{k,l}(\hat{\boldsymbol{\eta}}^k, \hat{\boldsymbol{\gamma}}^{\ell}) \hspace{0.25cm} (k = 1,...,K, \ell = 1,...,L).
\end{align*}
The first constraint requires that the weights are nonnegative. The second constraint imposes that the weights sum up to 1. The third constraint weighs each postulated model evaluated at the biased samples to represent the population mean. For $i = 1,...,M$, $j = 1,...,n_i$, now define $\hat{\boldsymbol{\eta}} = \{(\hat{\boldsymbol{\eta}}^{1})^\mathrm{T},...,(\hat{\boldsymbol{\eta}}^{K})^\mathrm{T}\}$, $\hat{\boldsymbol{\gamma}} = \{(\hat{\boldsymbol{\gamma}}^{1})^\mathrm{T},...,(\hat{\boldsymbol{\gamma}}^{L})^\mathrm{T}\}$, and 
\begin{align*}
    \hat{g}_{ij}(\hat{\boldsymbol{\eta}},\hat{\boldsymbol{\gamma}}) 
        = \{& \phi_{ij}^{1}(\boldsymbol{\hat{\eta}}^1)\lambda_{i}^{1}(\boldsymbol{\hat{\gamma}}^1) - \chi^{1,1}(\hat{\boldsymbol{\eta}}^1, \hat{\boldsymbol{\gamma}}^1),...,\phi_{ij}^{1}(\boldsymbol{\hat{\eta}}^1)\lambda_{i}^{L}(\boldsymbol{\hat{\gamma}}^{L}) - \chi^{1,L}(\hat{\boldsymbol{\eta}}^1, \hat{\boldsymbol{\gamma}}^{L}), \\
        &\phi_{ij}^{2}(\boldsymbol{\hat{\eta}}^2)\lambda_{i}^{1}(\boldsymbol{\hat{\gamma}}^1) - \chi^{2,1}(\hat{\boldsymbol{\eta}}^2, \hat{\boldsymbol{\gamma}}^1),...,\phi_{ij}^{2}(\boldsymbol{\hat{\eta}}^2)\lambda_{i}^{L}(\boldsymbol{\hat{\gamma}}^{L}) - \chi^{2,L}(\hat{\boldsymbol{\eta}}^2, \hat{\boldsymbol{\gamma}}^{L}), \ldots, \\
        &\phi_{ij}^{K}(\boldsymbol{\hat{\eta}}^K)\lambda_{i}^{1}(\boldsymbol{\hat{\gamma}}^1) - \chi^{K,1}(\hat{\boldsymbol{\eta}}^K, \hat{\boldsymbol{\gamma}}^1),...,\phi_{ij}^{K}(\boldsymbol{\hat{\eta}}^K)\lambda_{i}^{L}(\boldsymbol{\hat{\gamma}}^{L}) - \chi^{K,L}(\hat{\boldsymbol{\eta}}^K, \hat{\boldsymbol{\gamma}}^{L}) \}^\mathrm{T}.
\end{align*}
The constrained optimization problem (\ref{model:MRW_CO}) can be solved through the Lagrange multiplier technique, which yields:
\begin{align}\label{eq:MRW_sol1}
    \left. \hat{w}_{ij}^{MR} = \left\{\frac{1}{1 + \hat{\boldsymbol{\rho}}^T \hat{g}_{ij}(\hat{\boldsymbol{\eta}},\hat{\boldsymbol{\gamma}})}\right\}
    \middle/ \left\{\sum_{i=1}^{s}m_i \right\} \right. \hspace{0.25cm} i = 1,...,s, j = 1,...,m_i,
\end{align}
where $\hat{\boldsymbol{\rho}}^T = (\hat{\rho}_{11},\hat{\rho}_{12},...,\hat{\rho}_{KL})$ is a $(KL) \times 1$ vector by solving the following equation:
\begin{align}\label{eq:MRW_sol2}
     \sum_{i=1}^{s}\sum_{j=1}^{m_i} \frac{ \hat{g}_{ij}(\hat{\boldsymbol{\eta}},\hat{\boldsymbol{\gamma}})}{1 + \boldsymbol{\rho}^T \hat{g}_{ij}(\hat{\boldsymbol{\eta}},\hat{\boldsymbol{\gamma}})} = \boldsymbol{0}.
\end{align}
The detailed derivation can be found in Web Appendix A. There may be multiple roots to Equation (\ref{eq:MRW_sol2}).  We apply convex minimization from \cite{han2014multiply} to obtain $\hat{\boldsymbol{\rho}}$. 

\subsection{Consistency and Asymptotic Normality of MMR-GEE}
\label{ss:2.5}
In this section, we first demonstrate that the MMR-GEE estimator has the multiply robust property, i.e.,  $\hat{\boldsymbol{\beta}}_{MR}$ is a consistent estimator for $\boldsymbol{\beta}$ when both $\mathcal{P}_1$ and $\mathcal{P}_2$ contain a correctly specified PS model. 
We then establish the asymptotic distribution of $\hat{\boldsymbol{\beta}}_{MR}$. 
We consider the case where the number of clusters $M$ grows to infinity and the cluster sizes are bounded above. For notational simplicity, we consider the setting with fixed cluster size. The results can be generalized to varying cluster sizes by invoking the Lindeberg-Feller central limit theorem. In what follows, we use subscript asterisk to denote probability limits, $(\boldsymbol{\eta}_0, \boldsymbol{\gamma}_0)$ to denote the true parameters of the PS models, and $\boldsymbol{\beta}_0$ to denote the true parameters of the marginal mean model.

\subsubsection{Multiple Robustness of $\hat{\boldsymbol{\beta}}_{MR}$}
\label{ss:MR}

To prove the multiple robustness of $\hat{\boldsymbol{\beta}}_{MR}$, we first show that the multiply robust weights of Equation (\ref{model:MRW}) under which both $\mathcal{P}_1$ and $\mathcal{P}_2$ contain a correctly specified PS model are asymptotically equivalent to the multi-level inverse probability weights of equation (\ref{model:MW}) when the true correct models are known. This asymptotic equivalence can be established by building the connection between the multiply robust weights and another version of the empirical likelihood weights conditional on the observed sample assuming that the correct PS models are known, which we will derive below.  

Without loss of generality, let $\phi_{ij}^{1}(\boldsymbol{\eta}^{1})$ and $\lambda_{i}^{1}(\boldsymbol{\gamma}^{1})$, the first model in $\mathcal{P}_1$ and $\mathcal{P}_2$, respectively, be the correctly specified models. Furthermore, let $p_{ij}$ be the empirical probability of $(Y_{ij}, A_i, \boldsymbol{Z}_i, \boldsymbol{X}_{ij})$ conditional on $R_{ij}C_{i} = 1$ for $i =1,...,s, j = 1,...,m_i$. The estimator for $p_{ij}$, denoted by $\hat{p}_{ij}$, can be obtained by solving the empirical version of the constrained optimization problem of (\ref{model:MRW_CO}) using the same Lagrange multipliers method as in \ref{ss:2.4}. With some algebra manipulation, $\hat{w}_{ij}^{MR}$ can be expressed as (see Web Appendix B for derivation): 
\begin{equation}\label{eq:MR_1}
    \hat{w}_{ij}^{MR} = \hat{p}_{ij} \frac{ \chi^{1,1}\left(\hat{\boldsymbol{\eta}}^{1}, \hat{\boldsymbol{\gamma}}^{1}\right)}{\phi_{ij}^1(\hat{\boldsymbol{\eta}}^1)\lambda_{i}^1(\hat{\boldsymbol{\gamma}}^1)} =\frac{1+ o_p(1)}{\left(\sum_{i=1}^{M} n_i\right)\phi_{ij}^1(\boldsymbol{\eta}_{0})\lambda_{i}^1(\boldsymbol{\gamma}_{0})} , \hspace{0.5cm} i = 1, \ldots, s, j = 1, \ldots m_i.
\end{equation}


Plugging $\hat{w}_{ij}^{MR}$ back to the weighting matrix of (\ref{model:MRW}), we establish the relationship that:
\begin{align}\label{eq:MR_3}
    \hat{\boldsymbol{W}}_{i}^{MR} 
    = \text{diag}\left[R_{ij}C_{i}\hat{w}_{ij}^{MR}\right]_{j = 1, \ldots, n_i}  
    = \text{diag}\left[
    \frac{R_{ij}C_{i} + o_p(1)}{\left(\sum_{i=1}^{M} n_i\right)\phi_{ij}^1(\boldsymbol{\eta}_{0})\lambda_{i}^1(\boldsymbol{\gamma}_{0})} \right]_{j = 1, \ldots, n_i},    
\end{align}
which are asymptotically proportional to the weights of the correctly specified MIPW-GEE estimator of (\ref{model:MW}). As the number of clusters $M$ goes to infinity, we have:
\begin{align} \label{eq:21}
    \nonumber &\hspace{0.5cm}\left(\frac{\sum_{i=1}^{M}{n_i}}{M} \right)\sum_{i=1}^{M} \frac{\partial{\boldsymbol{\mu}_i(\boldsymbol{\beta}_0, A_i)}}{\partial{\boldsymbol{\beta}}}\boldsymbol{V}_{i}^{-1} \hat{\boldsymbol{W}}_{i}^{MR}   (\boldsymbol{Y}_i - \boldsymbol{\mu}_i(\boldsymbol{\beta}_0, A_i)) \\
    \nonumber &= \frac{1}{M}\sum_{i=1}^{M} \frac{\partial{\boldsymbol{\mu}_i(\boldsymbol{\beta}_0, A_i)}}{\partial{\boldsymbol{\beta}}}\boldsymbol{V}_{i}^{-1} \text{diag}\left[
    \frac{R_{ij}C_{i} + o_p(1)}{\phi_{ij}^1(\boldsymbol{\eta}_{0})\lambda_{i}^1(\boldsymbol{\gamma}_{0})}\right] (\boldsymbol{Y}_i - \boldsymbol{\mu}_i(\boldsymbol{\beta}_0, A_i)) \\
    &\stackrel{p}{\longrightarrow} E\left\{\frac{\partial{\boldsymbol{\mu}_i(\boldsymbol{\beta}_0, A_i)}}{\partial{\boldsymbol{\beta}}}\boldsymbol{V}_{i}^{-1} \text{diag}\left[
    \frac{R_{ij}C_{i}}{\phi_{ij}^1(\boldsymbol{\eta}_{0})\lambda_{i}^1(\boldsymbol{\gamma}_{0})}\right] (\boldsymbol{Y}_i - \boldsymbol{\mu}_i(\boldsymbol{\beta}_0, A_i))\right\} = 0,
\end{align}
which proves the consistency of $\hat{\boldsymbol{\beta}}_{MR}$. The results are summarized below: 
\begin{theorem}
When $\mathcal{P}_1$ contains a correct model for $\phi_{ij}\left(\boldsymbol{\eta}\right)$ and $\mathcal{P}_2$ contains a correct model for  $\lambda_i\left(\gamma\right)$, as $M \to \infty$, 
$\hat{\boldsymbol{\beta}}_{MR} \stackrel{P}{\to} \boldsymbol{\beta}_0$. 
\end{theorem}
The proof is provided in Web Appendix B. 

\subsubsection{Asymptotic Distribution}
\label{ss:AD}
We derive the asymptotic distribution of $\hat{\boldsymbol{\beta}}_{MR}$ by following the approach from Theorem 2 of \cite{han2014multiply} assuming that the correct models are known. Without loss of generality, let $\phi_{i j}^{1}\left(\boldsymbol{\eta}^{1}\right)$ and $\lambda_{i}^{1}\left(\boldsymbol{\gamma}^{1}\right)$ be the correctly specified models for $\phi_{i j}\left(\boldsymbol{\eta}\right)$ and $\lambda_{i}\left( \boldsymbol{\gamma}\right)$, respectively. 
The score functions of $\boldsymbol{\eta}^{1}$ and $\boldsymbol{\gamma}^{1}$, 
denoted by $\boldsymbol{S}_{1,i}(\boldsymbol{\eta}^{1})$ and $\boldsymbol{S}_{2,i}(\boldsymbol{\gamma}^{1})$, are:
\begin{align*}
    \boldsymbol{S}_{1,i}(\boldsymbol{\eta}^{1}) &=  \sum_{j=1}^{n_i}\frac{C_i[R_{ij}-\phi_{ij}^{1}\left(\boldsymbol{\eta}^{1}\right)]}{\phi_{ij}^{1}\left(\boldsymbol{\eta}^{1}\right)\left\{1-\phi_{ij}^{1}\left(\boldsymbol{\eta}^{1}\right)\right\}} \frac{\partial \phi_{ij}^{1}\left(\boldsymbol{\eta}^{1}\right)}{\partial \boldsymbol{\eta}^{1}},~
    \boldsymbol{S}_{2,i}(\boldsymbol{\gamma}^{1}) 
    &= \frac{C_i-\lambda_i^{1}\left(\boldsymbol{\gamma}^{1}\right)}{\lambda_i^{1}\left(\boldsymbol{\gamma}^{1}\right)\left\{1-\lambda_i^{1}\left(\boldsymbol{\gamma}^{1}\right)\right\}} \frac{\partial \lambda_i^{1}\left(\boldsymbol{\gamma}^{1}\right)}{\partial \boldsymbol{\gamma}^{1}}.
\end{align*}
Let
\begin{align}
\nonumber \boldsymbol{L} &= \mathbb{E}\left[  \text{diag} \left\{\frac{ \boldsymbol{g}_{i j}(\boldsymbol{\eta}_{*}, \boldsymbol{\gamma}_{*}) }{\phi_{i j}^{1}\left(\boldsymbol{\eta}_{0}\right) \lambda_{i}^{1}\left(\boldsymbol{\gamma}_{0}\right)}  \right\} \boldsymbol{U}_{i}\left(\boldsymbol{\beta}\right) \right], \\
\boldsymbol{G} &= \mathbbm{1}^{\mathrm{T}} \mathbb{E}\left[ \text{diag} \left\{ \frac{\boldsymbol{g}_{ij}\left(\boldsymbol{\eta}_{*}, \boldsymbol{\gamma}_{*}\right)^{\otimes 2}}{\phi_{ij}^{1}\left(\boldsymbol{\eta}_{0} \right) \lambda_{i}^{1}\left(\boldsymbol{\gamma}_{0} \right)} \right\}\right] \mathbbm{1}, \\ 
\nonumber \boldsymbol{Q}_{i}\left(\boldsymbol{\gamma}^{1}, \boldsymbol{\eta}^{1}\right) &= \text{diag} \left\{ \frac{R_{ij}C_{i}}{\phi_{i j}^{1}\left(\boldsymbol{\eta}^{1}\right) \lambda_{i}^{1}\left(\boldsymbol{\gamma}^{1}\right)}  \right\} \boldsymbol{U}_{i}\left(\boldsymbol{\beta}\right)  -\mathbf{L} \mathbf{G}^{-1}\mathbbm{1}^{\mathrm{T}} \text{diag} \left\{\frac{R_{ij}C_{i}-\phi_{ij}^{1}\left(\boldsymbol{\eta}^{1}\right)\lambda_{i}^{1}\left(\boldsymbol{\gamma}^{1}\right)}{\phi_{ij}^{1}\left(\boldsymbol{\eta}^{1}\right)\lambda_{i}^{1}\left(\boldsymbol{\gamma}^{1}\right)} \boldsymbol{g}_{ij}\left(\boldsymbol{\eta}_{*}, \boldsymbol{\gamma}_{*}\right)\right\} \mathbbm{1},
\end{align}
where $\mathbbm{1} = (1, 1, \ldots, 1)^{\mathrm{T}}$ and for any matrix $\boldsymbol{C}$, $\boldsymbol{C}^{\otimes 2}=\boldsymbol{C} \boldsymbol{C}^{\mathrm{T}}$.
Furthermore, write $\boldsymbol{S}_{1,i} = \boldsymbol{S}_{1,i}\left(\boldsymbol{\eta}_0\right)$, $\boldsymbol{S}_{2,i} = \boldsymbol{S}_{2,i}\left(\boldsymbol{\gamma}_0\right) $, and $\boldsymbol{Q}_i = \boldsymbol{Q}_i(\boldsymbol{\gamma}_0, \boldsymbol{\eta}_0)$. The following theorem gives the asymptotic distribution of $\hat{\boldsymbol{\beta}}_{MR}$. 
\begin{theorem}
When both $\mathcal{P}_1$ and $\mathcal{P}_2$ contain a correctly specified model for $\phi_{ij}\left(\boldsymbol{\eta}\right)$ and $\lambda_{i}\left(\boldsymbol{\gamma}\right)$, respectively, $\sqrt{M}(\hat{\boldsymbol{\beta}}_{MR} - \boldsymbol{\beta}_0)$ has an asymptotic normal distribution with mean $\boldsymbol{0}$ and variance var($\boldsymbol{Z}_i$), where
\begin{align*}
    \boldsymbol{Z}_i = \left[E\left\{\frac{\partial \boldsymbol{U}_i\left(\boldsymbol{\beta}_0\right)}{\partial \boldsymbol{\beta}}\right\}\right]^{-1}\left[\boldsymbol{Q}_i-\left\{\mathbb{E}\left(\boldsymbol{Q}_i \boldsymbol{S}_{1,i}^{\mathrm{T}}\right)\right\}\left\{\mathbb{E}\left(\boldsymbol{S}_{1,i}^{\otimes 2}\right)\right\}^{-1} \boldsymbol{S}_{1,i}-\left\{\mathbb{E}\left(\boldsymbol{Q}_i \boldsymbol{S}_{2,i}^{\mathrm{T}}\right)\right\}\left\{\mathbb{E}\left(\boldsymbol{S}_{2,i}^{\otimes 2}\right)\right\}^{-1} \boldsymbol{S}_{2,i}\right].
\end{align*}
\end{theorem}
See Web Appendix C for detailed derivation and proof.

The asymptotic variance of the MMR-GEE estimator requires the knowledge of correct PS models, which are usually unavailable to the investigators. Therefore, the asymptotic variance formula cannot easily be used to obtain the standard error estimates. We recommend using the non-parametric bootstrapping approach for inference \citep{davison1997bootstrap}. Two commonly used bootstrapping strategies for clustered data include ``clustered bootstrap'' and ``individual bootstrap. The ``clustered bootstrap'' approach samples $M$ clusters with replacement, with all individuals from the resampled clusters included in the bootstrap sample \citep{field2007bootstrapping}. The ``individual bootstrap'' approach samples $n_i$ individuals with replacement within the original clusters, keeping total sample size fixed across all bootstrap replicates \citep{roberts2004bootstrapping}. In the current setting where outcomes from an entire cluster may be missing, the ``individual bootstrap" approach applied to the clusters with observed data only does not account for the uncertainty in estimating the cluster-level PS model and may lead to underestimation of standard errors. Therefore, we recommend using the ``cluster bootstrap" approach.

\subsection{An EM Algorithm to Address Misclassification of Cluster-level Missingness Indicators}
\label{ss:2.6}


\label{ss:MIPW-GEE-EM}
The consistency of the proposed MMR-GEE estimator requires parameters of the correctly specified PS models to be consistently estimated, which can usually be achieved by, for example, moment-based or likelihood-based estimators. In the observed data, when no individual outcomes from a cluster is available, it is possible that outcome data from this cluster are missing by the cluster-level missingness process (that is, the true cluster-level missingness indicator  $C_i=0$); it is also possible that the cluster remains in the study, but all individual outcomes from this cluster are missing, especially when cluster size is small (that is, $C_i=1$, but $R_{ij}=0$ for all $j=1,\dots, n_i$). Let $C_i^O$ denote the observed cluster-level missingness indicator. In both cases, we observe $C_i^O=I(\sum_{j=1}^{n_i} R_{ij}>0)=0$, but $C_i$ can be either $0$ or $1$. Because consistent estimation of parameters in the PS models requires knowing $C_i$, potential  misclassification can occur if one naively assigns $C_i=C_i^O$. We summarize all possible patterns of $(C_i, C_i^{O})$, which include: 
\begin{enumerate}
    \item $C_i = 0, C_i^O = 0$: when a cluster drops out after randomization, all participants' outcomes in that cluster cannot be observed so $C_i^O$ is 0.
    \item $C_i = 1, C_i^O = 0$:  When the cluster does not drop out, we might still observe $C_i^{O} = 0$ if all individual outcomes within the cluster are missing. Such scenario is more likely to happen for small cluster sizes.
    \item $C_i = 1, C_i^O = 1$: When the cluster does not drop out, the observed cluster-level missingness indicator is the true cluster-level missingness indicator when at least one participant's outcome in cluster $i$ is observed. \\
\end{enumerate}    

The patterns are also summarized in Table S1 in Web Appendix D.1. Under pattern (2), the observed $C_i^O$ misclassifies the true $C_i$, leading to bias in the estimated parameters for the PS models. More specifically, 
suppose that the PS models are:
\begin{align}
    \label{eq:5}
    \text{logit} \left\{ \lambda_{i}(A_i,\boldsymbol{Z}_{i}; \boldsymbol{\gamma}) \right\} = \boldsymbol{Z}_{i}^* \boldsymbol{\gamma}, \hspace{0.3cm}
    \text{logit}\left\{\phi_{ij}(A_i, \boldsymbol{Z}_{i}, \boldsymbol{X}_{ij}\mid C_i = 1; \boldsymbol{\eta}) \right\}  = \boldsymbol{X}_{ij}^* \boldsymbol{\eta},
\end{align}
where $\boldsymbol{Z}_i^* = \left(\begin{array}{ccc}1 & A_i & \boldsymbol{Z}_{i}^\mathrm{T} \end{array}\right)$, $\boldsymbol{X}_{ij}^* = \left(\begin{array}{cccc}1 & A_i & \boldsymbol{Z}_{i}^\mathrm{T} & \boldsymbol{X}_{ij}^\mathrm{T}  \end{array}\right)$, $\boldsymbol{\gamma} = \left(\begin{array}{ccc} \gamma_{I} & \gamma_{A} & \boldsymbol{\gamma}_{\boldsymbol{Z}}^\mathrm{T} \end{array}\right)^\mathrm{T}$, and $\boldsymbol{\eta} = \left(\begin{array}{cccc} \eta_{I} & \eta_{A} & \boldsymbol{\eta}_{\boldsymbol{Z}}^\mathrm{T} & \boldsymbol{\eta}_{\boldsymbol{X}}^\mathrm{T} \end{array}\right)^\mathrm{T}$. Because $P(R_{ij} = 1|C_i = 1, A_i, \boldsymbol{Z}_{i}, \boldsymbol{X}_{ij}) \neq P(R_{ij} = 1|C_i^{O} = 1, A_i, \boldsymbol{Z}_{i}, \boldsymbol{X}_{ij})$ and $P(C_i = 1|A_i, \boldsymbol{Z}_{i}) \neq P(C_i^{O} = 1|A_i, \boldsymbol{Z}_{i})$, estimators of $\boldsymbol{\gamma}$ and $\boldsymbol{\eta}$ based on $C_i^O$ may be biased even if the PS models are correctly specified. 
To address this potential misclassification problem, we treat $C_i$ as partially observed data and propose an EM algorithm \citep{dempster1977maximum} to estimate the parameters in the PS models.

In the current setting, the ``complete'' data is $\{C_i, \boldsymbol{R}_i, \boldsymbol{Y}_i, A_i, \boldsymbol{Z}_i, \boldsymbol{X}_i\}_{i=1}^{M}$, which is denoted by ($\boldsymbol{C}, \boldsymbol{R}, \boldsymbol{Y}, \boldsymbol{A}, \boldsymbol{Z}, \boldsymbol{X}$) for simplicity of notation.
The complete data log likelihood is:
\begin{align}
\label{eq:EM_1}
    \nonumber & \hspace{0.3cm} \ell_c(\boldsymbol{C},  \boldsymbol{R}, \boldsymbol{Y}, \boldsymbol{A}, \boldsymbol{Z}, \boldsymbol{X}; \boldsymbol{\gamma}, \boldsymbol{\eta}) \\
    &= \nonumber \sum_{i=1}^{M} C_i \left\{ \log \left( \operatorname{expit}\left(\boldsymbol{Z}_{i}^{*} \gamma\right) \right)  +  \sum_{j=1}^{n_i} R_{ij} \log\left(\operatorname{expit}\left(\boldsymbol{X}_{i j}^{*} \boldsymbol{\eta}\right)\right)  + (1 - R_{ij})  \log\left(1 - \operatorname{expit}\left(\boldsymbol{X}_{i j}^{*} \boldsymbol{\eta}\right)  \right) \vphantom{\frac12}\right\} + \\
    & (1 - C_i) \log\left(1 - \operatorname{expit}\left(\boldsymbol{Z}_{i}^{*} \gamma\right)\right) + 
     C_i  \left( \log \left( \operatorname{expit}\left(\boldsymbol{Z}_{i}^{*} \gamma\right) \right) + \sum_{j=1}^{n_i} \log \left(1 - \operatorname{expit} \left(\boldsymbol{X}_{i j}^{*}  \boldsymbol{\eta} \right)  \right)\right). 
\end{align} 
The conditional expectation of the Expectation step (E-step) at iteration $\nu$ given the observed data $(\boldsymbol{C}^O, \boldsymbol{R}, \boldsymbol{Y}, \boldsymbol{A}, \boldsymbol{Z}, \boldsymbol{X})$ is:
\begin{align}
    \label{eq:EM_2}
    \nonumber & \hspace{0.3cm} Q\left(\boldsymbol{\gamma}, \boldsymbol{\eta}, \boldsymbol{\gamma}^{(\nu)}, \boldsymbol{\eta}^{(\nu)}, \boldsymbol{C}^O, \boldsymbol{R}, \boldsymbol{Y}, \boldsymbol{A}, \boldsymbol{Z}, \boldsymbol{X}\right) \\
    \nonumber &=\mathbb{E}_{\boldsymbol{\gamma}^{(\nu)}, \boldsymbol{\eta}^{(\nu)}}\left\{\ell(\boldsymbol{C},  \boldsymbol{R}, \boldsymbol{Y}, \boldsymbol{A}, \boldsymbol{Z}, \boldsymbol{X}; \boldsymbol{\gamma}, \boldsymbol{\eta}) \mid \boldsymbol{R}, \boldsymbol{Y}, \boldsymbol{A}, \boldsymbol{Z}, \boldsymbol{X}\right\} \\
    \nonumber &= \sum_{i=1}^{M} C_i^O \left\{ \log \left( \operatorname{expit}\left(\boldsymbol{Z}_{i}^{*} \gamma\right) \right)  +  \sum_{j=1}^{n_i} R_{ij} \log\left(\operatorname{expit}\left(\boldsymbol{X}_{i j}^{*} \boldsymbol{\eta}\right)\right) + (1 - R_{ij})  \log\left(1 - \operatorname{expit}\left(\boldsymbol{X}_{i j}^{*} \boldsymbol{\eta}\right)  \right) \vphantom{\frac12}\right\} + \\
    \nonumber &\hspace{1.3cm}  (1 - w_i^{(\nu)})(1 - C_i^O) \log\left(1 - \operatorname{expit}\left(\boldsymbol{Z}_{i}^{*} \gamma\right)\right) + \\
    &\hspace{1.3cm} w_i^{(\nu)}(1 - C_i^O)  \left( \log \left( \operatorname{expit}\left(\boldsymbol{Z}_{i}^{*} \gamma\right) \right) + \sum_{j=1}^{n_i} \log \left(1 - \operatorname{expit} \left(\boldsymbol{X}_{i j}^{*}  \boldsymbol{\eta} \right)  \right)\right), \\
    \nonumber &\text{where } w_{i}^{(\nu)} 
    = C_i^O + (1 - C_i^O) \times \hspace{0.3cm} \left\{ \frac{\prod_{j=1}^{n_i} \left( 1 -  \operatorname{expit}(\boldsymbol{X}_{ij}^* \boldsymbol{\eta}^{(\nu)})\right) \operatorname{expit}(\boldsymbol{Z}_{i}^* \boldsymbol{\gamma}^{(\nu)})}{1 - \operatorname{expit}(\boldsymbol{Z}_{i}^* \boldsymbol{\gamma}^{(\nu)}) \left[ 1 -  \prod_{j=1}^{n_i} \left(1 -  \operatorname{expit}(\boldsymbol{X}_{ij}^* \boldsymbol{\eta}^{(\nu)})\right) \right]} \right\}.
\end{align}
For the Maximization step (M-step), we recommend using the optimization software such as the \textit{Optim} function in R \citep{optimr, R} to maximize the complete data likelihood. 
Applying the E-step and M-step iteratively, the EM estimators $\hat{\boldsymbol{\gamma}}^{EM}$ and $\hat{\boldsymbol{\eta}}^{EM}$ can be obtained after the algorithm converges. When the PS models are correctly specified, $\hat{\boldsymbol{\gamma}}^{EM}$ and $\hat{\boldsymbol{\eta}}^{EM}$ would be consistent for the true parameters of the PS models despite misclassification of the cluster-level missingness indicators, i.e., $\hat{\boldsymbol{\eta}}^{EM} \stackrel{P}{\longrightarrow} \boldsymbol{\eta}_{0}$ and $\hat{\boldsymbol{\gamma}}^{EM} \stackrel{P}{\longrightarrow} \boldsymbol{\gamma}_{0}$. The detailed derivation of the complete data likelihood and E-step as well as pseudo code for the algorithm can be found in Web Appendix D.

\subsection{Extension to three-level CRTs}
\label{ss:2.7}

In three-level CRTs, study participants are nested in subclusters such as households or healthcare providers, and subclusters are nested in clusters such as regions or clinics. In this Section, we describe how the proposed methods can be modified to address the informative outcome missingness at both the subcluster and individual levels, as in the Pro-CCM study described in the Introduction \citep{ratovoson2022proactive}. 


Let $Y_{ijk}$ be the outcome and $\boldsymbol{X}_{ijk} = (X_{ijk}^{1},X_{ijk}^{2},\ldots)^\mathrm{T}$ be a vector of baseline covariates for participant $k = 1,...,n_{cs}$ from subcluster $j = 1,...,n_s$ in cluster $i = 1,...,n_c$. Here, the baseline covariates $\boldsymbol{X}_{ijk}$ can contain individual-, subcluster-, and cluster-level information. Similar to before, we consider a two-arm parallel CRTs using the same binary treatment indicator notation $A_i$; we also assume that all covariates are fully observed before randomization. 
For the multi-level missingness processes, $\boldsymbol{R}_{ij} = (R_{ij1}, R_{ij2}, \ldots)^\mathrm{T}$ is used to denote the vector of individual-level missingness indicator and $C_{ij}$ is used to denote the subcluster-level missingness indicator for outcomes $\boldsymbol{Y}_{ij} = (Y_{ij1}, Y_{ij2}, \ldots)^\mathrm{T}$. $R_{ijk} = 1$ when $Y_{ijk}$ is observed and $R_{ijk} = 0$ when $Y_{ijk}$ is missing. $C_{ij} = 0$ when all participants' outcomes in subcluster $j$ are missing and $C_{ij} = 1$ otherwise. Essentially, Table \ref{tab:1} represents the data structure of one cluster, and the data structure for three-level CRTs is the concatenation of all clusters. 

Under the multi-level missingness setting for three-level CRTs, the estimating equation for MIPW-GEE and MMR-GEE can be modified as follows:
\begin{equation}\label{model:MIPW-GEE-three-level}
    \sum_{i=1}^{n_c} \frac{\partial{\boldsymbol{\mu}_i(\boldsymbol{\beta}, A_i)}}{\partial{\boldsymbol{\beta}}}\boldsymbol{V}_{i}^{-1}\boldsymbol{W}_{i}(A_i, \boldsymbol{X}_{i};\boldsymbol{\gamma}, \boldsymbol{\beta})(\boldsymbol{Y}_i - \boldsymbol{\mu}_i(\boldsymbol{\beta}, A_i)) = 0.
\end{equation}
$\frac{\partial{\boldsymbol{\mu}_i(\boldsymbol{\beta}, A_i)}}{\partial{\boldsymbol{\beta}}}$ is the design matrix with $\boldsymbol{\mu}_{i}(\boldsymbol{\beta}, A_i) = \{\{\boldsymbol{\mu}_{ijk}(\boldsymbol{\beta}, A_i)\}_{k = 1}^{n_{cs}}\}_{j = 1}^{n_s}$. $\boldsymbol{V}_{i} = \boldsymbol{F}_{i}^{1/2}\boldsymbol{C}(\boldsymbol{\alpha})\boldsymbol{F}_{i}^{1/2}$ is the covariance matrix with $\boldsymbol{F}_{i} = \text{diag}\{\text{diag}[\text{var}(y_{ijk})]_{k = 1}^{n_{cs}}\}_{j = 1}^{n_s}$. Under three-level CRTs, exchangeable and block exchangeable correlation structures are common choices for the specification of $\boldsymbol{C}(\boldsymbol{\alpha})$. To account for informative missing subclusters, the multi-level weighting matrix takes the following form:
\begin{align}\label{model:MW-three-level}
    \boldsymbol{W}_{i}(A_i, \boldsymbol{X}_{i}; \boldsymbol{\gamma}, \boldsymbol{\eta})  
    &= \text{diag}\left\{\boldsymbol{W}_{ij}(A_i, \boldsymbol{X}_{ij}; \boldsymbol{\gamma}, \boldsymbol{\eta})\right\}_{j = 1}^{n_{s}}  \nonumber \\ 
    &= \text{diag}\left\{\text{diag}\left[\boldsymbol{W}_{ijk}(A_i, \boldsymbol{X}_{ijk}; \boldsymbol{\gamma}, \boldsymbol{\eta})\right]_{k = 1}^{n_{cs}} \right\}_{j = 1}^{n_{s}}  \\
    &= \text{diag}\left\{\text{diag}\left[\frac{R_{ijk}C_{ij}}{\phi_{ijk}(A_i, \boldsymbol{X}_{ijk} \mid C_{ij} = 1; \boldsymbol{\eta})\lambda_{ij}(A_i, \boldsymbol{X}_{ij}; \boldsymbol{\gamma}) }\right]_{k = 1}^{n_{cs}} \right\}_{j = 1}^{n_{s}} \nonumber,
\end{align}
where $\phi_{ijk}(A_i, \boldsymbol{X}_{ijk}|C_{ij} = 1; \boldsymbol{\eta})$ is the individual-level missingness process for $P(R_{ijk} = 1 |C_{ij} = 1, A_i, \boldsymbol{X}_{ij})$ and $\lambda_{ij}(A_i, \boldsymbol{X}_{ij}; \boldsymbol{\gamma})$ is the subcluster-level missingness process for $P(C_{ij} = 1|A_i, \boldsymbol{X}_{ij})$. The extension of the MMR-GEE estimator can be obtained by replacing the weighting matrix of Equation (\ref{model:MIPW-GEE-three-level}) by:
\begin{align}
\label{model:MRW-three-level}
    \boldsymbol{W}_{i}^{MR} = \text{diag}\{\text{diag}\left[R_{ijk}C_{ij}w_{ijk}^{MR}\right]_{k = 1}^{n_{cs}}\}_{j = 1}^{n_s}. 
\end{align} 
The estimation of $w_{ijk}^{MR}$ follows the same strategy as in Section \ref{ss:2.4}.

\section{A Simulation Study: Treatment of  Anemia in Malaria-Endemic Ghana}
\label{s:simulation}

We carried out a simulation study to investigate the finite-sample performance of our proposed MMR-GEE estimator and the operating characteristics of the EM algorithm under varying cluster sizes and model specifications for the PS. We designed the simulation study based on the Treatment of Iron Deficiency Anemia in Malaria-Endemic Ghana study, a CRT evaluating the use of iron supplements on reducing the incidence of malaria among infants and young children in Ghana, West Africa \citep{zlotkin2013effect}. The study randomized children aged 6 to 35 months by cluster, defined as a compound including 1 or more households. Seven hundred eighty clusters (967 children) in the treatment group received daily micronutrient powder (MNP) with iron; seven hundred seventy two clusters (991 children) in the control group received daily MNP without iron. At the end of the study, 12 clusters (25 children) in the iron group and 19 clusters (29 children) in the no iron group were lost to follow-up. The primary outcome was incidence of malaria over the course of the study. We focused on cumulative incidence of malaria during 1-month post-intervention period. Descriptive summaries of the outcome and baseline covariates extracted from \cite{zlotkin2013effect} are provided in Table S2 in Web Appendix E.1. 

\subsection{Data Generating Processes}
\label{ss:data_generation}
We generated data to match the summary statistics of Table S2. We treated the primary outcome $Y_{ij}$ (incidence of malaria per 100 child-year) as a continuous variable, which was generated as
\begin{align} \label{example:OM1}
    Y_{ij} = \beta_{I}^{*} + \beta_{A}^{*} A_i + \boldsymbol{Z}_{i}^{O}\boldsymbol{\beta}_{Z} +\boldsymbol{X}_{ij}^{O}\boldsymbol{\beta}_{X} +   A_i\boldsymbol{Z}_{i}^{O}\boldsymbol{\beta}_{AZ} +A_i\boldsymbol{X}_{ij}^{O}\boldsymbol{\beta}_{AX} +  \delta_{i} + \varepsilon_i,
\end{align}
where the treatment assignment ($A_i$) was simulated from the Bernoulli distribution with probability $p = 0.5$, $\boldsymbol{Z}_{i}^{O}$ contained \textit{household size}, \textit{household education}, and \textit{wealth}. $\boldsymbol{X}_{ij}^{O}$ contained \textit{age}, \textit{wasting z score}, and \textit{stunted growth z score}. $\delta_{i} \sim N(0,\sigma^2_\delta)$ was the cluster random intercept and $\varepsilon_i \sim N(0,\sigma^2_\varepsilon)$ was the residual standard error.

The true marginal mean model induced by marginalizing over $(\boldsymbol{Z}_{i}^{O},\boldsymbol{X}_{ij}^{O}, A_i\boldsymbol{Z}_{i}^{O},A_i\boldsymbol{X}_{ij}^{O})$, the random intercept, and residual standard error was  $\mathbb{E}[Y_{ij}|A_i] = \beta_{I} + \beta_{A} A_i$. The study concluded that daily use of MNP with iron did not increase the risk of malaria during the post-intervention period, so $(\beta_I^*, \beta_A^*, \boldsymbol{\beta}_Z^{\mathrm{T}}, \boldsymbol{\beta}_X^{\mathrm{T}}, \boldsymbol{\beta}_{AZ}^{\mathrm{T}}, \boldsymbol{\beta}_{AX}^{\mathrm{T}})$ were chosen such that the intercept ($\beta_I$) was 63.5 
and the marginal treatment effect ($\beta_A$) was 0.

Missing outcome data processes were induced through the following logistic models:
\begin{equation} \label{example:PS1}
    \text{logit}\left(\lambda_i\left(A_{i},\boldsymbol{Z}_{i}^{C}; \boldsymbol{\gamma}\right)\right) = \gamma_{I} + \gamma_{A} A_i + \boldsymbol{Z}_{i}^{C}\boldsymbol{\gamma}_{\boldsymbol{Z}} + A_{i}\boldsymbol{Z}_{i}^{C}\boldsymbol{\gamma}_{A\boldsymbol{Z}},
\end{equation}
\begin{equation} \label{example:PS2}
    \text{logit}\left(\phi_{ij}\left(A_{i}, Z_{i}^{I}, \boldsymbol{X}_{ij}^{I} \mid C_{i}=1 ; \boldsymbol{\eta}\right)\right) = \eta_{I} + \eta_{A} A_i + \eta_{Z}Z_{i}^{I} + \boldsymbol{X}_{ij}^{I}\boldsymbol{\eta}_{X} + A_i\boldsymbol{X}_{ij}^{I}\boldsymbol{\eta}_{AX}.
\end{equation}
$\boldsymbol{Z}_{i}^{C}$ included \textit{household size}, \textit{household education}, and \textit{wealth}. $Z_{i}^{I}$ included \textit{Complementary foods $\leq$ 6 mo}. $\boldsymbol{X}_{ij}^{I}$ included \textit{wasting z score} and \textit{stunted growth z score}. Around 2\% of the clusters and 2.8\% of the overall participants were missing in the original study. For illustration, we inflated the missingness by choosing $\boldsymbol{\gamma}$ and $\boldsymbol{\eta}$ such that 12\% of the clusters were missing and 30\% of the overall participants had missing outcomes. We considered a range of settings by varying the marginal effect parameter $\beta_A=(0, 1.5)$, cluster sizes, and ICC values. A detailed description of parameter values is provided in Web Appendix E.2. 


\subsection{Analysis Approaches}

To demonstrate the importance of correcting potential bias due to informative multi-level missing outcomes, we compared the following four approaches. First, we carried out an unweighted CC-GEE analysis based on Model (\ref{model:GEE}). Second, we applied the IPW-GEE method based on Model (\ref{model:IPW-GEE}), where the PS was estimated by the unconditional logistics model with the same functional form as Model (\ref{example:PS2}) but ignored cluster-level missingness:
\begin{equation} \label{example:PS3}
     \text{logit}\left(\pi_{ij}\left(A_{i}, Z_{i}^{I}, \boldsymbol{X}_{ij}^{I}; \boldsymbol{\theta}\right)\right) = \theta_{I} + \theta_{A} A_i + \theta_{Z} Z_{i}^{I}+ \boldsymbol{X}_{ij}^{I}\boldsymbol{\theta}_{X} +  A_i\boldsymbol{X}_{ij}^{I}\boldsymbol{\theta}_{AX}.
\end{equation}
Third, we employed the MIPW-GEE method, where both the cluster- and individual-level PS models were correctly specified. To estimate the parameters in the PS models, we fitted the standard logistic regression models based on $C_i^O$ (denoted as MIPW-GEE-no-EM) and also applied the EM algorithm (denoted as MIPW-GEE-EM). Lastly, we implemented our proposed MMR-GEE estimator by specifying $\mathcal{P}_1 = \{\lambda_{i}^{k}(\boldsymbol{\gamma}^{k}), k = 1, 2\}$ and $\mathcal{P}_2 = \{\phi_{ij}^{\ell}(\boldsymbol{\eta}^{\ell}), \ell = 1, 2\}$ with all parameters estimated by the EM algorithm. Both $\mathcal{P}_1$ and $\mathcal{P}_2$ contained one correctly specified and one misspecified models. See Web Appendix E.3 for details of the PS model specification.
For all approaches, we adopted an exchangeable working correlation structure and used the ``cluster bootstrap'' method for obtaining standard error estimates.


\subsection{Simulation Results}

\begin{figure}
    \includegraphics[scale = 0.7]{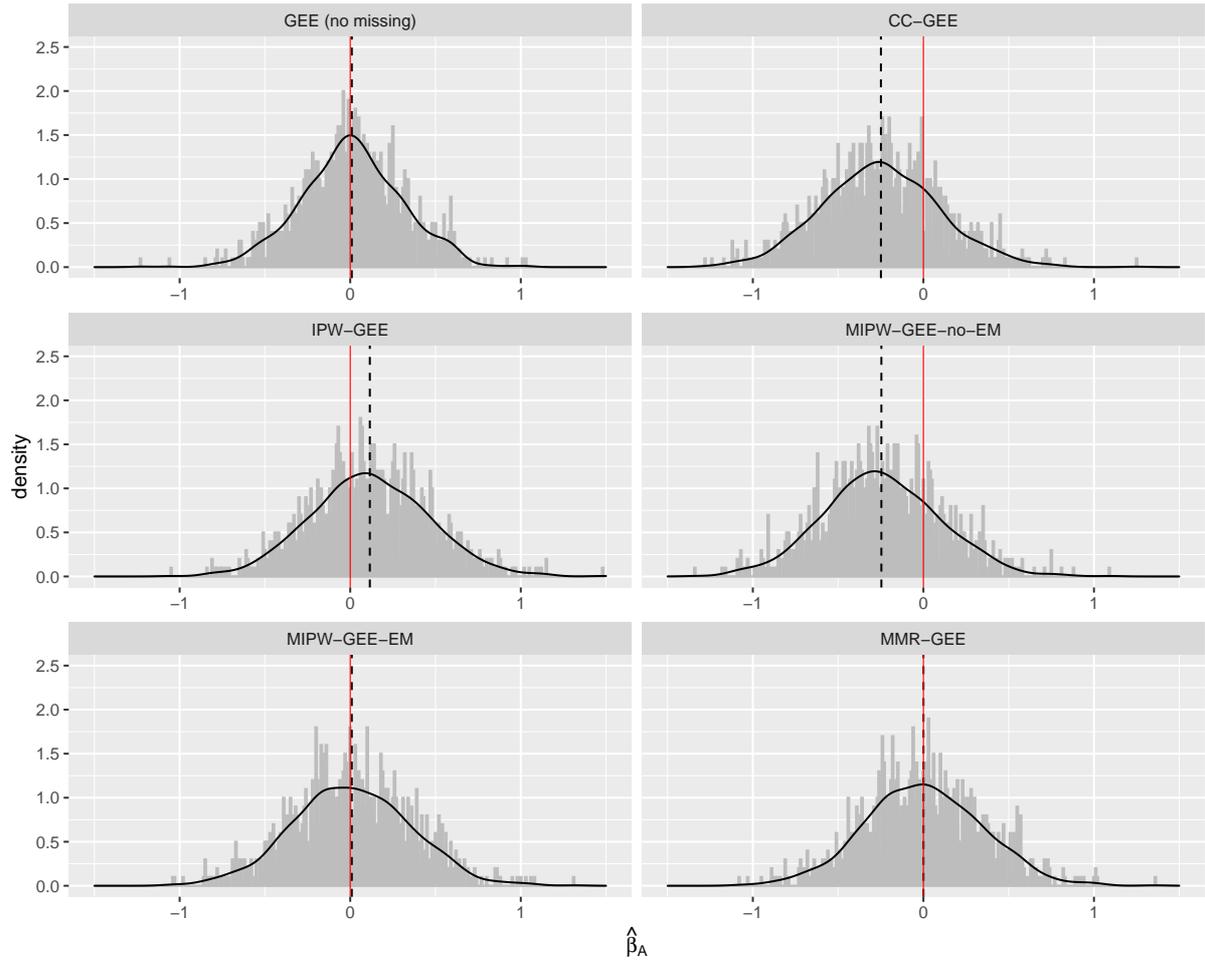}
    \caption{Empirical distribution of the estimated marginal treatment effect ($\hat{\beta}_A$) based on 1,000 replicates under $n_i \sim DU(1,5), M = 1552, ICC = 0.08$. The red line denotes the true marginal treatment effect (0). The black dotted line denotes the empirical mean of the estimated marginal treatment effect.}
    \label{fig:1}
\end{figure}

Figure \ref{fig:1} presents the empirical distribution of the estimated marginal treatment effect under the null (i.e., $\beta_A=0$) with $M = 1552$, $n_i\sim DU(1,5)$, $ICC = 0.0804$. The red line denotes the true marginal effect and the black dotted line is the empirical mean across 1,000 effect estimates. Overall, the proposed MMR-GEE  and MIPW-GEE-EM estimators led to estimates that were centered at the true $\beta_A$. Because CC-GEE ignored informative missing data and IPW-GEE failed to account for cluster-level missingness, they both resulted in biased estimates. MIPW-GEE-no-EM attempted to adjust for the multi-level missingness processes. However, without using the EM algorithm to correct the misclassfication in $C_i$, estimates from MIPW-GEE-no-EM led to bias. On the other hand, MIPW-GEE-EM appropriately accounted for the misclassification in $C_i$ and the bias disappeared. The corresponding figures under other settings are provided in Web Appendix E.4.

Table \ref{tab:3} summarizes the empirical bias, empirical standard errors, the average standard error estimates, and coverage probability. Under all scenarios, CC-GEE and IPW-GEE provided biased estimates of the marginal treatment effect (bias ranging from -0.262 to 0.599 for CC-GEE and 0.105 to 0.330 for IPW-GEE), whereas the average of estimates from MIPW-GEE-EM and MMR-GEE were very close to the true parameter values. The results for MIPW-GEE-no-EM were biased for small cluster sizes but the bias vanished as $n_i$ increased to $DU(30,50)$. Under all scenarios, the average of the bootstrapping-based standard errors was close to the empirical standard error. When the marginal treatment effect was consistently estimated (i.e., MMR-GEE and MIPW-GEE-EM, and MIPW-GEE-no-EM for $n_i \sim DU(30,50)$), the percentage of 95\% confidence interval that covered the true parameter values was close to 95\%. On the other hand, the empirical coverage associated with CC-GEE and IPW-GEE can be substantially lower than the nominal level (e.g., $<$40\% for CC-GEE and $<$70\% for IPW-GEE).

\begin{table}
\caption{Empirical bias, empirical standard errors, mean of estimated bootstrapping-based standard errors using the ``cluster bootstrap'' method, and coverage probability based on 1,000 replicates and 100 bootstrapping resamples. The coverage probability is the percentage of true $\beta_A$ contained in the 95\% confidence interval constructed from the bootstrapping-based standard errors.}
\label{tab:3}
 \centering
\begin{tabular*}{\textwidth}{@{}l@{\extracolsep{\fill}}*{4}{p{0.1\textwidth}}}
\hline
\hline
& & & \multicolumn{2}{c}{cluster bootstrap} \\
\cline{4-5} \\
[-10pt]
& \multicolumn{1}{c}{emp. bias} & \multicolumn{1}{c}{emp. SE}  & \multicolumn{1}{c}{est. SE} & \multicolumn{1}{c}{cov prob. (\%)} \\
\hline
\multicolumn{5}{l}{\underline{Data under null ($\beta_A = 0$)}}  \\
\multicolumn{5}{c}{\underline{$M = 1552, n_i \sim DU(1,5), ICC = 0.0804, P(C_i = C_i^O) = 0.87$}} \\
GEE (no missing) & 0.009 & 0.296 & 0.293 & 94.4 \\
CC-GEE & -0.249 & 0.345 & 0.350 & 89.3 \\
IPW-GEE & 0.115 & 0.345 & 0.349 & 93.4 \\
MIPW-GEE-no-EM & -0.246 & 0.343 &  0.346 & 88.8  \\
MIPW-GEE-EM & 0.009 & 0.344 & 0.348 & 94.9  \\
MMR-GEE & $<0.001$ & 0.343 & 0.346 & 94.5 \\
\multicolumn{5}{c}{\underline{$M = 1552, n_i = 3, ICC = 0.0804, P(C_i = C_i^O) = 0.99$}} \\
GEE (no missing) & -0.010 & 0.205 & 0.202 & 94.0 \\
CC-GEE & -0.262 & 0.237 & 0.237 & 80.0 \\
IPW-GEE & 0.105 & 0.237 & 0.236 & 90.4   \\
MIPW-GEE-no-EM & -0.025 & 0.239 & 0.235 & 94.1  \\
MIPW-GEE-EM & -0.007 & 0.240 & 0.236 & 94.0 \\
MMR-GEE & -0.009 & 0.239 & 0.234 & 94.1 \\
\multicolumn{5}{c}{\underline{$M = 300, n_i \sim DU(30,50), ICC = 0.2, P(C_i = C_i^O) = 1$}} \\
GEE (no missing) & 0.007 & 0.266 & 0.282 & 94.6 \\
CC-GEE &  -0.247 & 0.291 & 0.308 & 86.6 \\
IPW-GEE & 0.121 & 0.292 & 0.309 & 92.8 \\
MIPW-GEE-no-EM & 0.009 & 0.282 & 0.299 & 93.9 \\
MIPW-GEE-EM & 0.008 & 0.281 & 0.299 & 93.9 \\
MMR-GEE & 0.013 & 0.282 & 0.299 & 94.1 \\
\hline
\multicolumn{5}{l}{\underline{Data under alternative ($\beta_A = 1.5$)}}  \\
\multicolumn{5}{c}{\underline{$M = 1552, n_i \sim DU(1,4), ICC = 0.0804, P(C_i = C_i^O) = 0.94$}}  \\
Full Data & 0.001 & 0.179 & 0.188 & 95.6\\
Complete Cases & 0.483 & 0.204 & 0.211 & 37.8 \\
IPW-GEE & 0.270 & 0.201 & 0.209 & 73.8 \\
MIPW-GEE-no-EM & 0.072 & 0.211 & 0.216 & 93.0 \\
MIPW-GEE-EM & 0.002 & 0.212 & 0.217 & 95.1 \\
MMR-GEE & 0.009 & 0.212 & 0.214 & 95.1 \\
\multicolumn{5}{c}{\underline{$M = 1552, n_i = 3, ICC = 0.2, P(C_i = C_i^O) = 0.99$}}\\
Full Data & -0.002 & 0.174 & 0.177 & 94.4 \\
Complete Cases & 0.511 & 0.198 & 0.197 & 27.0 \\
IPW-GEE & 0.279 & 0.195 & 0.196 & 69.9 \\
MIPW-GEE-no-EM & 0.034 & 0.202 & 0.200 & 93.5 \\
MIPW-GEE-EM & $<0.001$ & 0.202 & 0.200 & 94.8 \\
MMR-GEE & 0.006 & 0.196 & 0.196 & 95.1 \\
\multicolumn{5}{c}{\underline{$M = 300, n_i \sim DU(30,50), ICC = 0.2, P(C_i = C_i^O) = 1$}} \\
Full Data & 0.009 & 0.308 & 0.314 & 95.1 \\
Complete Cases & 0.599 & 0.339 & 0.343 & 56.9 \\
IPW-GEE & 0.330 & 0.335 & 0.341 & 85.2 \\
MIPW-GEE-no-EM & 0.013 & 0.327 & 0.328 & 94.9 \\
MIPW-GEE-EM & 0.013 & 0.327 &  0.328 & 94.9\\
MMR-GEE & 0.011 & 0.323 & 0.329 & 95.3  \\
\hline
\end{tabular*}
\end{table}

We further compared strategies for estimating the parameters in the correctly specified PS models with and without using the EM algorithm and presented the empirical absolute bias in Table \ref{tab:4}. Under all scenarios, $\boldsymbol{\hat{\gamma}}^{EM}$ and $\boldsymbol{\hat{\eta}}^{EM}$ were centered at the true values, whereas the parameters estimated by the standard logistic regression using the observed $C_i^O$ can be substantially biased for $n_i \sim DU(1,4)$, $n_i \sim DU(1,5)$, and $n_i = 3$ but the bias vanished for $n_i \sim DU(30,50)$. As mentioned in Section \ref{ss:MIPW-GEE-EM}, the misclassification in $C_i$ due to the missingness of all individual outcomes in the cluster is more likely to happen for small $n_i$. As cluster size increases, the misclassification in $C_i$ becomes less probable because the probability of all individual outcomes within a cluster being missing, $\prod_{j=1}^{n_i} P(R_{ij} = 0| A_i, \boldsymbol{Z}_i, \boldsymbol{X}_{ij})$, would be very small. Therefore, large cluster size obviates the need to apply the EM algorithm. Also of note is that even when the EM estimator is not needed, applying the EM algorithm does not impact the estimated parameters as we can see that the MIPW-GEE-EM estimator performed almost identically to the MIPW-GEE-no-EM estimator for $n_i \sim DU(30,50)$.

\begin{table}[h!]
\caption{Empirical absolute bias of the estimated parameters in the PS models based on 1,000 replicates. All parameters were estimated from the correctly specified cluster- and individual-level PS models but with different estimation approaches: one applied the proposed EM algorithm, and the other fitted the standard logistic regression based on $C_i^O$}
\label{tab:4}
\begin{center}
\begin{tabular*}{\textwidth}{@{\extracolsep{\fill}}*{7}{c}}
\hline
\hline
&  with EM & w.o. EM & with EM & w.o. EM & with EM & w.o. EM \\ \hline 
\multirow{2}{*}{\underline{$\beta_A = 0$}} & \multicolumn{2}{c}{$n_i \sim DU(1,5)$} & \multicolumn{2}{c}{$n_i = 3$} & \multicolumn{2}{c}{$n_i \sim DU(30,50)$} \\
\cline{2-3} \cline{4-5} \cline{6-7} \\
[-10pt]
$\gamma_{I}$ & 0.01  & 0.41 & 0.01 & $<0.01$ & 0.02 & 0.02 \\
$\gamma_{A}$ &  0.03 & 0.03 & 0.02 & 0.08 & 0.03 & 0.03  \\
$\gamma_{Z_1}$ & $<0.01$ & 0.09  & $<0.01$ & $<0.01$ & 0.01 & 0.01 \\
$\gamma_{Z_2}$ & 0.02 & 0.14 & 0.01 & $<0.01$ & 0.03 & 0.03 \\
$\gamma_{Z_3}$ & $<0.01$ & 0.10 & $<0.01$ & $<0.01$ & 0.01 & 0.01 \\
$\gamma_{AZ_1}$ & $<0.01$ & 0.25 & $<0.01$ & 0.08 & 0.03 & 0.03 \\
$\gamma_{AZ_2}$ & 0.04 & 0.41 & 0.03 & 0.10 & 0.04 & 0.04 \\
$\gamma_{AZ_3}$ & $<0.01$ & 0.29 & 0.02 & 0.01 & 0.07 & 0.03 \\
$\eta_{I}$ & $<0.01$ & 1.18 & 0.01 & 0.04 & $<0.01$ & $<0.01$ \\
$\eta_A$ & 0.01 & 0.23 & $<0.01$ & 0.04 & $<0.01$ & $<0.01$ \\
$\eta_{X_4}$ & $<0.01$ & $<0.01$ & $<0.01$ & $<0.01$ & $<0.01$ & $<0.01$ \\
$\eta_{X_5}$ & $<0.01$ & 0.01 & $<0.01$ & $<0.01$ & $<0.01$ & $<0.01$ \\
$\eta_{Z_4}$ & 0.01 & 0.04 & 0.01 & 0.01 & $<0.01$ & $<0.01$ \\
$\eta_{AX_4}$ & $<0.01$ & $<0.01$ & $<0.01$ & $<0.01$ & $<0.01$ & $<0.01$ \\
$\eta_{AX_5}$ & $<0.01$ & 0.01 & $<0.01$ & $<0.01$ & $<0.01$ & $<0.01$ \\
\hline
\multirow{2}{*}{\underline{$\beta_A = 1.5$}} & \multicolumn{2}{c}{$n_i \sim DU(1,4)$} & \multicolumn{2}{c}{$n_i = 3$} & \multicolumn{2}{c}{$n_i \sim DU(30,50)$} \\
\cline{2-3} \cline{4-5} \cline{6-7} \\
[-10pt]
$\gamma_{I}$ & 0.04 & 0.36 & 0.02 & 0.01 & 0.06 & 0.06 \\
$\gamma_{A}$ &  0.01 & 0.40 & $<0.01$ & 0.20 & 0.03 & 0.03 \\
$\gamma_{Z_3}$ & 0.01 & 0.08  & $<0.01$ & 0.02 & 0.03 & 0.03 \\
$\gamma_{Z_4}$ & 0.01 & 0.09 & $<0.01$ & 0.01 & 0.02 & 0.02 \\
$\gamma_{AZ_3}$ & 0.01 & 0.03 & $<0.01$ & 0.03 & 0.04 & 0.04 \\
$\gamma_{AZ_4}$ & $<0.01$ & 0.13 & $<0.01$ & 0.07 & 0.01 & 0.01 \\
$\eta_{I}$ & $<0.01$ & $<0.01$ & $<0.01$ & $<0.01$ & $<0.01$ & $<0.01$ \\
$\eta_A$ & $<0.01$ & 0.01 & 0.01 & 0.01 & 0.01 & 0.01 \\
$\eta_{Z_3}$ & $<0.01$ & $<0.01$ & $<0.01$ & $<0.01$ & $<0.01$ & $<0.01$ \\
$\eta_{X_1}$ & $<0.01$ & $<0.01$ & $<0.01$ & $<0.01$ & $<0.01$ & $<0.01$ \\
$\eta_{X_2}$ & $<0.01$ & $<0.01$ & $<0.01$ & $<0.01$ & $<0.01$ & $<0.01$ \\
$\eta_{X_3}$ & 0.01 & 0.01 & 0.01 & 0.01 & $<0.01$ & $<0.01$\\
$\eta_{X_4}$ & $<0.01$ & $<0.01$ & $<0.01$ & $<0.01$ & $<0.01$ & $<0.01$ \\
$\eta_{AZ_3}$ & $<0.01$ & $<0.01$ & $<0.01$ & $<0.01$ & $<0.01$ & $<0.01$ \\
$\eta_{AX_1}$ & $<0.01$ & $<0.01$ & 0.01 & 0.01 & $<0.01$ & $<0.01$\\
$\eta_{AX_2}$ & $<0.01$ & $<0.01$ & 0.01 & $<0.01$ & $<0.01$ & $<0.01$ \\
$\eta_{AX_3}$ & 0.01 & 0.01 & 0.01 & 0.01 & $<0.01$ & $<0.01$ \\
$\eta_{AX_4}$ & $<0.01$ & $<0.01$ & $<0.01$ & $<0.01$ & $<0.01$ & $<0.01$ \\
\hline
\end{tabular*}
\end{center}
\end{table}

\newpage

\section{Application}
\label{s:applications}

We illustrate our proposed methods using the Pro-CCM study described in Introduction and 
Section \ref{ss:2.7}. \cite{ratovoson2022proactive} investigated the efficacy of  the pro-CCM intervention in reducing the prevalence of malaria in the Mananjary district of Madagascar. Twenty-two clusters (i.e. fokontany) were randomized with a 1:1 ratio to pro-CCM (treatment) or iCCM (control). Study participants were nested in households, which were nested in each fokontany. The disease status of each study participant was assessed using the rapid diagnostic tests (RDTs) at baseline and at endline. Here we focus on the individual-level RDT result at endline (RDT = 1 if positive, 0 if negative). The dataset consists of 29,683 participants 
with seven individual-level baseline covariates (male indicator $X_{ijk, 1}$, age $X_{ijk, 2}$, primary school indicator $X_{ijk, 3}$, secondary school indicator $X_{ijk, 4}$, high level school indicator $X_{ijk, 5}$, sleep in mosquito nets indicator $X_{ijk, 6}$, sleep in the yard indicator $X_{ijk, 7}$) and four household-level baseline covariates (household size $Z_{ij, 1}$, \% of male $Z_{ij, 2}$, highest education level $Z_{ij, 3}$, indoor residual spraying indicator $Z_{ij, 4}$).

The overall missingness of the individual-level outcome at endline was 31\%, corresponding to 22.3\% of missing households. Results based on a mixed effects model adjusting for socio-demographic characteristics suggested no statistical differences in RDT positivity at endline between participants in the intervention and control arm (OR = 0.71; 95\% CI: 0.36–1.43) \citep{ratovoson2022proactive}. We reanalyzed this dataset using the GEE approaches to estimate the marginal treatment effect while assuming outcomes had covariate-dependent missingness. We applied backward step-wise procedure based on the AIC to select covariates for the PS, which yielded the following models:
\begin{align} 
    \label{app:PS1}
    \operatorname{logit}\left(\pi_{i jk}; \boldsymbol{\theta}\right)=& \theta_{I}+\theta_{A} A_i+ \sum_{q \in\{1,2,3,4\}} \theta_{X}^{(q)} X_{ijk, q}+ \sum_{q \in\{1,2,3,4\}} \theta_{Z}^{(q)} Z_{ij, q} \\
    &+A_i \sum_{q \in\{1,2,3,4\}} \theta_{A X}^{(q)} X_{ijk, q}+ A_i \sum_{q \in\{1,2,4\}} \theta_{A Z }^{(q)} Z_{ij, q} \nonumber \\
    \label{app:PS2_1}
\operatorname{logit}\left(\lambda_{i j}; \boldsymbol{\gamma}\right)=& \gamma_{I}+\gamma_{A} A_i+ \sum_{q \in\{1,3,4\}} \gamma_{Z}^{(q)} Z_{ij, q} + A_i \sum_{q \in\{4\}} \gamma_{A Z }^{(q)} Z_{ij, q} \\ \label{app:PS2_2}
    \operatorname{logit}\left(\phi_{i jk}; \boldsymbol{\eta}\right)=& \eta_{I}+\eta_{A} A_i+ \sum_{q \in\{1,2,3,4\}} \eta_{X}^{(q)} X_{ijk, q} +\sum_{q \in\{1,2,3,4\}} \eta_{Z}^{(q)} Z_{ij, q} A_i \sum_{q \in\{1,2,3\}} \eta_{A Z }^{(q)} Z_{ij, q}
\end{align}
The model fitting results are provided in Table S3 in Web Appendix F. We carried out the following four analyses: CC-GEE based on participants with observed RDT test results at endline, IPW-GEE using Model (\ref{app:PS1}) for the PS, MIPW-GEE using Models (\ref{app:PS2_1}) and (\ref{app:PS2_2}) for the subcluster- and individual-level PS, and MMR-GEE with two sets of PS models $\mathcal{P}_1 = \{\lambda_{ij}^{k}(\boldsymbol{\gamma}^{k}), k = 1, 2\}$ and $\mathcal{P}_2 = \{\phi_{ijk}^{\ell}(\boldsymbol{\eta}^{\ell}), \ell = 1, 2\}$. $\mathcal{P}_1$ contained Model (\ref{app:PS2_1}) and another model that included ($A$, $A Z_1$, $AZ_3$); $\mathcal{P}_2$ contained Model (\ref{app:PS2_2}) and another model that included $(A$, $X_2$, $(X_2)^2$, $X_3$, $X_5$, $Z_4$, $A Z_4$). The parameters in the PS models for MIPW-GEE and MMR-GEE were estimated with the EM algorithm proposed in Section \ref{ss:2.6}.

CC-GEE yielded marginal effect estimate ($\widehat{OR}$ = 0.76, 95\% CI: 0.46-1.23) that was similar to the original finding. Approaches that incorporate potentially informative missing outcomes led to effect estimate slightly closer to the null (IPW-GEE $\widehat{OR}$ = 0.81, 95\% CI, 0.49-1.33; MIPW-GEE $\widehat{OR}$ = 0.84, 95\% CI, 0.51-1.38; MMR-GEE $\widehat{OR}$ = 0.82, 95\% CI, 0.50-1.34). Nevertheless, the conclusion remained the same as confidence intervals from all approaches included the null. The IPW-GEE estimator without explicitly modeling the subcluster-level missingness yielded effect estimates similar to our proposed MIPW-GEE and MMR-GEE estimators. Such similarity suggested that subclusters may be missing completely at random. Indeed, even though 22.3 \% of households were missing, the estimated probability for the subcluster-level missingness were all close to 1 (i.e., mean of $P(C_{ij}=1|A_i, \boldsymbol{Z}_i; \hat{\boldsymbol{\gamma}}^{EM})$ = 0.99 with range 0.96-1.00). While in this particular example, all approaches led to the same conclusions, the availability of proposed methods permit the assessment of the impact of potentially informative missingness on effect estimates under a range of assumptions about the outcome missingness mechanisms at multiple levels \citep{little2012prevention}.

\section{Discussion}
\label{s:discussion}
Drawing upon the empirical likelihood theory, this paper proposed a new estimation procedure and inference method for estimating the marginal treatment effect in CRTs with multi-level missing outcomes that guards against the partial misspecification of the PS models. To handle informative missingness at both the cluster and the individual level, we derived the multi-level inverse probability weights and applied the EM algorithm to correct potential misclassification of the cluster-level missingness indicators. The proposed MMR-GEE estimator allows analysts to specify multiple sets of PS models and leads to consistent treatment effect estimates provided that one model in $\mathcal{P}_1$ and one model in $\mathcal{P}_2$ are correctly specified and the parameters in the PS models are consistently estimated.

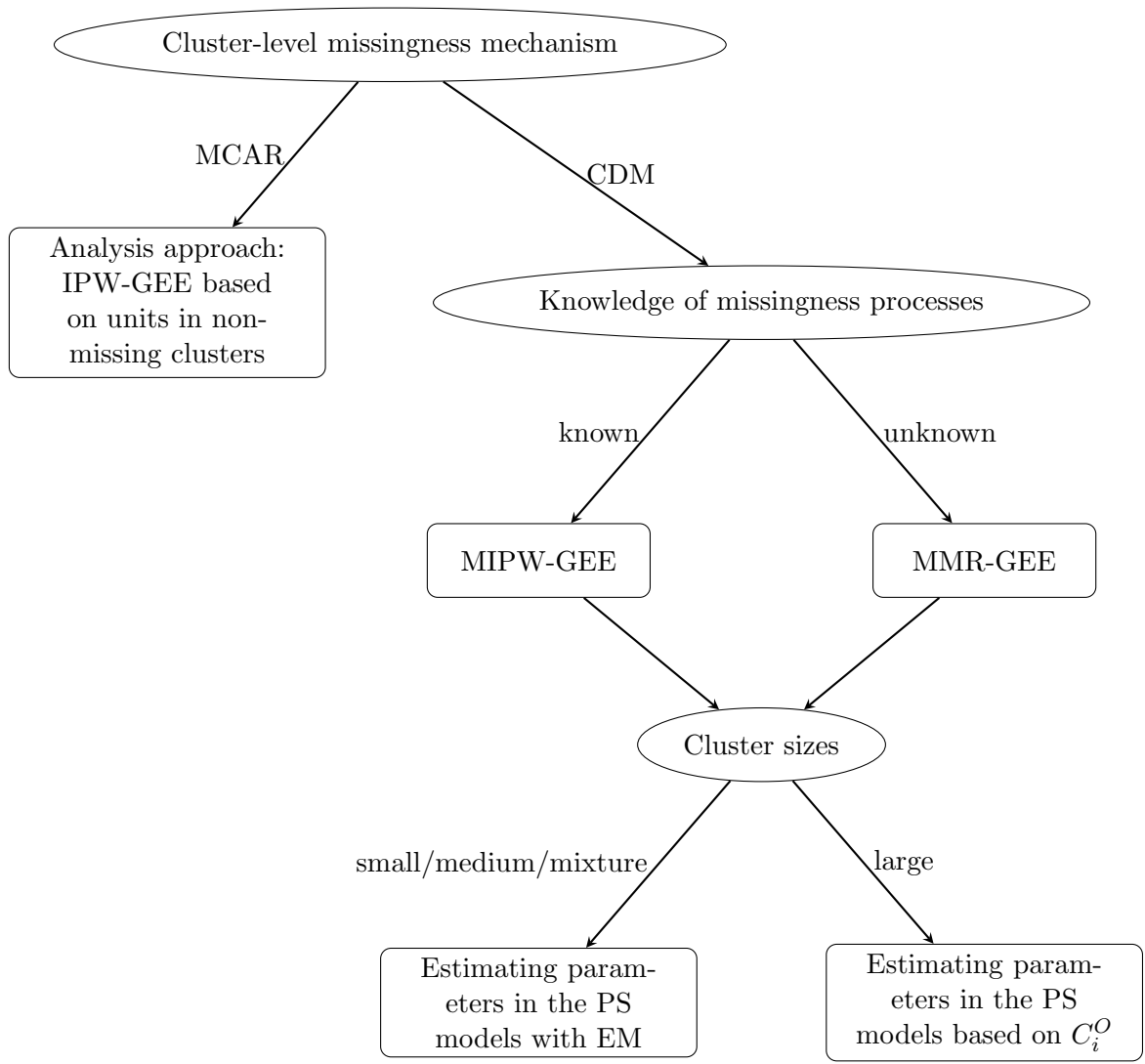
\begin{figure}
    \centering
    \tikzstyle{startstop} = [ellipse, minimum width=3cm, minimum height=1cm,text centered, draw=black]
    \tikzstyle{decision} = [rectangle, rounded corners, minimum width=3cm, minimum height=1cm, text centered, draw=black]
    \tikzstyle{arrow} = [thick,->,>=stealth]
    \begin{tikzpicture}[%
    >=stealth, 
    node distance=2cm,
    on grid,
    auto]
    \node (start) [startstop] {Cluster-level missingness mechanism};
    \node (IPW-GEE) [decision, below of=start, yshift=-1.5cm, xshift =-3cm, text width = 4cm] {Analysis approach: IPW-GEE based on units in non-missing clusters};
    \node (mis-proc) [startstop, below of=start, yshift=-1.5cm, xshift =5cm] {Knowledge of missingness processes};
    \draw [arrow] (start) -- node[anchor=east] {MCAR} (IPW-GEE);
    \draw [arrow] (start) -- node[anchor=west] {CDM} (mis-proc);
    \node (MIPW-GEE) [decision, below of=mis-proc, yshift=-1.5cm, xshift =-3cm] {MIPW-GEE};
    \node (MMR-GEE) [decision, below of=mis-proc, yshift=-1.5cm, xshift =3cm] {MMR-GEE};
    \draw [arrow] (mis-proc) -- node[anchor=east] {known} (MIPW-GEE);
    \draw [arrow] (mis-proc) -- node[anchor=west] {unknown} (MMR-GEE);
    \node (size) [startstop, below of=mis-proc, yshift=-4cm] {Cluster sizes};
    \draw [arrow] (MIPW-GEE) -- (size);
    \draw [arrow] (MMR-GEE) -- (size);
    \node (EM) [decision, below of=size, yshift=-1.5cm, xshift =-3cm, text width = 4cm] {Estimating parameters in the PS models with EM};
    \node (no-EM) [decision, below of=size, yshift=-1.5cm, xshift =3cm, text width = 4cm] {Estimating parameters in the PS models based on $C_i^O$ };
    \draw [arrow] (size) -- node[anchor=east] {small/medium/mixture} (EM);
    \draw [arrow] (size) -- node[anchor=west] {large} (no-EM);
    \end{tikzpicture}
    \caption{Modeling and estimation strategy under various missingness mechanisms and scenarios.}
    \label{fig:2}
\end{figure}

We investigated several analysis strategies in the presence of missing outcome data at multiple levels and created a flowchart to help with making choices among these approaches (see Figure \ref{fig:2}). First, our proposed approach is targeted towards the multi-level CDM setting. If one believes that clusters are missing completely at random, it may be suffice to apply the IPW-GEE method to incorporate individual-level informative missing outcome. 
Second, although MMR-GEE provides the flexibility to specify multiple sets of PS models, analysts can apply the MIPW-GEE method if they have substantial knowledge about the true multi-level missingness processes. Finally, the goal of EM algorithm is to address the challenge in estimating parameters of the PS models due to misclassification in $C_i$, which is more likely to happen for small cluster sizes. When cluster sizes are large, the probability of all individual outcomes within a cluster being missing becomes very small and this type of misclassification in $C_i$ is nearly impossible to happen, so applying the EM algorithm is unnecessary. Nonetheless we recommend using the EM algorithm when clusters contain mixture of different sizes and there is uncertainty regarding the likelihood of this misclassification. In the absence of misclassfication in $C_i$, the EM algorithm would converge immediately so the added computational burden is minimal. 

\section*{Data Availability Statement}
The data that support the findings in this paper are openly available at \url{ https://doi.org/10.7910/DVN/IIDE2B}.

\section*{Acknowledgements}
Research in this article was in part supported by the National Institute of Allergy and Infectious Diseases of the National Institutes of Health (NIH) R01 AI136947. The content is
solely the responsibility of the authors and does not necessarily represent the official views of
the National Institutes of Health.

\section*{Supplementary Materials}

Web Appendices, Tables, and Figures referenced in Sections \ref{ss:2.4}, \ref{ss:2.5}, \ref{ss:2.6}, \ref{s:simulation}, and \ref{s:applications} are provided in Supplementary Materials.

\label{lastpage}


%

\bibliographystyle{plainnat}
\bibliography{multilevel.bib}

\end{document}


\maketitle

\renewcommand{\theequation}{a.\arabic{equation}}
\setcounter{equation}{0}
\section{Solution to the Constrained Optimization Problem}\label{appdx:A}

Recall that the multiply robust weights for individuals with observed outcomes are:
\begin{equation*}
    \hat{w}_{ij}^{MR} = \underset{w_{ij}}{\text{argmax}}\prod_{i=1}^{s}\prod_{j=1}^{m_i} w_{ij}
\end{equation*}
subject to the following constraints:
\begin{equation}\label{eq:A1}
    w_{ij} \geq 0 \hspace{0.25cm} (i = 1,...,s, j = 1,...,m_i),
\end{equation}
\begin{equation}\label{eq:A2}
    \sum_{i=1}^{s}\sum_{j=1}^{m_i}w_{ij} = 1,
\end{equation}
\begin{equation}\label{eq:A3}
    \sum_{i=1}^{s}\sum_{j=1}^{m_i} w_{ij}\phi_{ij}^{k}(\boldsymbol{\hat{\eta}}^k)\lambda_{i}^{\ell}(\boldsymbol{\hat{\gamma}}^\ell)  = \hat{\chi}^{k,\ell} \hspace{0.25cm} (k = 1,...,K, \ell = 1,...,L),
\end{equation}
where $\hat{\chi}^{k,l} = \left(\sum_{i=1}^{M} n_i\right)^{-1}\sum_{i=1}^{M}\sum_{j=1}^{n_i} \phi_{ij}^{k}(\hat{\boldsymbol{\eta}}^k)\lambda_{i}^{\ell}(\hat{\boldsymbol{\gamma}}^\ell)$. Given constraint (\ref{eq:A2}), (\ref{eq:A3}) can be written as
\begin{align}\label{eq:A4}
    0 
    &= \sum_{i=1}^{s}\sum_{j=1}^{m_i} w_{ij}\phi_{ij}^{k}(\boldsymbol{\hat{\eta}}^k)\lambda_{i}^{\ell}(\boldsymbol{\hat{\gamma}}^\ell) - \hat{\chi}^{k,l} \nonumber \\
    &= \sum_{i=1}^{s}\sum_{j=1}^{m_i} w_{ij}\phi_{ij}^{k}(\boldsymbol{\hat{\eta}}^k)\lambda_{i}^{\ell}(\boldsymbol{\hat{\gamma}}^\ell) -  \hat{\chi}^{k,l} \underbrace{\sum_{i=1}^{s}\sum_{j=1}^{m_i} w_{ij}}_{=1} \nonumber \\
    &= \sum_{i=1}^{s}\sum_{j=1}^{m_i} w_{ij}[\phi_{ij}^{k}(\boldsymbol{\hat{\eta}}^k)\lambda_{i}^{\ell}(\boldsymbol{\hat{\gamma}}^\ell) - \hat{\chi}^{k,l}] \hspace{0.25cm} (k = 1,...,K,\ell = 1,...,L).
\end{align}
Recall that $\hat{g}_{ij}(\hat{\boldsymbol{\eta}},\hat{\boldsymbol{\gamma}}) = \{\phi_{ij}^{1}(\boldsymbol{\hat{\eta}}^1)\lambda_{i}^{1}(\boldsymbol{\hat{\gamma}}^1) - \hat{\chi}^{1,1},...,\phi_{ij}^{K}(\boldsymbol{\hat{\eta}}^K)\lambda_{i}^{L}(\boldsymbol{\hat{\gamma}}^L) - \hat{\chi}^{K,L}\}^\mathrm{T}$. Thus, (\ref{eq:A4}) in vector form is equivalent to:
\begin{align}\label{eq:A5}
    \boldsymbol{0} &= \sum_{i=1}^{s}\sum_{j=1}^{m_i} w_{ij}\left\{\phi_{ij}^{1}(\boldsymbol{\hat{\eta}}^1)\lambda_{i}^{1}(\boldsymbol{\hat{\gamma}}^1) - \hat{\chi}^{1,1},...,\phi_{ij}^{K}(\boldsymbol{\hat{\eta}}^K)\lambda_{i}^{L}(\boldsymbol{\hat{\gamma}}^L) - \hat{\chi}^{K,L}\right\}^\mathrm{T} \nonumber \\
    &=  \sum_{i=1}^{s}\sum_{j=1}^{m_i} w_{ij}\hat{g}_{ij}(\hat{\boldsymbol{\eta}},\hat{\boldsymbol{\gamma}}).
\end{align}
Note that maximizing $\prod_{i=1}^{s}\prod_{j=1}^{m_i} w_{ij}$ is the same as maximizing $\sum_{i=1}^{s}\sum_{j=1}^{m_i} \log w_{ij}$ subject to constraints (\ref{eq:A1})-(\ref{eq:A3}). We use the Lagrange multiplier method to solve this optimization problem. Let
\begin{align} \label{eq:A6}
    Q &= \sum_{i=1}^{s}\sum_{j=1}^{m_i} \log w_{ij} + S\left(\sum_{i=1}^{s}\sum_{j=1}^{m_i} w_{ij} - 1\right) 
    - \left(\sum_{i=1}^{s} m_i\right)\boldsymbol{\rho}^\mathrm{T} \sum_{i=1}^{s}\sum_{j=1}^{m_i} w_{ij}\hat{g}_{ij}(\hat{\boldsymbol{\eta}},\hat{\boldsymbol{\gamma}}),
\end{align}
where $S$ is the Lagrange multiplier for constraint (\ref{eq:A2}) and $\boldsymbol{\rho}^\mathrm{T} = \left(\rho_{11},\rho_{12},...,\rho_{KL} \right)$ is the Lagrange multiplier for (\ref{eq:A5}). We also introduce a constant term, $\sum_{i=1}^{s} m_i$, for convenience. Take the partial derivative of $Q$ w.r.t. $w_{ij}, S, \boldsymbol{\rho}^\mathrm{T}$ and set them to 0, we have
\begin{align}
    \frac{\partial{Q}}{\partial{w_{ij}}} 
    &= w_{ij}^{-1} + S - \left(\sum_{i=1}^{s} m_i\right)\boldsymbol{\rho}^\mathrm{T} \hat{g}_{ij}(\hat{\boldsymbol{\eta}},\hat{\boldsymbol{\gamma}})
    \overset{\mathrm{set}}{=} 0 \hspace{0.2cm}, \label{eq:A7}\\
    \frac{\partial{Q}}{\partial{S}}
    &= \sum_{i=1}^{s}\sum_{i=1}^{m_i}w_{ij} - 1 \overset{\mathrm{set}}{=} 0,  \label{eq:A8}\\
    \frac{\partial{Q}}{\partial{\boldsymbol{\rho}^\mathrm{T}}}
    &= -\left(\sum_{i=1}^{s} m_i\right) \sum_{i=1}^{s}\sum_{j=1}^{m_i} w_{ij}\hat{g}_{ij}(\hat{\boldsymbol{\eta}},\hat{\boldsymbol{\gamma}}) \overset{\mathrm{set}}{=} \boldsymbol{0}. \label{eq:A9}
\end{align}
Multiply (\ref{eq:A7}) by $w_{ij}$, we have 
\begin{align*}
    0 
    &= w_{ij}\frac{\partial{Q}}{\partial{w_{ij}}}  \\
    &= \sum_{i=1}^{s}\sum_{j=1}^{m_i}  w_{ij}\frac{\partial{Q}}{\partial{w_{ij}}}   \\
    &= \sum_{i=1}^{s}\sum_{j=1}^{m_i} \left[1 + S w_{ij} - \left(\sum_{i=1}^{s} m_i\right)\boldsymbol{\rho}^\mathrm{T} \hat{g}_{ij}(\hat{\boldsymbol{\eta}},\hat{\boldsymbol{\gamma}})w_{ij} \right] \\
    &= \sum_{i=1}^{s} m_i + S \underbrace{\sum_{i=1}^{s}\sum_{j=1}^{m_i} w_{ij}}_{=1} - \left(\sum_{i=1}^{s} m_i\right)\boldsymbol{\rho}^\mathrm{T}  \underbrace{\sum_{i=1}^{s}\sum_{j=1}^{m_i} w_{ij}\hat{g}_{ij}(\hat{\boldsymbol{\eta}},\hat{\boldsymbol{\gamma}})}_{=0} \\
   &= \sum_{i=1}^{s} m_i + S \\
   &\Longrightarrow \hat{S} = -\sum_{i=1}^{s} m_i
\end{align*}
Plugging $\hat{S}$ back to (\ref{eq:A7}), we have
\begin{align}\label{eq:A10}
    \hat{w}_{ij}(\boldsymbol{\rho})
    &= \left[\sum_{i=1}^{s} m_i + \left(\sum_{i=1}^{s} m_i\right)\boldsymbol{\rho}^\mathrm{T} \hat{g}_{ij}(\hat{\boldsymbol{\eta}},\hat{\boldsymbol{\gamma}})\right]^{-1} \nonumber \\
    &= \frac{1}{\sum_{i=1}^{s} m_i} \frac{1}{1 + \boldsymbol{\rho}^\mathrm{T} \hat{g}_{ij}(\hat{\boldsymbol{\eta}},\hat{\boldsymbol{\gamma}})}.
\end{align}
$\boldsymbol{\hat{\rho}}^\mathrm{T}$ can be obtained by plugging $\hat{w}_{ij}(\boldsymbol{\rho})$ back to (\ref{eq:A9}) and solving
\begin{align}\label{eq:A11}
    -\left(\sum_{i=1}^{s} m_i\right)\sum_{i=1}^{s} \sum_{j=1}^{m_i} \hat{w}_{ij}(\boldsymbol{\rho})\hat{g}_{ij}(\hat{\boldsymbol{\eta}},\hat{\boldsymbol{\gamma}}) &=  - \sum_{i=1}^{s}\sum_{j=1}^{m_i} \frac{\hat{g}_{ij}(\hat{\boldsymbol{\eta}},\hat{\boldsymbol{\gamma}})}{1 + \boldsymbol{\rho}^\mathrm{T} \hat{g}_{ij}(\hat{\boldsymbol{\eta}},\hat{\boldsymbol{\gamma}})} = \boldsymbol{0} \nonumber \\
    &\Longrightarrow  \sum_{i=1}^{s}\sum_{j=1}^{m_i} \frac{\hat{g}_{ij}(\hat{\boldsymbol{\eta}},\hat{\boldsymbol{\gamma}})}{1 + \boldsymbol{\rho}^\mathrm{T} \hat{g}_{ij}(\hat{\boldsymbol{\eta}},\hat{\boldsymbol{\gamma}})} = \boldsymbol{0}
\end{align}
Hence, the multiply robust weight is
\begin{align}\label{eq:A12}
    \hat{w}_{ij}^{MR}
    &= \frac{1}{\sum_{i=1}^{s} m_i} \frac{1}{1 + \hat{\boldsymbol{\rho}}^\mathrm{T} \hat{g}_{ij}(\hat{\boldsymbol{\eta}},\hat{\boldsymbol{\gamma}})}.
\end{align}

\renewcommand{\theequation}{b.\arabic{equation}}
\setcounter{equation}{0}
\section{Multiple Robustness of $\hat{\boldsymbol{\beta}}_{MR}$} \label{appdx:B}
This section proves the multiple robustness of our proposed MMR-GEE estimator. Here, we consider the setting with $n_i = n$ and $M \to \infty$. Furthermore, we use subscript asterisk to denote probability limits, $(\boldsymbol{\eta}_0, \boldsymbol{\gamma}_0)$ to denote the true parameters of the PS models, and $\boldsymbol{\beta}_0$ to denote the true parameters of the marginal model.

\begin{theorem} \label{thm:1}
When $\mathcal{P}_1$ contains a correct model for $\phi_{ij}\left(\boldsymbol{\eta}\right)$ and $\mathcal{P}_2$ contains a correct model for  $\lambda_i\left(\gamma\right)$, as $M \to \infty$, 
$\hat{\boldsymbol{\beta}}_{MR} \stackrel{P}{\to} \boldsymbol{\beta}_0$. 
\end{theorem}

\begin{proof}
Without loss of generality, let $\phi_{ij}^{1}(\boldsymbol{\eta}^{1})$ and $\lambda_{i}^{1}(\boldsymbol{\gamma}^{1})$, the first model in $\mathcal{P}_1$ and $\mathcal{P}_2$, respectively, be the correctly specified models. We first show that:
\begin{align} \label{eq:B1}
   \chi^{1,1}\left(\hat{\boldsymbol{\eta}}^{1}, \hat{\boldsymbol{\gamma}}^{1}\right) \stackrel{P}{\longrightarrow} \mathbb{E}\left[\phi_{ij}^{1}\left(\boldsymbol{\eta}^{1}_{*}\right) \lambda_{i}^{1}\left(\boldsymbol{\gamma}^{1}_{*}\right)\right] \stackrel{\text { def }}{=} \chi^{1,1}_{*}. 
\end{align}
See Web Appendix \ref{appdx:B1} for detailed proof. Next, we establish the connection between the multiply robust weights and $p_{ij}$, the empirical probability of $(Y_{ij}, A_i, \boldsymbol{Z}_i, \boldsymbol{X}_{ij})$ conditional on $R_{ij}C_{i} = 1$ for $i =1,...,s, j = 1,...,m_i$. The estimator for $p_{ij}$ can be obtained by solving another version of the constrained optimization problem as follows:
\begin{align}\label{eq:B2}
    \hat{p}_{ij} = \underset{p_{ij}}{\text{argmax}}\prod_{i=1}^{s}\prod_{j=1}^{m_i} p_{ij},
\end{align}
subject to the following constraints:
\begin{align}
    p_{ij} \geq 0,  \sum_{i=1}^{s}\sum_{j=1}^{m_i}p_{ij} = 1,
    \sum_{i=1}^{s}\sum_{j=1}^{m_i} p_{ij}\frac{\left\{\phi_{ij}^{k}(\boldsymbol{\hat{\eta}}^k)\lambda_{i}^{\ell}(\boldsymbol{\hat{\gamma}}^{\ell})  - \chi^{k,l}(\hat{\boldsymbol{\eta}}^k,
    \hat{\boldsymbol{\gamma}}^{\ell}) \right\}}{\phi_{ij}^{k}(\boldsymbol{\hat{\eta}}^k)\lambda_{i}^{\ell}(\boldsymbol{\hat{\gamma}}^{\ell})} = 0 \hspace{0.2cm} (k = 1,...,K, \ell = 1,...,L).
\end{align}
Here, the first two constraints ensure $p_{ij}$ to be proper empirical probabilities and the third constraint is the empirical version of (\ref{eq:A3}). Applying the same Lagrange multiplier method from \ref{appdx:A}, we have:
\begin{align}\label{eq:B4}
    \hat{p}_{ij} = \frac{1}{\sum_{i=1}^{s}m_i} \frac{1}{1 + \hat{\boldsymbol{\varepsilon}}^{\mathrm{T}}\hat{g}_{ij}(\hat{\boldsymbol{\eta}},\hat{\boldsymbol{\gamma}})/\{\phi_{ij}^{1}(\boldsymbol{\hat{\eta}}^1)\lambda_{i}^{1}(\boldsymbol{\hat{\gamma}}^1)\}},
\end{align}
where $\hat{\boldsymbol{\varepsilon}}^{\mathrm{T}} = (\hat{\varepsilon}_{11},\hat{\varepsilon}_{12},...,\hat{\varepsilon}_{KL})$ is a $(KL) \times 1$ vector that solves the following equation:
\begin{align}\label{eq:B5}
    \sum_{i=1}^{s}\sum_{j=1}^{m_i} \frac{\hat{g}_{ij}(\hat{\boldsymbol{\eta}},\hat{\boldsymbol{\gamma}})/\{\phi_{ij}^{1}(\boldsymbol{\hat{\eta}}^1)\lambda_{i}^{1}(\boldsymbol{\hat{\gamma}}^1)\}}{1 + \boldsymbol{\varepsilon}^{\mathrm{T}}\hat{g}_{ij}(\hat{\boldsymbol{\eta}},\hat{\boldsymbol{\gamma}})/\{\phi_{ij}^{1}(\boldsymbol{\hat{\eta}}^1)\lambda_{i}^{1}(\boldsymbol{\hat{\gamma}}^1)\}} = \boldsymbol{0}.
\end{align}
With simple algebra manipulation, $\hat{\boldsymbol{\rho}}^{\mathrm{T}}$ can be expressed by $\hat{\boldsymbol{\varepsilon}}$ as:
\begin{align}\label{eq:B6}
    \hat{\boldsymbol{\rho}}^{\mathrm{T}} = \left\{\frac{\hat{\varepsilon}_{11}+1}{\chi^{1,1}\left(\hat{\boldsymbol{\eta}}^{1}, \hat{\boldsymbol{\gamma}}^{1}\right)}, \frac{\hat{\varepsilon}_{12}}{\chi^{1,1}\left(\hat{\boldsymbol{\eta}}^{1}, \hat{\boldsymbol{\gamma}}^{1}\right)},...,\frac{\hat{\varepsilon}_{KL}}{\chi^{1,1}\left(\hat{\boldsymbol{\eta}}^{1}, \hat{\boldsymbol{\gamma}}^{1}\right)} \right\},
\end{align}
which allows us to write $\hat{w}_{ij}^{MR}$ as:
\begin{equation}\label{eq:B7}
    \hat{w}_{ij}^{MR} = \hat{p}_{ij} \frac{ \chi^{1,1}\left(\hat{\boldsymbol{\eta}}^{1}, \hat{\boldsymbol{\gamma}}^{1}\right)}{\phi_{ij}^1(\hat{\boldsymbol{\eta}}^1)\lambda_{i}^1(\hat{\boldsymbol{\gamma}}^1)}.
\end{equation}
The detailed proof for the relationship between $\hat{w}_{ij}$ and $\hat{p}_{ij}$ can be found in Web Appendix \ref{appdx:B2}. Lastly, the proportion of observed individuals would converge in probability to $\chi_{*}^{1,1}$ under the assumption that $\phi_{ij}^1(\boldsymbol{\eta}^1)$ and $\lambda_{i}^1(\boldsymbol{\gamma}^1)$ are the correctly specified models. That is,
\begin{align}\label{eq:B8}
    \frac{\sum_{i=1}^{s} m_i}{\sum_{i=1}^{M} n_{i}}\stackrel{P}{\longrightarrow} \chi_{*}^{1,1},
\end{align}
where the proof can be found in Web Appendix \ref{appdx:B3}. 

Using (1) the relationship between $\hat{w}_{ij}^{MR}$ and $\hat{p}_{ij}$ from Equation (\ref{eq:B7}), (2) the convergence of the proportion of observed individuals from Equation (\ref{eq:B8}), and (3) the fact that $\hat{\boldsymbol{\varepsilon}} \stackrel{P}{\longrightarrow} \boldsymbol{0}$ based on the empirical likelihood theory \citep{han2014multiply}, we show that the multiply robust weight is equivalent to (see Web Appendix \ref{appdx:B4} for the derivation):
\begin{equation}\label{eq:B9}
    \hat{w}_{ij}^{MR} = \frac{1+ o_p(1)}{\left(\sum_{i=1}^{M} n_i\right)\phi_{ij}^1(\boldsymbol{\eta}_{0})\lambda_{i}^1(\boldsymbol{\gamma}_{0})} , \hspace{0.5cm} i = 1, \ldots, s, j = 1, \ldots m_i,
\end{equation}
which allows us to establish the relationship that:
\begin{align}\label{eq:B10}
    \hat{\boldsymbol{W}}_{i}^{MR} 
    = \text{diag}\left[
    R_{ij}C_{i}\hat{w}_{ij}^{MR}\right]_{j = 1, \ldots, n_i}  
    = \text{diag}\left[
    \frac{R_{ij} + o_p(1)}{\left(\sum_{i=1}^{M} n_i\right)\phi_{ij}^1(\boldsymbol{\eta}_{0})\lambda_{i}^1(\boldsymbol{\gamma}_{0})} \right]_{j = 1, \ldots, n_i}.   
\end{align}
Equation (\ref{eq:B10}) shows that the multiply robust weights are asymptotically proportional to the multi-level weights of the correctly specified MIPW-GEE estimator. As the number of clusters $M$ goes to infinity, we have:
\begin{align} \label{eq:B11}
    \nonumber &\hspace{0.5cm}\left(\frac{\sum_{i=1}^{M}{n_i}}{M} \right)\sum_{i=1}^{M} \frac{\partial{\boldsymbol{\mu}_i(\boldsymbol{\beta}_0, A_i)}}{\partial{\boldsymbol{\beta}}}\boldsymbol{V}_{i}^{-1} \hat{\boldsymbol{W}}_{i}^{MR}   (\boldsymbol{Y}_i - \boldsymbol{\mu}_i(\boldsymbol{\beta}_0, A_i)) \\
    \nonumber &= \frac{1}{M}\sum_{i=1}^{M} \frac{\partial{\boldsymbol{\mu}_i(\boldsymbol{\beta}_0, A_i)}}{\partial{\boldsymbol{\beta}}}\boldsymbol{V}_{i}^{-1} \text{diag}\left[
    \frac{R_{ij}C_{i} + o_p(1)}{\phi_{ij}^1(\boldsymbol{\eta}_{0})\lambda_{i}^1(\boldsymbol{\gamma}_{0})}\right] (\boldsymbol{Y}_i - \boldsymbol{\mu}_i(\boldsymbol{\beta}_0, A_i)) \\
    &\stackrel{p}{\longrightarrow} E\left\{\frac{\partial{\boldsymbol{\mu}_i(\boldsymbol{\beta}_0, A_i)}}{\partial{\boldsymbol{\beta}}}\boldsymbol{V}_{i}^{-1} \text{diag}\left[
    \frac{R_{ij}C_{i}}{\phi_{ij}^1(\boldsymbol{\eta}_{0})\lambda_{i}^1(\boldsymbol{\gamma}_{0})}\right] (\boldsymbol{Y}_i - \boldsymbol{\mu}_i(\boldsymbol{\beta}_0, A_i))\right\} = 0.
\end{align}
See Web Appendix \ref{appdx:B5} for the proof of (\ref{eq:B11}).
\end{proof}

\subsection{Showing that $\chi^{1,1}\left(\hat{\boldsymbol{\eta}}^{1}, \hat{\boldsymbol{\gamma}}^{1}\right) \stackrel{P}{\rightarrow} \mathbb{E}\left[\phi_{ij}^{1}\left(\boldsymbol{\eta}^{1}_{*}\right) \lambda_{i}^{1}\left(\boldsymbol{\gamma}^{1}_{*}\right)\right] \stackrel{\text { def }}{=} \chi^{1,1}_{*}$}
\label{appdx:B1}

By the continuous mapping theorem, $\hat{\boldsymbol{\eta}}^{1} \stackrel{P}{\rightarrow} \boldsymbol{\eta}^{1}_{*}$ and $\hat{\boldsymbol{\gamma}}^{1} \stackrel{P}{\rightarrow} \boldsymbol{\gamma}^{1}_{*}$ imply that $\phi_{ij}(\hat{\boldsymbol{\eta}}^{1}) \stackrel{P}{\rightarrow} \phi_{ij}(\boldsymbol{\eta}^{1}_{*})$ and $\lambda_{i}(\hat{\boldsymbol{\gamma}}^{1}) \stackrel{P}{\rightarrow} \lambda_{i}(\boldsymbol{\gamma}^{1}_{*})$. Applying the continuous mapping theorem again, we have $\phi_{ij}(\hat{\boldsymbol{\eta}}^{1})\lambda_{i}(\hat{\boldsymbol{\gamma}}^{1}) \stackrel{P}{\rightarrow} \phi_{ij}(\boldsymbol{\eta}^{1}_{*})\lambda_{i}(\boldsymbol{\gamma}^{1}_{*})$.
\begin{align*}
    \chi^{1,1}\left(\hat{\boldsymbol{\eta}}^{1}, \hat{\boldsymbol{\gamma}}^{1}\right)
    &= \left(\sum_{i=1}^{M} n_i\right)^{-1}\sum_{i=1}^{M}\sum_{j=1}^{n_i} \phi_{ij}^{k}(\hat{\boldsymbol{\eta}}^k)\lambda_{i}^{\ell}(\hat{\boldsymbol{\gamma}}^\ell) \\
    &= \frac{1}{M} \sum_{i=1}^{M} \left(\frac{M}{\sum_{i=1}^{M} n_i}\right)\sum_{j=1}^{n_i} \phi_{ij}^{k}(\hat{\boldsymbol{\eta}}^k)\lambda_{i}^{\ell}(\hat{\boldsymbol{\gamma}}^\ell) \\
    &\stackrel{P}{\rightarrow} \lim_{M \to \infty}\frac{1}{M} \sum_{i=1}^{M} \mathbb{E} \left[\left(\frac{M}{\sum_{i=1}^{M} n_i}\right)\sum_{j=1}^{n_i} \phi_{ij}^{k}(\boldsymbol{\eta}_{*}^k)\lambda_{i}^{\ell}(\boldsymbol{\gamma}_{*}^\ell)\right] \\
    &= \lim_{M \to \infty}\frac{\sum_{i=1}^{M}\sum_{j=1}^{n_i} \mathbb{E} \left[ \phi_{ij}^{k}(\boldsymbol{\eta}_{*}^k)\lambda_{i}^{\ell}(\boldsymbol{\gamma}_{*}^\ell)\right]}{\sum_{i=1}^{M} n_i}  \\
    &= \chi^{1,1}_{*}
\end{align*}

\subsection{Showing that $\hat{w}_{ij}^{MR} = \hat{p}_{ij} \chi^{1,1}\left(\hat{\boldsymbol{\eta}}^{1}, \hat{\boldsymbol{\gamma}}^{1}\right)/\{\phi_{ij}^1(\hat{\boldsymbol{\eta}}^1)\lambda_{i}^1(\hat{\boldsymbol{\gamma}}^1)\}$} \label{appdx:B2}
First, we show that solving $\hat{\boldsymbol{\rho}}$ of the multiply robust weight is equivalent to solving $\hat{\boldsymbol{\varepsilon}}$ of the empirical probability. Recall that $\hat{\boldsymbol{\varepsilon}}$ is obtained by solving the equation:
\begin{align*}
        \sum_{i=1}^{s}\sum_{j=1}^{m_i} \frac{\hat{g}_{ij}(\hat{\boldsymbol{\eta}},\hat{\boldsymbol{\gamma}})/\{\phi_{ij}^{1}(\boldsymbol{\hat{\eta}}^1)\lambda_{i}^{1}(\boldsymbol{\hat{\gamma}}^1)\}}{1 + \boldsymbol{\varepsilon}^{\mathrm{T}}\hat{g}_{ij}(\hat{\boldsymbol{\eta}},\hat{\boldsymbol{\gamma}})/\{\phi_{ij}^{1}(\boldsymbol{\hat{\eta}}^1)\lambda_{i}^{1}(\boldsymbol{\hat{\gamma}}^1)\}} = \boldsymbol{0}.
\end{align*}
Rearranging the equation, we have
\begin{align*}
    \boldsymbol{0} 
    &= \sum_{i=1}^{s}\sum_{j=1}^{m_i} \frac{\hat{g}_{ij}(\hat{\boldsymbol{\eta}},\hat{\boldsymbol{\gamma}})/\{\phi_{ij}^{1}(\boldsymbol{\hat{\eta}}^1)\lambda_{i}^{1}(\boldsymbol{\hat{\gamma}}^1)\}}{1 + \hat{\boldsymbol{\varepsilon}}^{\mathrm{T}}\hat{g}_{ij}(\hat{\boldsymbol{\eta}},\hat{\boldsymbol{\gamma}})/\{\phi_{ij}^{1}(\boldsymbol{\hat{\eta}}^1)\lambda_{i}^{1}(\boldsymbol{\hat{\gamma}}^1)\}} \\
    &=  \sum_{i=1}^{s}\sum_{j=1}^{m_i} \frac{\hat{g}_{ij}(\hat{\boldsymbol{\eta}},\hat{\boldsymbol{\gamma}})}{\phi_{ij}^{1}(\boldsymbol{\hat{\eta}}^1)\lambda_{i}^{1}(\boldsymbol{\hat{\gamma}}^1) + \hat{\boldsymbol{\varepsilon}}^\mathrm{T}\hat{g}_{ij}(\hat{\boldsymbol{\eta}},\hat{\boldsymbol{\gamma}})} \\
    &= \frac{\chi^{1,1}\left(\hat{\boldsymbol{\eta}}^{1}, \hat{\boldsymbol{\gamma}}^{1}\right)}{\chi^{1,1}\left(\hat{\boldsymbol{\eta}}^{1}, \hat{\boldsymbol{\gamma}}^{1}\right)}  \sum_{i=1}^{s}\sum_{j=1}^{m_i} \frac{\hat{g}_{ij}(\hat{\boldsymbol{\eta}},\hat{\boldsymbol{\gamma}})}{\phi_{ij}^{1}(\boldsymbol{\hat{\eta}}^1)\lambda_{i}^{1}(\boldsymbol{\hat{\gamma}}^1) + \hat{\boldsymbol{\varepsilon}}^\mathrm{T}\hat{g}_{ij}(\hat{\boldsymbol{\eta}},\hat{\boldsymbol{\gamma}})} \\
    &= \frac{1}{\chi^{1,1}\left(\hat{\boldsymbol{\eta}}^{1}, \hat{\boldsymbol{\gamma}}^{1}\right)}  \sum_{i=1}^{s}\sum_{j=1}^{m_i} \frac{\hat{g}_{ij}(\hat{\boldsymbol{\eta}},\hat{\boldsymbol{\gamma}})}{\frac{\phi_{ij}^{1}(\boldsymbol{\hat{\eta}}^1)\lambda_{i}^{1}(\boldsymbol{\hat{\gamma}}^1)}{\chi^{1,1}\left(\hat{\boldsymbol{\eta}}^{1}, \hat{\boldsymbol{\gamma}}^{1}\right)} + \frac{\hat{\boldsymbol{\varepsilon}}^\mathrm{T}\hat{g}_{ij}(\hat{\boldsymbol{\eta}},\hat{\boldsymbol{\gamma}})}{\chi^{1,1}\left(\hat{\boldsymbol{\eta}}^{1}, \hat{\boldsymbol{\gamma}}^{1}\right)}} \\
    &= \frac{1}{\chi^{1,1}\left(\hat{\boldsymbol{\eta}}^{1}, \hat{\boldsymbol{\gamma}}^{1}\right)} \sum_{i=1}^{s}\sum_{j=1}^{m_i} \frac{\hat{g}_{ij}(\hat{\boldsymbol{\eta}},\hat{\boldsymbol{\gamma}})}{1 + \frac{\phi_{ij}^{1}(\boldsymbol{\hat{\eta}}^1)\lambda_{i}^{1}(\boldsymbol{\hat{\gamma}}^1) - \chi^{1,1}\left(\hat{\boldsymbol{\eta}}^{1}, \hat{\boldsymbol{\gamma}}^{1}\right)}{\chi^{1,1}\left(\hat{\boldsymbol{\eta}}^{1}, \hat{\boldsymbol{\gamma}}^{1}\right)} + \frac{\hat{\boldsymbol{\varepsilon}}^\mathrm{T}}{\chi^{1,1}\left(\hat{\boldsymbol{\eta}}^{1}, \hat{\boldsymbol{\gamma}}^{1}\right)}\hat{g}_{ij}(\hat{\boldsymbol{\eta}},\hat{\boldsymbol{\gamma}})} \\
    &= \frac{1}{\chi^{1,1}\left(\hat{\boldsymbol{\eta}}^{1}, \hat{\boldsymbol{\gamma}}^{1}\right)} \sum_{i=1}^{s}\sum_{j=1}^{m_i} \frac{\hat{g}_{ij}(\hat{\boldsymbol{\eta}},\hat{\boldsymbol{\gamma}})}{1 + \left\{\frac{\hat{\varepsilon}_{11}+1}{\chi^{1,1}\left(\hat{\boldsymbol{\eta}}^{1}, \hat{\boldsymbol{\gamma}}^{1}\right)}, \frac{\hat{\varepsilon}_{12}}{\chi^{1,1}\left(\hat{\boldsymbol{\eta}}^{1}, \hat{\boldsymbol{\gamma}}^{1}\right)},...,\frac{\hat{\varepsilon}_{KL}}{\chi^{1,1}\left(\hat{\boldsymbol{\eta}}^{1}, \hat{\boldsymbol{\gamma}}^{1}\right)} \right\}^\mathrm{T} \hat{g}_{ij}(\hat{\boldsymbol{\eta}},\hat{\boldsymbol{\gamma}})}
\end{align*}
Thus, we establish the relationship that the solution to equation (\ref{eq:B5}) is equivalent to
\begin{align*}
    \hat{\boldsymbol{\rho}}^\mathrm{T} = \left\{\frac{\hat{\varepsilon}_{11}+1}{\chi^{1,1}\left(\hat{\boldsymbol{\eta}}^{1}, \hat{\boldsymbol{\gamma}}^{1}\right)}, \frac{\hat{\varepsilon}_{12}}{\chi^{1,1}\left(\hat{\boldsymbol{\eta}}^{1}, \hat{\boldsymbol{\gamma}}^{1}\right)},...,\frac{\hat{\varepsilon}_{KL}}{\chi^{1,1}\left(\hat{\boldsymbol{\eta}}^{1}, \hat{\boldsymbol{\gamma}}^{1}\right)} \right\}.
\end{align*}

Plugging this alternative form of $\hat{\boldsymbol{\rho}}^\mathrm{T}$ into the multiply robust weight, we have
\begin{align*}
    \hat{w}_{ij}^{MR}
    &= \left( \frac{1}{\sum_{i=1}^{s}m_i} \right) \left( \frac{1}{1 + \hat{\boldsymbol{\rho}}^\mathrm{T} \hat{g}_{ij}(\hat{\boldsymbol{\eta}},\hat{\boldsymbol{\gamma}})} \right) \\
    &= \left( \frac{1}{\sum_{i=1}^{s}m_i} \right) \left( \frac{1}{1 + \left\{\frac{\hat{\varepsilon}_{11}+1}{\chi^{1,1}\left(\hat{\boldsymbol{\eta}}^{1}, \hat{\boldsymbol{\gamma}}^{1}\right)}, \frac{\hat{\varepsilon}_{12}}{\chi^{1,1}\left(\hat{\boldsymbol{\eta}}^{1}, \hat{\boldsymbol{\gamma}}^{1}\right)},...,\frac{\hat{\varepsilon}_{KL}}{\chi^{1,1}\left(\hat{\boldsymbol{\eta}}^{1}, \hat{\boldsymbol{\gamma}}^{1}\right)} \right\} \hat{g}_{ij}(\hat{\boldsymbol{\eta}},\hat{\boldsymbol{\gamma}})} \right) \\
    &= \left( \frac{1}{\sum_{i=1}^{s}m_i} \right) \left( \frac{1}{1 + \frac{\phi_{ij}^{1}(\boldsymbol{\hat{\eta}}^1)\lambda_{i}^{1}(\boldsymbol{\hat{\gamma}}^1) - \chi^{1,1}\left(\hat{\boldsymbol{\eta}}^{1}, \hat{\boldsymbol{\gamma}}^{1}\right)}{\chi^{1,1}\left(\hat{\boldsymbol{\eta}}^{1}, \hat{\boldsymbol{\gamma}}^{1}\right)} + \left\{\frac{\hat{\varepsilon}_{11}}{\chi^{1,1}\left(\hat{\boldsymbol{\eta}}^{1}, \hat{\boldsymbol{\gamma}}^{1}\right)}, \frac{\hat{\varepsilon}_{12}}{\chi^{1,1}\left(\hat{\boldsymbol{\eta}}^{1}, \hat{\boldsymbol{\gamma}}^{1}\right)},...,\frac{\hat{\varepsilon}_{KL}}{\chi^{1,1}\left(\hat{\boldsymbol{\eta}}^{1}, \hat{\boldsymbol{\gamma}}^{1}\right)} \right\} \hat{g}_{ij}(\hat{\boldsymbol{\eta}},\hat{\boldsymbol{\gamma}})} \right) \\
    &= \left( \frac{1}{\sum_{i=1}^{s}m_i} \right) \left( \frac{1}{\frac{\phi_{ij}^{1}(\boldsymbol{\hat{\eta}}^1)\lambda_{i}^{1}(\boldsymbol{\hat{\gamma}}^1)}{\chi^{1,1}\left(\hat{\boldsymbol{\eta}}^{1}, \hat{\boldsymbol{\gamma}}^{1}\right)} + \left\{\frac{\hat{\varepsilon}_{11}}{\chi^{1,1}\left(\hat{\boldsymbol{\eta}}^{1}, \hat{\boldsymbol{\gamma}}^{1}\right)}, \frac{\hat{\varepsilon}_{12}}{\chi^{1,1}\left(\hat{\boldsymbol{\eta}}^{1}, \hat{\boldsymbol{\gamma}}^{1}\right)},...,\frac{\hat{\varepsilon}_{KL}}{\chi^{1,1}\left(\hat{\boldsymbol{\eta}}^{1}, \hat{\boldsymbol{\gamma}}^{1}\right)} \right\} \hat{g}_{ij}(\hat{\boldsymbol{\eta}},\hat{\boldsymbol{\gamma}})} \right) \\
    &= \left( \frac{1}{\sum_{i=1}^{s}m_i} \right) \left( \frac{\chi^{1,1}\left(\hat{\boldsymbol{\eta}}^{1}, \hat{\boldsymbol{\gamma}}^{1}\right)}{\phi_{ij}^{1}(\boldsymbol{\hat{\eta}}^1)\lambda_{i}^{1}(\boldsymbol{\hat{\gamma}}^1) + \left\{\hat{\varepsilon}_{11}, \hat{\varepsilon}_{12},...,\hat{\varepsilon}_{KL} \right\} \hat{g}_{ij}(\hat{\boldsymbol{\eta}},\hat{\boldsymbol{\gamma}})} \right) \\
    &= \left( \frac{1}{\sum_{i=1}^{s}m_i} \right) \left( \frac{\chi^{1,1}\left(\hat{\boldsymbol{\eta}}^{1}, \hat{\boldsymbol{\gamma}}^{1}\right)/\phi_{ij}^{1}(\boldsymbol{\hat{\eta}}^1)\lambda_{i}^{1}(\boldsymbol{\hat{\gamma}}^1)}{1 + \hat{\boldsymbol{\varepsilon}}^\mathrm{T} \hat{g}_{ij}(\hat{\boldsymbol{\eta}},\hat{\boldsymbol{\gamma}})/\phi_{ij}^{1}(\boldsymbol{\hat{\eta}}^1)\lambda_{i}^{1}(\boldsymbol{\hat{\gamma}}^1)} \right) \\
    &= \underbrace{\left( \frac{1}{\sum_{i=1}^{s}m_i} \right) \left( \frac{1}{1 + \hat{\boldsymbol{\varepsilon}}^\mathrm{T} \hat{g}_{ij}(\hat{\boldsymbol{\eta}},\hat{\boldsymbol{\gamma}})/\phi_{ij}^{1}(\boldsymbol{\hat{\eta}}^1)\lambda_{i}^{1}(\boldsymbol{\hat{\gamma}}^1)} \right)}_{\hat{p}_{ij}} \left\{ \chi^{1,1}\left(\hat{\boldsymbol{\eta}}^{1}, \hat{\boldsymbol{\gamma}}^{1}\right)/\phi_{ij}^{1}(\boldsymbol{\hat{\eta}}^1)\lambda_{i}^{1}(\boldsymbol{\hat{\gamma}}^1) \right\} \\ 
    &= \hat{p}_{ij} \chi^{1,1}\left(\hat{\boldsymbol{\eta}}^{1}, \hat{\boldsymbol{\gamma}}^{1}\right)/\{\phi_{ij}^1(\hat{\boldsymbol{\eta}}^1)\lambda_{i}^1(\hat{\boldsymbol{\gamma}}^1)\}
\end{align*}

\subsection{Showing that $\frac{\sum_{i=1}^{s} m_i }{\sum_{i=1}^{M} n_{i}}\stackrel{P}{\rightarrow} \chi_{*}^{1,1}$}
\label{appdx:B3}

\begin{align*}
    \frac{\sum_{i=1}^{s} m_i }{\sum_{i=1}^{M} n_{i}}
    &= \frac{\sum_{i=1}^{M}\sum_{j=1}^{n_i} R_{ij} C_{i}}{\sum_{i=1}^{M} n_{i}} \\
    &= \frac{1}{M} \sum_{i=1}^{M} \left(\frac{M }{\sum_{i=1}^{M} n_{i}}\right)\sum_{j=1}^{n_i} R_{ij}C_{i} \\
    &\stackrel{P}{\rightarrow} \lim_{M \to \infty}\frac{1}{M} \sum_{i=1}^{M} \mathbb{E}\left[ \left(\frac{M }{\sum_{i=1}^{M} n_{i}}\right)\sum_{j=1}^{n_i} R_{ij}C_{i}\right] \\
    &= \lim_{M \to \infty}\frac{1 }{\sum_{i=1}^{M} n_{i}} \sum_{i=1}^{M}\sum_{j=1}^{n_i} \mathbb{E}\left[ R_{ij}C_{i}\right] \\
    &= \lim_{M \to \infty}\frac{1 }{\sum_{i=1}^{M} n_{i}} \sum_{i=1}^{M}\sum_{j=1}^{n_i} \mathbb{E}\left\{ \mathbb{E}\left[R_{ij}C_{i}\mid \boldsymbol{X}_{ij}, \boldsymbol{Z}_{i}, A_i\right]\right\} \\
    &= \lim_{M \to \infty}\frac{1 }{\sum_{i=1}^{M} n_{i}} \sum_{i=1}^{M}\sum_{j=1}^{n_i} \mathbb{E}\left\{ \phi_{ij}^1(\boldsymbol{\eta}_0)\lambda_{i}^1(\boldsymbol{\gamma}_0)\right\} \\
    &= \mathbb{E}\left\{ \phi_{ij}^1(\boldsymbol{\eta}_0)\lambda_{i}^1(\boldsymbol{\gamma}_0)\right\} \\
    &= \chi_{*}^{1,1}
\end{align*}

\subsection{Showing that $\hat{w}_{ij}^{MR} = \frac{1 + o_p(1)}{\left(\sum_{i=1}^{M} n_i\right)\phi_{ij}^1(\boldsymbol{\eta_0})\lambda_{i}^1(\boldsymbol{\gamma_0})} $}
\label{appdx:B4}

\begin{align*}
    \hat{w}_{ij}^{MR} 
    &= \hat{p}_{ij} \frac{ \chi^{1,1}\left(\hat{\boldsymbol{\eta}}^{1}, \hat{\boldsymbol{\gamma}}^{1}\right)}{\phi_{ij}^1(\hat{\boldsymbol{\eta}}^1)\lambda_{i}^1(\hat{\boldsymbol{\gamma}}^1)} \\
    &= \left\{\left(\frac{1}{\sum_{i=1}^{s}m_i}\right) \frac{1}{1 + \hat{\boldsymbol{\varepsilon}}^{\mathrm{T}}\hat{g}_{ij}(\hat{\boldsymbol{\eta}},\hat{\boldsymbol{\gamma}})/\{\phi_{ij}^{1}(\boldsymbol{\hat{\eta}}^1)\lambda_{i}^{1}(\boldsymbol{\hat{\gamma}}^1)\}} \right\}\frac{ \chi^{1,1}\left(\hat{\boldsymbol{\eta}}^{1}, \hat{\boldsymbol{\gamma}}^{1}\right)}{\phi_{ij}^1(\hat{\boldsymbol{\eta}}^1)\lambda_{i}^1(\hat{\boldsymbol{\gamma}}^1)} \\
    &= \left(\frac{1}{\sum_{i=1}^{s}m_i}\right) \frac{\chi^{1,1}\left(\hat{\boldsymbol{\eta}}^{1}, \hat{\boldsymbol{\gamma}}^{1}\right)}{\phi_{ij}^{1}(\boldsymbol{\hat{\eta}}^1)\lambda_{i}^{1}(\boldsymbol{\hat{\gamma}}^1) + \hat{\boldsymbol{\varepsilon}}^{\mathrm{T}}\hat{g}_{ij}(\hat{\boldsymbol{\eta}},\hat{\boldsymbol{\gamma}})} \\
    &= \left(\frac{1}{\sum_{i=1}^{s}m_i}\right) \frac{\chi^{1,1}_{*} + o_p(1)}{\phi_{ij}^{1}(\boldsymbol{\eta}_{0}^1)\lambda_{i}^{1}(\boldsymbol{\gamma}_{0}^1) + \hat{\boldsymbol{\varepsilon}}^{\mathrm{T}}\hat{g}_{ij}(\hat{\boldsymbol{\eta}},\hat{\boldsymbol{\gamma}})} \\
    &= \frac{1}{\sum_{i=1}^{M}n_i} \underbrace{\left(\frac{\sum_{i=1}^{M}n_i}{\sum_{i=1}^{s}m_i} \right)}_{\stackrel{P}{\rightarrow} \left(\chi^{1,1}_{*}\right)^{-1}} \frac{\chi^{1,1}_{*} + o_p(1)}{\phi_{ij}^{1}(\boldsymbol{\eta}_{0}^1)\lambda_{i}^{1}(\boldsymbol{\gamma}_{0}^1) + \hat{\boldsymbol{\varepsilon}}^{\mathrm{T}}\hat{g}_{ij}(\hat{\boldsymbol{\eta}},\hat{\boldsymbol{\gamma}})} \\
    &= \left(\frac{1}{\sum_{i=1}^{M}n_i}\right) \frac{1 +  o_p(1)}{\phi_{ij}^{1}(\boldsymbol{\eta}_{0}^1)\lambda_{i}^{1}(\boldsymbol{\gamma}_{0}^1) + \hat{\boldsymbol{\varepsilon}}^{\mathrm{T}}\hat{g}_{ij}(\hat{\boldsymbol{\eta}},\hat{\boldsymbol{\gamma}})} \\
    &= \left(\frac{1}{\sum_{i=1}^{M}n_i}\right) \frac{1 +  o_p(1)}{\phi_{ij}^{1}(\boldsymbol{\eta}_0)\lambda_{i}^{1}(\boldsymbol{\gamma}_0)} \\
\end{align*}

\subsection{Showing that $\mathbb{E}\left\{\frac{\partial{\boldsymbol{\mu}_i(\boldsymbol{\beta}_0, A_i)}}{\partial{\boldsymbol{\beta}}}\boldsymbol{V}_{i}^{-1} \text{diag}\left[
\frac{R_{ij}}{\phi_{ij}^1(\boldsymbol{\eta_0})\lambda_{i}^1(\boldsymbol{\gamma_0})}\right] (\boldsymbol{Y}_i - \boldsymbol{\mu}_i(\boldsymbol{\beta}_0, A_i))\right\} = \boldsymbol{0}$}
\label{appdx:B5}
    
\begin{align*}
    &\hspace{0.5cm} \mathbb{E}_{\boldsymbol{Y}_i, \boldsymbol{R}_i,\boldsymbol{C}_i, A_i, \boldsymbol{Z}_i, \boldsymbol{X}_i }\left\{\frac{\partial{\boldsymbol{\mu}_i(\boldsymbol{\beta}_0, A_i)}}{\partial{\boldsymbol{\beta}}}\boldsymbol{V}_{i}^{-1} \text{diag}\left[
    \frac{R_{ij}C_{i}}{\phi_{ij}^1(\boldsymbol{\eta_0})\lambda_{i}^1(\boldsymbol{\gamma_0})}\right] (\boldsymbol{Y}_i - \boldsymbol{\mu}_i(\boldsymbol{\beta}_0, A_i))\right\} \\
    &= \mathbb{E}_{\boldsymbol{Y}_i, A_i, \boldsymbol{Z}_i, \boldsymbol{X}_i}\left\{\mathbb{E}_{\boldsymbol{R}_i,\boldsymbol{C}_{i}|\boldsymbol{Y}_i, A_i, \boldsymbol{Z}_i, \boldsymbol{X}_i} \left\{ \frac{\partial{\boldsymbol{\mu}_i(\boldsymbol{\beta}_0, A_i)}}{\partial{\boldsymbol{\beta}}} \boldsymbol{V}_{i}^{-1}   \text{diag}\left[
    \frac{R_{ij}C_{i}}{\phi_{ij}^1(\boldsymbol{\eta_0})\lambda_{i}^1(\boldsymbol{\gamma_0})}\right] (\boldsymbol{Y}_i - \boldsymbol{\mu}_i(\boldsymbol{\beta}_0, A_i))\right\} \right\}  \\
    &= \mathbb{E}_{\boldsymbol{Y}_i,A_i, \boldsymbol{Z}_i, \boldsymbol{X}_i}\left\{ \frac{\partial{\boldsymbol{\mu}_i(\boldsymbol{\beta}_0, A_i)}}{\partial{\boldsymbol{\beta}}} \boldsymbol{V}_{i}^{-1}  \underbrace{\mathbb{E}_{\boldsymbol{R}_i, \boldsymbol{C}_{i}| \boldsymbol{X}_i,\boldsymbol{Z}_i, A_i} \left\{ \text{diag}\left[
    \frac{R_{ij}C_{i}}{\phi_{ij}^1(\boldsymbol{\eta_0})\lambda_{i}^1(\boldsymbol{\gamma_0})}\right]\right\}}_{= I} (\boldsymbol{Y}_i - \boldsymbol{\mu}_i(\boldsymbol{\beta}_0, A_i)) \right\}    \\
    &= \mathbb{E}_{\boldsymbol{Y}_i,A_i, \boldsymbol{Z}_i, \boldsymbol{X}_i} \left\{ \frac{\partial{\boldsymbol{\mu}_i(\boldsymbol{\beta}_0, A_i)}}{\partial{\boldsymbol{\beta}}}\boldsymbol{V}_{i}^{-1}  (\boldsymbol{Y}_i - \boldsymbol{\mu}_i(\boldsymbol{\beta}_0, A_i))\right\}  \\
    &= \mathbb{E}_{\boldsymbol{Y}_i, A_i} \left\{ \frac{\partial{\boldsymbol{\mu}_i(\boldsymbol{\beta}_0, A_i)}}{\partial{\boldsymbol{\beta}}}\boldsymbol{V}_{i}^{-1}  (\boldsymbol{Y}_i - \boldsymbol{\mu}_i(\boldsymbol{\beta}_0, A_i))\right\}  \\
    &= \mathbb{E}_{A_i} \left\{\mathbb{E}_{\boldsymbol{Y}_i|A_i} \left\{ \frac{\partial{\boldsymbol{\mu}_i(\boldsymbol{\beta}_0, A_i)}}{\partial{\boldsymbol{\beta}}}\boldsymbol{V}_{i}^{-1}  (\boldsymbol{Y}_i - \boldsymbol{\mu}_i(\boldsymbol{\beta}_0, A_i))\right\} \right\} \\
    &= \mathbb{E}_{A_i} \left\{ \frac{\partial{\boldsymbol{\mu}_i(\boldsymbol{\beta}_0, A_i)}}{\partial{\boldsymbol{\beta}}}\boldsymbol{V}_{i}^{-1}  \underbrace{\mathbb{E}_{\boldsymbol{Y}_i|A_i} \left\{\boldsymbol{Y}_i - \boldsymbol{\mu}_i(\boldsymbol{\beta}_0, A_i)\right\}}_{=\mathbf{0}} \right\} \\
    &= \boldsymbol{0}
\end{align*}

\renewcommand{\theequation}{C.\arabic{equation}}
\setcounter{equation}{0}
\section{Asymptotic Distribution of $\hat{\boldsymbol{\beta}}_{MR}$} \label{appdx:C}

This section proves the asymptotic normality of $\hat{\boldsymbol{\beta}}_{MR}$, which is summarized by Theorem \ref{thm:2} in Web Appendix \ref{appdx:C2}. Before proving Theorem \ref{thm:2}, we first establish Lemma \ref{lem:1} in Web Appendix \ref{appdx:C1}. Similar to the assumptions made in \ref{appdx:B}, we consider the setting with $n_i = n$ and $M \to \infty$. Furthermore, we use subscript asterisk to denote probability limits, $(\boldsymbol{\eta}_0, \boldsymbol{\gamma}_0)$ to denote the true parameters of the PS models, and $\boldsymbol{\beta}_0$ to denote the true parameters of the marginal model.

\subsection{Proof of Lemma 1}
\label{appdx:C1}
\begin{lemma}\label{lem:1}
When $\phi_{ij}^{1}(\boldsymbol{\eta}^{1})$ and $\lambda_{i}^{1}(\boldsymbol{\gamma}^{1})$ are the correctly specified models for $\phi_{ij}(\boldsymbol{\eta})$ and $\lambda_{i}(\boldsymbol{\gamma})$, respectively, we have:
\begin{align*}
M^{1 / 2} \hat{\boldsymbol{\varepsilon}}=& \boldsymbol{G}^{-1}\left(M^{-1 / 2} \sum_{i=1}^{M}\mathbbm{1}^{\mathrm{T}} \text{diag} \left\{\frac{R_{ij}C_{i}-\phi_{ij}^{1}\left(\boldsymbol{\eta}_{0}\right)\lambda_{i}^{1}\left(\boldsymbol{\gamma}_{0}\right)}{\phi_{ij}^{1}\left(\boldsymbol{\eta}_{0}\right)\lambda_{i}^{1}\left(\boldsymbol{\gamma}_{0}\right)} \boldsymbol{g}_{ij}\left(\boldsymbol{\eta}_{*}, \boldsymbol{\gamma}_{*}\right)\right\} \mathbbm{1} \right.\\
&\left.- \mathbbm{1}^{\mathrm{T}} \mathbb{E}\left[\text{diag} \left\{\frac{\boldsymbol{g}_{ij}\left(\boldsymbol{\eta}_{*}, \boldsymbol{\gamma}_{*}\right)}{\phi_{ij}^{1}\left(\boldsymbol{\eta}_{0}\right)}\left\{\frac{\partial \phi_{ij}^{1}\left(\boldsymbol{\eta}_{0}\right)}{\partial \boldsymbol{\eta}^{1}}\right\}^{\mathrm{T}} \right\}\right] \mathbbm{1} M^{1 / 2}\left(\hat{\boldsymbol{\eta}}^{1}-\boldsymbol{\eta}_{0}\right)
\right.\\
& \left.- \mathbbm{1}^{\mathrm{T}} \mathbb{E}\left[ \text{diag} \left\{ \frac{\boldsymbol{g}_{ij}\left(\boldsymbol{\eta}_{*}, \boldsymbol{\gamma}_{*}\right)}{\lambda_{i}^{1}\left(\boldsymbol{\gamma}_{0}\right)}\left\{\frac{\partial \lambda_{i}^{1}\left(\boldsymbol{\gamma}_{0}\right)}{\partial \boldsymbol{\gamma}^{1}}\right\}^{\mathrm{T}} \right\}\right] \mathbbm{1} M^{1 / 2}\left(\hat{\boldsymbol{\gamma}}^{1}-\boldsymbol{\gamma}_{0}\right) \right)+o_{p}(1),
\end{align*}
where $\boldsymbol{G} = \mathbbm{1}^{\mathrm{T}} \mathbb{E}\left[ \text{diag} \left\{ \frac{\boldsymbol{g}_{ij}\left(\boldsymbol{\eta}_{*}, \boldsymbol{\gamma}_{*}\right)^{\otimes 2}}{\phi_{ij}^{1}\left(\boldsymbol{\eta}_{0} \right) \lambda_{i}^{1}\left(\boldsymbol{\gamma}_{0} \right)} \right\}\right] \mathbbm{1}$, $\mathbbm{1} = (1, 1, \ldots, 1)^{\mathrm{T}}$, and $\hat{\boldsymbol{\varepsilon}}^{\mathrm{T}}=\left(\hat{\varepsilon}_{1}, \ldots, \hat{\varepsilon}_{K L}\right)$ is a $K L \times 1$ vector that solves the following equation:
$$
\sum_{i=1}^{s} \sum_{j=1}^{m_{i}} \frac{\hat{g}_{i j}(\hat{\boldsymbol{\eta}}, \hat{\boldsymbol{\gamma}}) /\left\{\phi_{i j}^{1}\left(\hat{\boldsymbol{\eta}}^{1}\right) \lambda_{i}^{1}\left(\hat{\boldsymbol{\gamma}}^{1}\right)\right\}}{1+\boldsymbol{\varepsilon}^{\mathrm{T}} \hat{g}_{i j}(\hat{\boldsymbol{\eta}}, \hat{\boldsymbol{\gamma}}) /\left\{\phi_{i j}^{1}\left(\hat{\boldsymbol{\eta}}^{1}\right) \lambda_{i}^{1}\left(\hat{\boldsymbol{\gamma}}^{1}\right)\right\}}=\mathbf{0}.
$$
\end{lemma}

\begin{proof}
For notational simplicity, let $N = \sum_{i=1}^{M} n_i$. Taking Taylor expansion of the left-hand side of Equation (\ref{eq:B5}) around $\left(\mathbf{0}^{\mathrm{T}}, \boldsymbol{\eta}_{*}^{\mathrm{T}}, \boldsymbol{\gamma}_{*}^{\mathrm{T}}\right)$ leads to:
\begin{align*}
\mathbf{0}
&= M^{-1/2}\sum_{i=1}^{s} \sum_{j=1}^{m_{i}} \frac{\hat{g}_{i j}(\hat{\boldsymbol{\eta}}, \hat{\boldsymbol{\gamma}}) /\left\{\phi_{i j}^{1}\left(\hat{\boldsymbol{\eta}}^{1}\right) \lambda_{i}^{1}\left(\hat{\boldsymbol{\gamma}}^{1}\right)\right\}}{1+\hat{\boldsymbol{\varepsilon}}^{\mathrm{T}} \hat{g}_{i j}(\hat{\boldsymbol{\eta}}, \hat{\boldsymbol{\gamma}}) /\left\{\phi_{i j}^{1}\left(\hat{\boldsymbol{\eta}}^{1}\right) \lambda_{i}^{1}\left(\hat{\boldsymbol{\gamma}}^{1}\right)\right\}} \\
&= M^{-1/2} \sum_{i=1}^{M} \sum_{j=1}^{n_{i}}  \frac{R_{ij}C_{i}\hat{g}_{i j}(\hat{\boldsymbol{\eta}}, \hat{\boldsymbol{\gamma}}) /\left\{\phi_{i j}^{1}\left(\hat{\boldsymbol{\eta}}^{1}\right) \lambda_{i}^{1}\left(\hat{\boldsymbol{\gamma}}^{1}\right)\right\}}{1+\hat{\boldsymbol{\varepsilon}}^{\mathrm{T}} \hat{g}_{i j}(\hat{\boldsymbol{\eta}}, \hat{\boldsymbol{\gamma}}) /\left\{\phi_{i j}^{1}\left(\hat{\boldsymbol{\eta}}^{1}\right) \lambda_{i}^{1}\left(\hat{\boldsymbol{\gamma}}^{1}\right)\right\}} \\
&= M^{-1/2} \sum_{i=1}^{M} \mathbbm{1}^{\mathrm{T}}
\text{diag} \left\{ \frac{R_{ij}C_{i}\hat{g}_{i j}(\hat{\boldsymbol{\eta}}, \hat{\boldsymbol{\gamma}}) /\left\{\phi_{i j}^{1}\left(\hat{\boldsymbol{\eta}}^{1}\right) \lambda_{i}^{1}\left(\hat{\boldsymbol{\gamma}}^{1}\right)\right\}}{1+\hat{\boldsymbol{\varepsilon}}^{\mathrm{T}} \hat{g}_{i j}(\hat{\boldsymbol{\eta}}, \hat{\boldsymbol{\gamma}}) /\left\{\phi_{i j}^{1}\left(\hat{\boldsymbol{\eta}}^{1}\right) \lambda_{i}^{1}\left(\hat{\boldsymbol{\gamma}}^{1}\right)\right\}} \right\} \mathbbm{1}\\
&= M^{-1/2} \sum_{i=1}^{M} \mathbbm{1}^{\mathrm{T}}
\text{diag} \left\{ \frac{R_{ij}C_{i}g_{i j}(\boldsymbol{\eta}_{*}, \boldsymbol{\gamma}_{*}) /\left\{\phi_{i j}^{1}\left(\boldsymbol{\eta}_{0}\right) \lambda_{i}^{1}\left(\boldsymbol{\gamma}_{0}\right)\right\}}{1+\hat{\boldsymbol{\varepsilon}}^{\mathrm{T}} g_{i j}(\boldsymbol{\eta}_{*},\boldsymbol{\gamma}_{*}) /\left\{\phi_{i j}^{1}\left(\boldsymbol{\eta}_{0}\right) \lambda_{i}^{1}\left(\boldsymbol{\gamma}_{0}\right)\right\}} \right\} \mathbbm{1} \\
&+\left\{\frac{1}{M} \sum_{i=1}^{M} \mathbbm{1}^{\mathrm{T}}
\text{diag} \left[  \frac{R_{ij}C_{i}}{\phi_{i j}^{1}\left(\boldsymbol{\eta}_{0}\right) \lambda_{i}^{1}\left(\boldsymbol{\gamma}_{0}\right)} \left\{ \frac{\partial g_{i j}(\boldsymbol{\eta}_{*}, \boldsymbol{\gamma}_{*})}{\partial \boldsymbol{\eta}^{1}}\right\} ^{\mathrm{T}} - \frac{R_{ij}C_{i} g_{i j}(\boldsymbol{\eta}_{*}, \boldsymbol{\gamma}_{*})}{\{\phi_{i j}^{1}\left(\boldsymbol{\eta}_{0}\right)\}^2 \lambda_{i}^{1}\left(\boldsymbol{\gamma}_{0}\right)} \left\{\frac{\partial\phi_{i j}^{1}\left(\boldsymbol{\eta}_{0}\right)  }{\partial \boldsymbol{\eta}^{1}}\right\} ^{\mathrm{T}} \right] \mathbbm{1} \right\} \\
&\times M^{1/2}(\hat{\boldsymbol{\eta}}^{1} - \boldsymbol{\eta}_{0} ) \\
&+  \left\{\frac{1}{M} \sum_{i=1}^{M} \mathbbm{1}^{\mathrm{T}}
\text{diag} \left[\frac{R_{ij}C_{i}}{\phi_{i j}^{1}\left(\boldsymbol{\eta}_{0}\right) \lambda_{i}^{1}\left(\boldsymbol{\gamma}_{0}\right)} \left\{ \frac{\partial g_{i j}(\boldsymbol{\eta}_{*}, \boldsymbol{\gamma}_{*})}{\partial \boldsymbol{\gamma}^{1}}\right\} ^{\mathrm{T}} - \frac{R_{ij}C_{i}g_{i j}(\boldsymbol{\eta}_{*}, \boldsymbol{\gamma}_{*})}{\phi_{ij }^{1}\left(\boldsymbol{\eta}_{0}\right) \left\{ \lambda_{i}^{1}\left(\boldsymbol{\gamma}_{0}\right)\right\}^2} \left\{ \frac{\partial\lambda_{i }^{1}\left(\boldsymbol{\gamma}_{0}\right)  }{\partial \boldsymbol{\gamma}^{1}}\right\} ^{\mathrm{T}} \right] \mathbbm{1} \right\} \\
&\times M^{1/2}(\hat{\boldsymbol{\gamma}}^{1} - \boldsymbol{\gamma}_{0} )  \\
&+ \sum_{k=2}^{K}\sum_{\ell=1}^{L}\left\{\frac{1}{M} \sum_{i=1}^{M} \mathbbm{1}^{\mathrm{T}}
\text{diag} \left[  \frac{R_{ij}C_{i}}{\phi_{i j}^{1}\left(\boldsymbol{\eta}_{0}\right) \lambda_{i}^{1}\left(\boldsymbol{\gamma}_{0}\right)} \left\{\frac{\partial g_{i j}(\boldsymbol{\eta}_{*}, \boldsymbol{\gamma}_{*})}{\partial \boldsymbol{\eta}_{*}^{k}}\right\} ^{\mathrm{T}} \right] \mathbbm{1} \right\}  M^{1/2}(\hat{\boldsymbol{\eta}}^{k} - \boldsymbol{\eta}_{*}^{k} )  \\
&+ \sum_{k=1}^{K}\sum_{\ell=2}^{L}\left\{\frac{1}{M} \sum_{i=1}^{M} \mathbbm{1}^{\mathrm{T}}
\text{diag} \left[ \frac{R_{ij}C_{i}}{\phi_{i j}^{1}\left(\boldsymbol{\eta}_{0}\right) \lambda_{i}^{1}\left(\boldsymbol{\gamma}_{0}\right)} \left\{\frac{\partial g_{i j}(\boldsymbol{\eta}_{*}, \boldsymbol{\gamma}_{*})}{\partial \boldsymbol{\gamma}^{\ell}} \right\}^{\mathrm{T}} \right] \mathbbm{1}\right\} M^{1/2}(\hat{\boldsymbol{\gamma}}^{\ell} - \boldsymbol{\gamma}_{*}^{\ell} ) + o_p(1)  \\
&= M^{-1/2} \sum_{i=1}^{M} \mathbbm{1}^{\mathrm{T}}
\text{diag} \left[ \frac{R_{ij}C_{i}g_{i j}(\boldsymbol{\eta}_{*}, \boldsymbol{\gamma}_{*})}{\phi_{i j}^{1}\left(\boldsymbol{\eta}_{0}\right) \lambda_{i}^{1}\left(\boldsymbol{\gamma}_{0}\right)} \right]\mathbbm{1} - \left\{\frac{1}{M} \sum_{i=1}^{M} \mathbbm{1}^{\mathrm{T}}
\text{diag} \left[   \frac{R_{ij}C_{i}g_{i j}(\boldsymbol{\eta}_{*}, \boldsymbol{\gamma}_{*})^{\otimes 2}}{\left\{\phi_{i j}^{1}\left(\boldsymbol{\eta}_{0}\right) \lambda_{i}^{1}\left(\boldsymbol{\gamma}_{0}\right)\right\}^2}\right]\mathbbm{1}\right\} M^{1/2}\hat{\boldsymbol{\varepsilon}}  \\
&+  \left\{\frac{1}{M}\sum_{i=1}^{M} \mathbbm{1}^{\mathrm{T}}
\text{diag} \left[   \frac{R_{ij}C_{i}}{\phi_{i j}^{1}\left(\boldsymbol{\eta}_{0}\right) \lambda_{i}^{1}\left(\boldsymbol{\gamma}_{0}\right)} \left( \left\{ \frac{\partial \phi_{i j}^{1}\left(\boldsymbol{\eta}_{0}\right) }{\partial \boldsymbol{\eta}^{1}} \lambda_{i }^{1}\left(\boldsymbol{\gamma}_{0}\right) \right\}^{\mathrm{T}} - \left\{ \frac{1}{N} \sum_{i=1}^{M} \sum_{j=1}^{n_i} \frac{\partial \phi_{i j}^{1}\left(\boldsymbol{\eta}_{0}\right) }{\partial \boldsymbol{\eta}^{1}} \lambda_{i }^{1}\left(\boldsymbol{\gamma}_{0}\right) \right\}^{\mathrm{T}}, \boldsymbol{0}  \right) \right. \right. \\
&-\left. \left. \frac{R_{ij}C_{i}g_{i j}(\boldsymbol{\eta}_{*}, \boldsymbol{\gamma}_{*})}{\left\{\phi_{ij }^{1}\left(\boldsymbol{\eta}_{0}\right)\right\}^2 \lambda_{i}^{1}\left(\boldsymbol{\gamma}_{0}\right)} \left\{\frac{\partial\phi_{ij }^{1}\left(\boldsymbol{\eta}_{0}\right)  }{\partial \boldsymbol{\eta}^{1}} \right\} ^{\mathrm{T}} \right] \mathbbm{1} \right\}  M^{1/2}(\hat{\boldsymbol{\eta}}^{1} - \boldsymbol{\eta}_{0} )  \\
&+  \left\{\frac{1}{M}\sum_{i=1}^{M} \mathbbm{1}^{\mathrm{T}}
\text{diag} \left[   \frac{R_{ij}C_{i}}{\phi_{i j}^{1}\left(\boldsymbol{\eta}_{0}\right) \lambda_{i}^{1}\left(\boldsymbol{\gamma}_{0}\right)} \left( \left\{ \phi_{ij }^{1}\left(\boldsymbol{\eta}_{0}\right)\frac{\partial \lambda_{i}^{1}\left(\boldsymbol{\gamma}_{0}\right) }{\partial \boldsymbol{\gamma}^{1}}  \right\}^{\mathrm{T}} - \left\{ \frac{1}{N} \sum_{i=1}^{M} \sum_{j=1}^{n_i} \phi_{ij }^{1}\left(\boldsymbol{\eta}_{0}\right)\frac{\partial \lambda_{i}^{1}\left(\boldsymbol{\gamma}_{0}\right) }{\partial \boldsymbol{\gamma}^{1}}  \right\}^{\mathrm{T}}, \boldsymbol{0}  \right) \right. \right. \\
&-\left. \left. \frac{R_{ij}C_{i} g_{i j}(\boldsymbol{\eta}_{*}, \boldsymbol{\gamma}_{*})}{\phi_{ij }^{1}\left(\boldsymbol{\eta}_{0}\right) \left\{ \lambda_{i}^{1}\left(\boldsymbol{\gamma}_{0}\right)\right\}^2} \left\{\frac{\partial\lambda_{i }^{1}\left(\boldsymbol{\gamma}_{0}\right)  }{\partial \boldsymbol{\gamma}^{1}} \right\} ^{\mathrm{T}} \right] \mathbbm{1} \right\}  M^{1/2}(\hat{\boldsymbol{\gamma}}^{1} - \boldsymbol{\gamma}_{0} )  \\
&+ \sum_{k=2}^{K} \sum_{\ell=1}^{L} \left\{\frac{1}{M}\sum_{i=1}^{M} \mathbbm{1}^{\mathrm{T}}
\text{diag} \left[  \frac{R_{ij}C_{i}}{\phi_{i j}^{1}\left(\boldsymbol{\eta}_{0}\right) \lambda_{i}^{1}\left(\boldsymbol{\gamma}_{0}\right)} \left(\boldsymbol{0}, \left\{ \frac{\partial \phi_{i j}^{k}\left(\boldsymbol{\eta}_{*}^{k}\right) }{\partial \boldsymbol{\eta}^{k}} \lambda_{i }^{\ell}\left(\boldsymbol{\gamma}_{*}^{\ell}\right) \right\}^{\mathrm{T}} \right. \right. \right. \\
&-\left. \left. \left.  \left\{ \frac{1}{N} \sum_{i=1}^{M} \sum_{j=1}^{n_i} \frac{\partial \phi_{i j}^{k}\left(\boldsymbol{\eta}_{*}^{k}\right) }{\partial \boldsymbol{\eta}^{k}} \lambda_{i }^{\ell}\left(\boldsymbol{\gamma}_{*}^{\ell}\right) \right\}^{\mathrm{T}}  \right) \right] \mathbbm{1} \right\}  M^{1/2}(\hat{\boldsymbol{\eta}}^{k} - \boldsymbol{\eta}_{*}^{k} )  \\
&+ \sum_{k=1}^{K} \sum_{\ell=2}^{L} \left\{\frac{1}{M}\sum_{i=1}^{M} \mathbbm{1}^{\mathrm{T}}
\text{diag} \left[  \frac{R_{ij}C_{i}}{\phi_{i j}^{1}\left(\boldsymbol{\eta}_{0}\right) \lambda_{i}^{1}\left(\boldsymbol{\gamma}_{0}\right)} \left (\boldsymbol{0}, \left\{ \phi_{ij }^{k}\left(\boldsymbol{\eta}_{*}^{k}\right) \frac{ \partial \lambda_{i }^{\ell}\left(\boldsymbol{\gamma}_{*}^{\ell}\right) }{\partial \boldsymbol{\gamma}^{\ell}} \right\}^{\mathrm{T}} \right. \right. \right. \\
&-\left. \left. \left. \left\{ \frac{1}{N} \sum_{i=1}^{M} \sum_{j=1}^{n_i} \phi_{ij }^{k}\left(\boldsymbol{\eta}_{*}^{k}\right) \frac{ \partial \lambda_{i }^{\ell}\left(\boldsymbol{\gamma}_{*}^{\ell}\right) }{\partial \boldsymbol{\gamma}^{\ell}} \right\}^{\mathrm{T}}  \right) \right] \mathbbm{1} \right\}  M^{1/2}(\hat{\boldsymbol{\gamma}}^{\ell} - \boldsymbol{\gamma}_{*}^{\ell} )  + o_p(1) \\
&= M^{-1/2} \sum_{i=1}^{M} \mathbbm{1}^{\mathrm{T}}
\text{diag} \left\{ \frac{R_{ij}C_{i} g_{i j}(\boldsymbol{\eta}_{*}, \boldsymbol{\gamma}_{*})}{\phi_{i j}^{1}\left(\boldsymbol{\eta}_{0}\right) \lambda_{i}^{1}\left(\boldsymbol{\gamma}_{0}\right)} \right\}\mathbbm{1} - \mathbbm{1}^{\mathrm{T}} \mathbb{E}\left[ \text{diag} \left\{ \frac{g_{i j}(\boldsymbol{\eta}_{*}, \boldsymbol{\gamma}_{*})^{\otimes 2}}{\phi_{i j}^{1}\left(\boldsymbol{\eta}_{0}\right) \lambda_{i}^{1}\left(\boldsymbol{\gamma}_{0}\right)} \right\}\right] \mathbbm{1} M^{1/2} \hat{\boldsymbol{\varepsilon}}  \\
&- \mathbbm{1}^{\mathrm{T}} \mathbb{E}\left[ \text{diag} \left\{ \frac{ g_{i j}\left(\boldsymbol{\eta}_{*}, \boldsymbol{\gamma}_{*}\right)}{\phi_{i j}^{1}\left(\boldsymbol{\eta}_{0}\right)}\left\{\frac{\partial \phi_{i j}^{1}\left(\boldsymbol{\eta}_{0}\right)}{\partial \boldsymbol{\eta}^{1}}\right\}^{\mathrm{T}} \right\} \right] \mathbbm{1} M^{1 / 2}\left(\hat{\boldsymbol{\eta}}^{1}-\boldsymbol{\eta}_{0}\right)  \\
&- \mathbbm{1}^{\mathrm{T}} \mathbb{E}\left[ \text{diag} \left\{ \frac{ g_{i j}\left(\boldsymbol{\eta}_{*}, \boldsymbol{\gamma}_{*}\right)}{\lambda_{i}^{1}\left(\boldsymbol{\gamma}_{0}\right)}\left\{\frac{\partial \lambda_{i }^{1}\left(\boldsymbol{\gamma}_{0}\right)}{\partial \boldsymbol{\gamma}^{1}}\right\}^{\mathrm{T}}\right\}\right] \mathbbm{1} M^{1 / 2}\left(\hat{\boldsymbol{\gamma}}^{1}-\boldsymbol{\gamma}_{0}\right) + o_p(1).
\end{align*}
Rearranging the last equation gives the result.
\end{proof}

\subsection{Proof of Theorem 2}
\label{appdx:C2}
\begin{theorem} \label{thm:2}
When both $\mathcal{P}_1$ and $\mathcal{P}_2$ contain a correctly specified model for $\lambda_{i}\left(\boldsymbol{\gamma}\right)$ and $\phi_{ij}\left(\boldsymbol{\eta}\right)$, respectively, $\sqrt{M}(\hat{\boldsymbol{\beta}}_{MR} - \boldsymbol{\beta}_0)$ has an asymptotic normal distribution with mean $\boldsymbol{0}$ and variance var($\boldsymbol{Z}_i$), where
\begin{align*}
    \boldsymbol{Z}_i = \left[E\left\{\frac{\partial \boldsymbol{U}_i\left(\boldsymbol{\beta}_0\right)}{\partial \boldsymbol{\beta}}\right\}\right]^{-1}\left[\boldsymbol{Q}_i-\left\{\mathbb{E}\left(\boldsymbol{Q}_i \boldsymbol{S}_{1,i}^{\mathrm{T}}\right)\right\}\left\{\mathbb{E}\left(\boldsymbol{S}_{1,i}^{\otimes 2}\right)\right\}^{-1} \boldsymbol{S}_{1,i}-\left\{\mathbb{E}\left(\boldsymbol{Q}_i \boldsymbol{S}_{2,i}^{\mathrm{T}}\right)\right\}\left\{\mathbb{E}\left(\boldsymbol{S}_{2,i}^{\otimes 2}\right)\right\}^{-1} \boldsymbol{S}_{2,i}\right].
\end{align*}
\end{theorem}
\begin{proof}
Recall that $\hat{\boldsymbol{\beta}}_{MR}$ is the solution to the estimating equation:
\begin{align*}
    \mathbf{0}
    &= \sum_{i=1}^{M} \text{diag} \left\{R_{ij}C_{i} \hat{w}_{ij}^{MR} \right\} \boldsymbol{U}_{i}\left(\boldsymbol{\beta}_{MR}\right).
\end{align*}
We take Taylor expansion of the above estimating equation around $\left(\boldsymbol{0}^{\mathrm{T}}, \boldsymbol{\eta}_{*}^{\mathrm{T}}, \boldsymbol{\gamma}_{*}^{\mathrm{T}},  \boldsymbol{\beta}_{0}^{\mathrm{T}} \right)$ and apply the asymptotic expansion of $\hat{\boldsymbol{\varepsilon}}$ from Lemma \ref{lem:1}, which yields:
\begin{align*}
        \mathbf{0}
        &= N M^{-1/2} \sum_{i=1}^{M} \text{diag} \left\{R_{ij} C_{i}\hat{w}_{ij}^{MR} \right\} \boldsymbol{U}_{i}\left(\hat{\boldsymbol{\beta}}_{MR}\right) \\
        &= N M^{-1/2} \sum_{i=1}^{M} \text{diag} \left\{R_{ij}C_{i} \hat{p}_{i j} \frac{\chi^{1,1}\left(\hat{\boldsymbol{\eta}}^{1}, \hat{\boldsymbol{\gamma}}^{1}\right)}{\left\{\phi_{i j}^{1}\left(\hat{\boldsymbol{\eta}}^{1}\right) \lambda_{i}^{1}\left(\hat{\boldsymbol{\gamma}}^{1}\right)\right\}} \right\} \boldsymbol{U}_{i}\left(\hat{\boldsymbol{\beta}}_{MR}\right) \\
        &= N M^{-1/2} \sum_{i=1}^{M}            \text{diag} \left\{R_{ij}           C_{i}\left(\frac{1}     {\sum_{i=1}^{s} m_{i}}\right) \frac{\chi^{1,1}\left(\hat{\boldsymbol{\eta}}^{1}, \hat{\boldsymbol{\gamma}}^{1}\right)/\left\{\phi_{i j}^{1}\left(\hat{\boldsymbol{\eta}}^{1}\right) \lambda_{i}^{1}\left(\hat{\boldsymbol{\gamma}}^{1}\right)\right\}}{1+\hat{\boldsymbol{\varepsilon}}^{T} \hat{g}_{i j}(\hat{\boldsymbol{\eta}}, \hat{\boldsymbol{\gamma}}) /\left\{\phi_{i j}^{1}\left(\hat{\boldsymbol{\eta}}^{1}\right) \lambda_{i}^{1}\left(\hat{\boldsymbol{\gamma}}^{1}\right)\right\}}  \right\} \boldsymbol{U}_{i}\left(\hat{\boldsymbol{\beta}}_{MR}\right) \\
        &= M^{-1/2} \left(\frac{N}{\sum_{i=1}^{s} m_{i}}\right) \left\{\frac{1}{N} \sum_{i=1}^{M} \sum_{j=1}^{n_i} \phi_{i j}^{1}\left(\boldsymbol{\eta}_{0}\right) \lambda_{i}^{1}\left(\boldsymbol{\gamma}_{0}\right) \right\} \sum_{i=1}^{M} \text{diag} \left\{ \frac{R_{ij}C_{i}}{\phi_{i j}^{1}\left(\boldsymbol{\eta}_{0}\right) \lambda_{i}^{1}\left(\boldsymbol{\gamma}_{0}\right)}  \right\} \boldsymbol{U}_{i}\left(\boldsymbol{\beta}_0\right) \\
        &-   M^{-1/2}\left(\frac{N}{\sum_{i=1}^{s} m_{i}}\right) \left\{\frac{1}{N} \sum_{i=1}^{M} \sum_{j=1}^{n_i} \phi_{i j}^{1}\left(\boldsymbol{\eta}_{0}\right) \lambda_{i}^{1}\left(\boldsymbol{\gamma}_{0}\right) \right\} \sum_{i=1}^{M} \text{diag} \left\{ \frac{R_{ij} C_{i}\hat{\boldsymbol{\varepsilon}} g_{i j}(\boldsymbol{\eta}_{*}, \boldsymbol{\gamma}_{*}) }{\left[\phi_{i j}^{1}\left(\boldsymbol{\eta}_{0}\right) \lambda_{i}^{1}\left(\boldsymbol{\gamma}_{0}\right)\right]^2}  \right\} \boldsymbol{U}_{i}\left(\boldsymbol{\beta}_0\right)\\
        &+ \left(\frac{N}{\sum_{i=1}^{s} m_{i}}\right) \left[ \sum_{i=1}^{M} \text{diag} \left\{ \frac{R_{ij} C_{i}}{\phi_{i j}^{1}\left(\boldsymbol{\eta}_{0}\right) \lambda_{i}^{1}\left(\boldsymbol{\gamma}_{0}\right)} \left\{\frac{1}{N} \sum_{i=1}^{M} \sum_{j=1}^{n_i} \frac{ \partial \phi_{i j}^{1}\left(\boldsymbol{\eta}_{0}\right) }{\partial \boldsymbol{\eta}^{1}} \lambda_{i}^{1}\left(\boldsymbol{\gamma}_{0}\right)  \right\}^\mathrm{T} \right. \right. \\
        &- \left. \left. \frac{R_{ij}C_{i}\chi^{1,1}\left(\boldsymbol{\eta}_0, \boldsymbol{\gamma}_0\right) }{\left\{\phi_{i j}^{1}\left(\boldsymbol{\eta}_{0}\right)\right\}^2  \lambda_{i}^{1}\left(\boldsymbol{\gamma}_{0}\right)} \left\{\frac{ \partial \phi_{i j}^{1}\left(\boldsymbol{\eta}_{0}\right)}{\partial \boldsymbol{\eta}^{1}} \right\}^{\mathrm{T}} \right\} \boldsymbol{U}_{i}\left(\boldsymbol{\beta}_0\right) \right]  M^{-1/2} (\hat{\boldsymbol{\eta}}^1 - \boldsymbol{\eta}_0) \\
        &+ \left(\frac{N}{\sum_{i=1}^{s} m_{i}}\right) \left[ \sum_{i=1}^{M} \text{diag} \left\{ \frac{R_{ij}C_{i}}{\phi_{i j}^{1}\left(\boldsymbol{\eta}_{0}\right) \lambda_{i}^{1}\left(\boldsymbol{\gamma}_{0}\right)} \left\{\frac{1}{N} \sum_{i=1}^{M} \sum_{j=1}^{n_i} \phi_{i j}^{1}\left(\boldsymbol{\eta}_{0}\right)\frac{ \partial  \lambda_{i}^{1}\left(\boldsymbol{\gamma}_{0}\right)}{\partial \boldsymbol{\gamma}^{1}} \right\}^\mathrm{T} \right. \right. \\
        &- \left. \left. \frac{R_{ij} C_{i}\chi^{1,1}\left(\boldsymbol{\eta}_0, \boldsymbol{\gamma}_0\right) }{\phi_{i j}^{1}\left(\boldsymbol{\eta}_{0}\right) \left\{\lambda_{i}^{1}\left(\boldsymbol{\gamma}_{0}\right)\right\}^2 } \left\{\frac{ \partial \lambda_{i}^{1}\left(\boldsymbol{\gamma}_{0}\right)}{\partial \boldsymbol{\gamma}^{1}} \right\}^{\mathrm{T}} \right\} \boldsymbol{U}_{i}\left(\boldsymbol{\beta}_0\right) \right]  M^{-1/2} (\hat{\boldsymbol{\gamma}}^1 - \boldsymbol{\gamma}_0) \\
        &+ \left(\frac{N}{\sum_{i=1}^{s} m_{i}}\right) \left\{\frac{1}{N} \sum_{i=1}^{M} \sum_{j=1}^{n_i} \phi_{i j}^{1}\left(\boldsymbol{\eta}_{0}\right) \lambda_{i}^{1}\left(\boldsymbol{\gamma}_{0}\right) \right\} \left[ \sum_{i=1}^{M} \text{diag} \left\{ \frac{R_{ij}C_{i}}{\phi_{i j}^{1}\left(\boldsymbol{\eta}_{0}\right) \lambda_{i}^{1}\left(\boldsymbol{\gamma}_{0}\right)}  \right\} \frac{\partial \boldsymbol{U}_{i}\left(\boldsymbol{\beta}_0\right)}{\partial \boldsymbol{\beta}} \right] \\
        &\times M^{-1/2} (\hat{\boldsymbol{\beta}}_{MR} - \boldsymbol{\beta}_0) + o_p(1) \\
        &= M^{-1/2} \sum_{i=1}^{M} \text{diag} \left\{ \frac{R_{ij}C_{i}}{\phi_{i j}^{1}\left(\boldsymbol{\eta}_{0}\right) \lambda_{i}^{1}\left(\boldsymbol{\gamma}_{0}\right)}  \right\} \boldsymbol{U}_{i}\left(\boldsymbol{\beta}_0\right) -  \underbrace{\mathbb{E}\left[  \text{diag} \left\{\frac{ g_{ij}(\boldsymbol{\eta}_{*}, \boldsymbol{\gamma}_{*}) }{\phi_{i j}^{1}\left(\boldsymbol{\eta}_{0}\right) \lambda_{i}^{1}\left(\boldsymbol{\gamma}_{0}\right)}  \right\} \boldsymbol{U}_{i}\left(\boldsymbol{\beta}_0\right) \right]}_{=\boldsymbol{L}}  M^{1/2}\hat{\boldsymbol{\varepsilon}} \\
        &-  \mathbb{E}\left[ \text{diag} \left\{\frac{1}{ \phi_{ij}^{1}\left(\boldsymbol{\eta}_{0}\right)} \left\{\frac{ \partial \phi_{i j}^{1}\left(\boldsymbol{\eta}_{0}\right)}{\partial \boldsymbol{\eta}^{1}} \right\}^{\mathrm{T}} \right\} \boldsymbol{U}_{i}\left(\boldsymbol{\beta}_0\right)   \right] M^{1/2} (\hat{\boldsymbol{\eta}}^1 - \boldsymbol{\eta}_0)  \\
        &-\mathbb{E} \left[ \text{diag} \left\{ \frac{1}{\lambda_{i}^{1}\left(\boldsymbol{\gamma}_{0}\right)} \left\{\frac{ \partial \lambda_{i}^{1}\left(\boldsymbol{\gamma}_{0}\right)}{\partial \boldsymbol{\gamma}^{1}} \right\}^{\mathrm{T}} \right\} \boldsymbol{U}_{i}\left(\boldsymbol{\beta}_0\right)  \right] M^{1/2} (\hat{\boldsymbol{\gamma}}^1 - \boldsymbol{\gamma}_0) \\
        &+ \mathbb{E}\left[\frac{\partial\boldsymbol{U}_{i}\left(\boldsymbol{\beta}_0\right)}{\partial \boldsymbol{\beta}} \right]M^{1/2} (\hat{\boldsymbol{\beta}}_{MR} - \boldsymbol{\beta}_0) + o_p(1) \\
        &= M^{-1/2} \sum_{i=1}^{M} \text{diag} \left\{ \frac{R_{ij}C_{i}}{\phi_{i j}^{1}\left(\boldsymbol{\eta}_{0}\right) \lambda_{i}^{1}\left(\boldsymbol{\gamma}_{0}\right)}  \right\} \boldsymbol{U}_{i}\left(\boldsymbol{\beta}_0\right)  \\ 
        &-\mathbf{L} \mathbf{G}^{-1}M^{-1 / 2} \sum_{i=1}^{M}\mathbbm{1}^{\mathrm{T}} \text{diag} \left\{\frac{R_{ij}C_{i}-\phi_{ij}^{1}\left(\boldsymbol{\eta}_{0}\right)\lambda_{i}^{1}\left(\boldsymbol{\gamma}_{0}\right)}{\phi_{ij}^{1}\left(\boldsymbol{\eta}_{0}\right)\lambda_{i}^{1}\left(\boldsymbol{\gamma}_{0}\right)} \boldsymbol{g}_{ij}\left(\boldsymbol{\eta}_{*}, \boldsymbol{\gamma}_{*}\right)\right\} \mathbbm{1}\\
        &-\mathbf{L}\mathbf{G}^{-1}\mathbbm{1}^{\mathrm{T}} \mathbb{E}\left[\text{diag} \left\{\frac{\boldsymbol{g}_{ij}\left(\boldsymbol{\eta}_{*}, \boldsymbol{\gamma}_{*}\right)}{\phi_{ij}^{1}\left(\boldsymbol{\eta}_{0}\right)}\left\{\frac{\partial \phi_{ij}^{1}\left(\boldsymbol{\eta}_{0}\right)}{\partial \boldsymbol{\eta}^{1}}\right\}^{\mathrm{T}} \right\}\right] \mathbbm{1} M^{1 / 2}\left(\hat{\boldsymbol{\eta}}^{1}-\boldsymbol{\eta}_{0}\right) \\
        &-\mathbf{L}\mathbf{G}^{-1} \mathbbm{1}^{\mathrm{T}} \mathbb{E}\left[\text{diag} \left\{\frac{\boldsymbol{g}_{ij}\left(\boldsymbol{\eta}_{*}, \boldsymbol{\gamma}_{*}\right)}{\phi_{ij}^{1}\left(\boldsymbol{\eta}_{0}\right)}\left\{\frac{\partial \phi_{ij}^{1}\left(\boldsymbol{\eta}_{0}\right)}{\partial \boldsymbol{\eta}^{1}}\right\}^{\mathrm{T}} \right\}\right] \mathbbm{1} M^{1 / 2}\left(\hat{\boldsymbol{\eta}}^{1}-\boldsymbol{\eta}_{0} \right) \\
         &-  \mathbb{E}\left[ \text{diag} \left\{\frac{1}{ \phi_{ij}^{1}\left(\boldsymbol{\eta}_{0}\right)} \left\{\frac{ \partial \phi_{i j}^{1}\left(\boldsymbol{\eta}_{0}\right)}{\partial \boldsymbol{\eta}^{1}} \right\}^{\mathrm{T}} \right\} \boldsymbol{U}_{i}\left(\boldsymbol{\beta}_0\right)   \right] M^{1/2} (\hat{\boldsymbol{\eta}}^1 - \boldsymbol{\eta}_0)  \\
        &-\mathbb{E} \left[ \text{diag} \left\{ \frac{1}{\lambda_{i}^{1}\left(\boldsymbol{\gamma}_{0}\right)} \left\{\frac{ \partial \lambda_{i}^{1}\left(\boldsymbol{\gamma}_{0}\right)}{\partial \boldsymbol{\gamma}^{1}} \right\}^{\mathrm{T}} \right\} \boldsymbol{U}_{i}\left(\boldsymbol{\beta}_0\right)  \right] M^{1/2} (\hat{\boldsymbol{\gamma}}^1 - \boldsymbol{\gamma}_0) \\
        &+ \mathbb{E}\left[\frac{\partial \boldsymbol{U}_{i}\left(\boldsymbol{\beta}_0\right)}{\partial \boldsymbol{\beta}} \right]M^{1/2} (\hat{\boldsymbol{\beta}}_{MR} - \boldsymbol{\beta}_0) + o_p(1) \\
        &= M^{-1/2} \sum_{i=1}^{M} \underbrace{\text{diag} \left\{ \frac{R_{ij}C_{i}}{\phi_{i j}^{1}\left(\boldsymbol{\eta}_{0}\right) \lambda_{i}^{1}\left(\boldsymbol{\gamma}_{0}\right)}  \right\} \boldsymbol{U}_{i}\left(\boldsymbol{\beta}_0\right)  -\mathbf{L} \mathbf{G}^{-1}\mathbbm{1}^{\mathrm{T}} \text{diag} \left\{\frac{R_{ij}C_{i}-\phi_{ij}^{1}\left(\boldsymbol{\eta}_{0}\right)\lambda_{i}^{1}\left(\boldsymbol{\gamma}_{0}\right)}{\phi_{ij}^{1}\left(\boldsymbol{\eta}_{0}\right)\lambda_{i}^{1}\left(\boldsymbol{\gamma}_{0}\right)} \boldsymbol{g}_{ij}\left(\boldsymbol{\eta}_{*}, \boldsymbol{\gamma}_{*}\right)\right\} \mathbbm{1}}_{=\boldsymbol{Q}_i(\boldsymbol{\gamma}_{0}, \boldsymbol{\eta}_{0})}\\
        &- \mathbb{E}\left[\text{diag} \left\{\frac{1}{ \phi_{ij}^{1}\left(\boldsymbol{\eta}_{0}\right)} \left\{\frac{ \partial \phi_{i j}^{1}\left(\boldsymbol{\eta}_{0}\right)}{\partial \boldsymbol{\eta}^{1}} \right\}^{\mathrm{T}} \right\} \boldsymbol{U}_{i}\left(\boldsymbol{\beta}_0\right)  + \mathbf{L}\mathbf{G}^{-1}\mathbbm{1}^{\mathrm{T}} \text{diag} \left\{\frac{\boldsymbol{g}_{ij}\left(\boldsymbol{\eta}_{*}, \boldsymbol{\gamma}_{*}\right)}{\phi_{ij}^{1}\left(\boldsymbol{\eta}_{0}\right)}\left\{\frac{\partial \phi_{ij}^{1}\left(\boldsymbol{\eta}_{0}\right)}{\partial \boldsymbol{\eta}^{1}}\right\}^{\mathrm{T}} \right\}\mathbbm{1}\right]  \\
        &\times M^{1 / 2}\left(\hat{\boldsymbol{\eta}}^{1}-\boldsymbol{\eta}_{0}\right) \\
        &- \mathbb{E}\left[\text{diag} \left\{ \frac{1}{\lambda_{i}^{1}\left(\boldsymbol{\gamma}_{0}\right)} \left\{\frac{ \partial \lambda_{i}^{1}\left(\boldsymbol{\gamma}_{0}\right)}{\partial \boldsymbol{\gamma}^{1}} \right\}^{\mathrm{T}} \right\} \boldsymbol{U}_{i}\left(\boldsymbol{\beta}_0\right) + \mathbf{L}\mathbf{G}^{-1} \mathbbm{1}^{\mathrm{T}}\text{diag} \left\{\frac{\boldsymbol{g}_{ij}\left(\boldsymbol{\eta}_{*}, \boldsymbol{\gamma}_{*}\right)}{\phi_{ij}^{1}\left(\boldsymbol{\eta}_{0}\right)}\left\{\frac{\partial \phi_{ij}^{1}\left(\boldsymbol{\eta}_{0}\right)}{\partial \boldsymbol{\eta}^{1}}\right\}^{\mathrm{T}} \right\}\mathbbm{1}\right]  \\
        &\times M^{1 / 2}\left(\hat{\boldsymbol{\gamma}}^{1}-\boldsymbol{\gamma}_{0} \right) \\
        &+ \mathbb{E}\left[\frac{\partial \boldsymbol{U}_{i}\left(\boldsymbol{\beta}_0\right)}{\partial \boldsymbol{\beta}} \right]M^{1/2} (\hat{\boldsymbol{\beta}}_{MR} - \boldsymbol{\beta}_0) + o_p(1) \\
        &= M^{-1/2} \sum_{i=1}^{M} \boldsymbol{Q}_i(\boldsymbol{\gamma}_{0}, \boldsymbol{\eta}_{0}) \\
        &- \mathbb{E}\left[\text{diag} \left\{\frac{1}{ \phi_{ij}^{1}\left(\boldsymbol{\eta}_{0}\right)} \left\{\frac{ \partial \phi_{i j}^{1}\left(\boldsymbol{\eta}_{0}\right)}{\partial \boldsymbol{\eta}^{1}} \right\}^{\mathrm{T}} \right\} \boldsymbol{U}_{i}\left(\boldsymbol{\beta}_0\right)  + \mathbf{L}\mathbf{G}^{-1}\mathbbm{1}^{\mathrm{T}} \text{diag} \left\{\frac{\boldsymbol{g}_{ij}\left(\boldsymbol{\eta}_{*}, \boldsymbol{\gamma}_{*}\right)}{\phi_{ij}^{1}\left(\boldsymbol{\eta}_{0}\right)}\left\{\frac{\partial \phi_{ij}^{1}\left(\boldsymbol{\eta}_{0}\right)}{\partial \boldsymbol{\eta}^{1}}\right\}^{\mathrm{T}} \right\}\mathbbm{1}\right]  \\
        &\times M^{-1 / 2}\left\{E\left(\boldsymbol{S}_{1}^{\otimes 2}\right)\right\}^{-1} \sum_{i=1}^{M}\boldsymbol{S}_{1,i} \\
        &- \mathbb{E}\left[\text{diag} \left\{ \frac{1}{\lambda_{i}^{1}\left(\boldsymbol{\gamma}_{0}\right)} \left\{\frac{ \partial \lambda_{i}^{1}\left(\boldsymbol{\gamma}_{0}\right)}{\partial \boldsymbol{\gamma}^{1}} \right\}^{\mathrm{T}} \right\} \boldsymbol{U}_{i}\left(\boldsymbol{\beta}_0\right) + \mathbf{L}\mathbf{G}^{-1} \mathbbm{1}^{\mathrm{T}}\text{diag} \left\{\frac{\boldsymbol{g}_{ij}\left(\boldsymbol{\eta}_{*}, \boldsymbol{\gamma}_{*}\right)}{\phi_{ij}^{1}\left(\boldsymbol{\eta}_{0}\right)}\left\{\frac{\partial \phi_{ij}^{1}\left(\boldsymbol{\eta}_{0}\right)}{\partial \boldsymbol{\eta}^{1}}\right\}^{\mathrm{T}} \right\}\mathbbm{1}\right]  \\
       &\times M^{-1 / 2}\left\{E\left(\boldsymbol{S}_{2}^{\otimes 2}\right)\right\}^{-1} \sum_{i=1}^{M} \boldsymbol{S}_{2,i} \\  &+ \mathbb{E}\left[\frac{\partial \boldsymbol{U}_{i}\left(\boldsymbol{\beta}_0\right)}{\partial \boldsymbol{\beta}} \right]M^{1/2} (\hat{\boldsymbol{\beta}}_{MR} - \boldsymbol{\beta}_0) + o_p(1).
    \end{align*}
    The last equality holds by taking the Taylor expansion of the score $\boldsymbol{S}_{1,i}(\hat{\boldsymbol{\eta}}^1)$ and $\boldsymbol{S}_{2,i}(\hat{\boldsymbol{\gamma}}^1)$ around $\boldsymbol{S}_{1,i}(\boldsymbol{\eta}_0)$ and $\boldsymbol{S}_{2,i}(\boldsymbol{\gamma}_0)$, respectively, which leads to:
    \begin{align*}
        M^{1 / 2}\left(\hat{\boldsymbol{\eta}}^{1}-\boldsymbol{\eta}_{0}\right) 
        &= M^{-1 / 2}\left\{\mathbb{E}\left(\boldsymbol{S}_{1,i}^{\otimes 2}\right)\right\}^{-1} \sum_{i=1}^{M} \boldsymbol{S}_{1,i} + o_p(1), \\
        M^{1 / 2}\left(\hat{\boldsymbol{\gamma}}^{1}-\boldsymbol{\gamma}_{0}\right) 
        &= M^{-1 / 2}\left\{\mathbb{E}\left(\boldsymbol{S}_{2,i}^{\otimes 2}\right)\right\}^{-1} \sum_{i=1}^{M}\boldsymbol{S}_{2,i} + o_p(1).
    \end{align*}
    With additional algebra manipulation and applying the generalized information equality, i.e. Lemma 9.1 of \citet{tsiatis2006semiparametric}, we have first,
    \begin{align*}
        &\hspace{0.5cm} \mathbb{E}\left[\text{diag} \left\{\frac{1}{ \phi_{ij}^{1}\left(\boldsymbol{\eta}_{0}\right)} \left\{\frac{ \partial \phi_{i j}^{1}\left(\boldsymbol{\eta}_{0}\right)}{\partial \boldsymbol{\eta}^{1}} \right\}^{\mathrm{T}} \right\} \boldsymbol{U}_{i}\left(\boldsymbol{\beta}_0\right)  + \mathbf{L}\mathbf{G}^{-1}\mathbbm{1}^{\mathrm{T}} \text{diag} \left\{\frac{\boldsymbol{g}_{ij}\left(\boldsymbol{\eta}_{*}, \boldsymbol{\gamma}_{*}\right)}{\phi_{ij}^{1}\left(\boldsymbol{\eta}_{0}\right)}\left\{\frac{\partial \phi_{ij}^{1}\left(\boldsymbol{\eta}_{0}\right)}{\partial \boldsymbol{\eta}^{1}}\right\}^{\mathrm{T}} \right\}\mathbbm{1}\right]  \\
        &= -\mathbb{E}\left\{\frac{\partial \boldsymbol{Q}_i\left(\boldsymbol{\gamma}_{0}, \boldsymbol{\eta}_{0}\right)}{\partial \boldsymbol{\eta}^{1}}\right\}  \\
        &= \mathbb{E}\left(\boldsymbol{Q}_i \boldsymbol{S}_{1,i}^{\mathrm{T}}\right),
    \end{align*}
        and second, 
    \begin{align*}
        &\hspace{0.5cm} \mathbb{E}\left[\text{diag} \left\{ \frac{1}{\lambda_{i}^{1}\left(\boldsymbol{\gamma}_{0}\right)} \left\{\frac{ \partial \lambda_{i}^{1}\left(\boldsymbol{\gamma}_{0}\right)}{\partial \boldsymbol{\gamma}^{1}} \right\}^{\mathrm{T}} \right\} \boldsymbol{U}_{i}\left(\boldsymbol{\beta}_0\right) + \mathbf{L}\mathbf{G}^{-1} \mathbbm{1}^{\mathrm{T}}\text{diag} \left\{\frac{\boldsymbol{g}_{ij}\left(\boldsymbol{\eta}_{*}, \boldsymbol{\gamma}_{*}\right)}{\phi_{ij}^{1}\left(\boldsymbol{\eta}_{0}\right)}\left\{\frac{\partial \phi_{ij}^{1}\left(\boldsymbol{\eta}_{0}\right)}{\partial \boldsymbol{\eta}^{1}}\right\}^{\mathrm{T}} \right\}\mathbbm{1}\right] \\
        &= -\mathbb{E}\left\{\frac{\partial \boldsymbol{Q}_i\left(\boldsymbol{\gamma}_{0}, \boldsymbol{\eta}_{0}\right)}{\partial \boldsymbol{\gamma}^{1}}\right\}  \\
        &= \mathbb{E}\left(\boldsymbol{Q}_i \boldsymbol{S}_{2,i}^{\mathrm{T}}\right).
    \end{align*}
    Thus, we have 
    \begin{align*}
    \mathbf{0}=& M^{-1 / 2} \sum_{i=1}^{M}\left[\boldsymbol{Q}_{i}-\left\{\mathbb{E}\left(\boldsymbol{Q}_i \boldsymbol{S}_{1,i}^{\mathrm{T}}\right)\right\}\left\{\mathbb{E}\left(\boldsymbol{S}_{1,i}^{\otimes 2}\right)\right\}^{-1} \boldsymbol{S}_{1,i}-\left\{\mathbb{E}\left(\boldsymbol{Q}_i \boldsymbol{S}_{2,i}^{\mathrm{T}}\right)\right\}\left\{\mathbb{E}\left(\boldsymbol{S}_{2,i}^{\otimes 2}\right)\right\}^{-1} \boldsymbol{S}_{2,i}\right] \\
    &+\mathbb{E}\left\{\frac{\partial \boldsymbol{U}_i\left(\boldsymbol{\beta}_0\right)}{\partial \boldsymbol{\beta}}\right\} M^{1 / 2}\left(\hat{\boldsymbol{\beta}}_{\mathrm{MR}}-\boldsymbol{\beta}_{0}\right)+o_{p}(1)
    \end{align*}
    Rearranging the above equation gives the results. 
\end{proof}

\renewcommand{\theequation}{D.\arabic{equation}}
\setcounter{equation}{0}
\section{Derivation and Implementation of the EM algorithm}\label{appdx:D}

\subsection{Derivation of the EM algorithm}
\label{appdx:D1}

\begin{table}[h]
\caption{Different patterns of $C_i^O$ and $C_i$ with cluster size of 3. $C_i^O$ is the observed cluster-level missingness indicators and $C_i$ is the true cluster-level missingness indicators induced by the cluster-level missingness process. Each pattern illustrates different combination of $(\boldsymbol{R}_{i}, C_i^O, C_i)$. Because the observed data only gives information about $(\boldsymbol{R}_{i}, C_i^O)$, one cannot distinguish pattern 1 from pattern 2. This leads to potential misclassification of $C_i$, especially under small cluster size.}
\label{tab:s1}
\small
\begin{center}
\begin{tabular}{|c|ccc|}
\hline
Pattern & $R_{ij}$ & $C_i^O$ & $C_i$ \\ 
\hline
    & 0 & 0 & 0 \\
1   & 0 & 0 & 0 \\
    & 0 & 0 & 0 \\ 
\hline
    & 0 & 0 & 1 \\
2   & 0 & 0 & 1 \\
    & 0 & 0 & 1 \\
\hline
    & 1 & 1 & 1 \\
3   & 0 & 1 & 1 \\
    & 0 & 1 & 1 \\
\hline
\end{tabular}
\end{center}
\end{table}

Table \ref{tab:s1} summarizes different patterns of $C_i^O$ and $C_i$. To motivate the derivation for the EM algorithm, we consider partitioning the observed data into two groups: the first group contains clusters with $C_i^O = 1$ and the second group contains clusters with $C_i^O = 0$. When $C_i^O = 1$, we must have $C_i = 1$. But when $C_i^O = 0$, the true $C_i$ can either be 1 or 0. Hence, for cluster in the first group, the observed $C_i^O$ is always equal to the true $C_i$. For cluster in the second group, however, $C_i^O$ can be potentially misclassified. We treat $C_i$ as partially observed data. The ``complete'' data is then ($\boldsymbol{C}, \boldsymbol{R}, \boldsymbol{Y}, \boldsymbol{A}, \boldsymbol{Z}, \boldsymbol{X}$). The complete data likelihood is:
\begin{align*}
    &\hspace{0.3cm} L_c(\boldsymbol{C}, \boldsymbol{R}, \boldsymbol{Y}, \boldsymbol{A}, \boldsymbol{Z}, \boldsymbol{X}; \boldsymbol{\eta}, \boldsymbol{\gamma} ) \\
    &= L(\boldsymbol{C}, \boldsymbol{R} \mid \boldsymbol{Y}, \boldsymbol{A}, \boldsymbol{Z}, \boldsymbol{X}; \boldsymbol{\eta}, \boldsymbol{\gamma} ) L(\boldsymbol{Y}, \boldsymbol{A}, \boldsymbol{Z}, \boldsymbol{X}) \\
    &\propto L(\boldsymbol{C}, \boldsymbol{R} \mid \boldsymbol{Y}, \boldsymbol{A}, \boldsymbol{Z}, \boldsymbol{X}; \boldsymbol{\eta}, \boldsymbol{\gamma} )  \\
    &= L(\boldsymbol{C}, \boldsymbol{R} \mid \boldsymbol{A}, \boldsymbol{Z}, \boldsymbol{X}; \boldsymbol{\eta}, \boldsymbol{\gamma} ) \\
    &= \prod_{i=1}^{M}  L(\boldsymbol{R}_i, C_i \mid A_i, \boldsymbol{Z}_{i}, \boldsymbol{X}_{i}; \boldsymbol{\eta}, \boldsymbol{\gamma}) \\
    &= \prod_{i: C_i^O = 1}  L(\boldsymbol{R}_i \mid C_i, A_i, \boldsymbol{Z}_{i}, \boldsymbol{X}_{i}; \boldsymbol{\eta}, \boldsymbol{\gamma})L( C_i \mid A_i, \boldsymbol{Z}_{i}; \boldsymbol{\eta}, \boldsymbol{\gamma}) \\
    &\hspace{0.4cm} \prod_{i: C_i^O = 0} L(\boldsymbol{R}_i \mid C_i, A_i, \boldsymbol{Z}_{i}, \boldsymbol{X}_{i}; \boldsymbol{\eta}, \boldsymbol{\gamma})L( C_i \mid A_i, \boldsymbol{Z}_{i}; \boldsymbol{\eta}, \boldsymbol{\gamma}) \\
    &= \prod_{i: C_i^O = 1}  P(\boldsymbol{R}_i = \boldsymbol{r}_i \mid C_i = 1, A_i, \boldsymbol{Z}_{i}, \boldsymbol{X}_{i}; \boldsymbol{\eta})^{I(\boldsymbol{R}_i = \boldsymbol{r}_i)} P(C_i = 1 \mid A_i, \boldsymbol{Z}_{i}, ; \boldsymbol{\gamma}) \\
    &\hspace{0.4cm}
    \prod_{i: C_i^O = 0} {\underbrace{P(\boldsymbol{R}_i = \boldsymbol{0} \mid C_i = 0, A_i, \boldsymbol{Z}_{i}, \boldsymbol{X}_{i}; \boldsymbol{\eta})}_{= 1}}^{(1 -C_i)}
    P(C_i = 0 \mid A_i, \boldsymbol{Z}_{i}, ; \boldsymbol{\gamma})^{(1 -C_i)} \\ &\hspace{1.5cm} P(\boldsymbol{R}_i = \boldsymbol{0} \mid C_i = 1, A_i, \boldsymbol{Z}_{i}, \boldsymbol{X}_{i}; \boldsymbol{\eta})^{C_i} P(C_i = 1 \mid A_i, \boldsymbol{Z}_{i}, ; \boldsymbol{\gamma})^{C_i} \\
    &= \prod_{i: C_i^O = 1} \left \{ \prod_{j=1}^{n_i} P(R_{ij} = 1 \mid C_i = 1, A_i, \boldsymbol{Z}_{i}, \boldsymbol{X}_{i}; \boldsymbol{\eta})^{R_{ij}} P(R_{ij} = 0 \mid C_i = 1, A_i, \boldsymbol{Z}_{i}, \boldsymbol{X}_{i}; \boldsymbol{ \eta})^{1 - R_{ij}} \right\} \\ 
    &\hspace{1.6cm}  P(C_i = 1 \mid A_i, \boldsymbol{Z}_{i}, ; \boldsymbol{\gamma}) \\
    &\hspace{0.4cm} \prod_{i: C_i^O = 0} P(C_i = 0 \mid A_i, \boldsymbol{Z}_{i}, ; \boldsymbol{\gamma})^{(1 -C_i)} \left \{ \prod_{j=1}^{n_i} P(R_{ij} = 0 \mid C_i = 1, A_i, \boldsymbol{Z}_{i}, \boldsymbol{X}_{i}; \boldsymbol{\eta}) \right\}^{C_i}  \\
    &\hspace{1.5cm} P(C_i = 1 \mid A_i, \boldsymbol{Z}_{i}, ; \boldsymbol{\gamma})^{C_i} \\ 
    &= \prod_{i=1}^{M} \left\{ \prod_{j=1}^{n_i} P(R_{ij} = 1 \mid C_i = 1, A_i, \boldsymbol{Z}_{i}, \boldsymbol{X}_{i}; \boldsymbol{\eta})^{ R_{ij}}  P(R_{ij} = 0 \mid C_i = 1, A_i, \boldsymbol{Z}_{i}, \boldsymbol{X}_{i}; \boldsymbol{\eta})^{(1 - R_{ij})}  \right \}^{C_i^O} \\
    &\hspace{1.2cm}  P(C_i = 1 \mid A_i, \boldsymbol{Z}_{i}, ; \boldsymbol{\gamma})^{C_i^O}  P(C_i = 0 \mid A_i, \boldsymbol{Z}_{i}, ; \boldsymbol{\gamma})^{(1 -C_i)(1 - C_i^O)} \\
    &\hspace{1cm} \left \{ \prod_{j=1}^{n_i} P(R_{ij} = 0 \mid C_i = 1, A_i, \boldsymbol{Z}_{i}, \boldsymbol{X}_{i}; \boldsymbol{\eta}) \right\}^{C_i(1 - C_i^O)}  P(C_i = 1 \mid A_i, \boldsymbol{Z}_{i}, ; \boldsymbol{\gamma})^{C_i(1 - C_i^O)} \\
    &= \prod_{i=1}^{M}  \left\{ \prod_{j=1}^{n_i} \left[
    \operatorname{expit} \left( \boldsymbol{X}_{i j}^{*} \boldsymbol{\eta}\right)^{R_{ij}}  \left(1 - \operatorname{expit}\left(\boldsymbol{X}_{i j}^{*} \boldsymbol{\eta}\right) \right)^{(1 - R_{ij})} \right] \right \}^{C_i^O}\operatorname{expit}\left(\boldsymbol{Z}_{i}^{*} \boldsymbol{\gamma}\right)^{C_i^O}\\
    &\hspace{1cm}  \left(1 - \operatorname{expit}\left(\boldsymbol{Z}_{i}^{*} \boldsymbol{\gamma}\right)\right)^{(1 - C_i)(1 - C_i^O)}  \left\{\prod_{j=1}^{n_i} \left(1 - \operatorname{expit} \left(\boldsymbol{X}_{i j}^{*}  \boldsymbol{\eta} \right) \right)\right\}^{C_i(1 - C_i^O)} \operatorname{expit}\left(\boldsymbol{Z}_{i}^{*} \boldsymbol{\gamma}\right)^{C_i(1 - C_i^O)}.
\end{align*}
The complete data log likelihood is:
\begin{align*}
    &\hspace{0.3cm} \ell(\boldsymbol{C}, \boldsymbol{C}^{O}, \boldsymbol{R}, \boldsymbol{Y}, \boldsymbol{A}, \boldsymbol{Z}, \boldsymbol{X}; \boldsymbol{\eta}, \boldsymbol{\gamma} ) \\
    &= \sum_{i=1}^{M} C_i^O \left( \sum_{j=1}^{n_i}
      R_{ij} \log\left(\operatorname{expit}\left(\boldsymbol{X}_{i j}^{*} \boldsymbol{\eta}\right)\right) +  (1 - R_{ij})  \log\left(1 - \operatorname{expit}\left(\boldsymbol{X}_{i j}^{*} \boldsymbol{\eta}\right)  \right) \right) + C_i^O \log \left( \operatorname{expit}\left(\boldsymbol{Z}_{i}^{*} \boldsymbol{\gamma}\right) \right) + \\
    &\hspace{1cm}  (1 - C_i)(1 - C_i^O) \log\left(1 - \operatorname{expit}\left(\boldsymbol{Z}_{i}^{*} \boldsymbol{\gamma}\right)\right) + C_i(1 - C_i^O) \left(  \sum_{j=1}^{n_i} \log\left( 1 - \operatorname{expit} \left(\boldsymbol{X}_{i j}^{*}  \boldsymbol{\eta} \right) \right) \right) \\
    &\hspace{1.1cm} C_i(1 - C_i^O) \log\left(\operatorname{expit}\left(\boldsymbol{Z}_{i}^{*} \boldsymbol{\gamma}\right) \right)   \\
    &= \sum_{i=1}^{M} C_i^O \left\{ \log \left( \operatorname{expit}\left(\boldsymbol{Z}_{i}^{*} \boldsymbol{\gamma}\right) \right)  +  \sum_{j=1}^{n_i} R_{ij} \log\left(\operatorname{expit}\left(\boldsymbol{X}_{i j}^{*} \boldsymbol{\eta}\right)\right) +  (1 - R_{ij})  \log\left(1 - \operatorname{expit}\left(\boldsymbol{X}_{i j}^{*} \boldsymbol{\eta}\right)  \right) \right\} + \\
    &\hspace{1cm}  (1 - C_i)(1 - C_i^O) \log\left(1 - \operatorname{expit}\left(\boldsymbol{Z}_{i}^{*} \boldsymbol{\gamma}\right)\right) + \\
    &\hspace{1cm} C_i(1 - C_i^O)  \left( \log \left( \operatorname{expit}\left(\boldsymbol{Z}_{i}^{*} \boldsymbol{\gamma}\right) \right) + \sum_{j=1}^{n_i} \log \left(1 - \operatorname{expit} \left(\boldsymbol{X}_{i j}^{*}  \boldsymbol{\eta} \right)  \right)\right).
\end{align*}
The conditional expectation of the E step is given by
\begin{align*}
    &\hspace{0.4cm}Q\left(\boldsymbol{\eta}, \boldsymbol{\gamma}, \boldsymbol{\eta}^{(\nu)}, \boldsymbol{\gamma}^{(\nu)}, \boldsymbol{R}, \boldsymbol{Y}, \boldsymbol{A}, \boldsymbol{Z}, \boldsymbol{X}\right) \\
    &=\mathbb{E}_{\boldsymbol{\eta}^{(\nu)}, \boldsymbol{\gamma}^{(\nu)}}\left\{\ell(\boldsymbol{C},  \boldsymbol{R}, \boldsymbol{Y}, \boldsymbol{A}, \boldsymbol{Z}, \boldsymbol{X}; \boldsymbol{\eta}, \boldsymbol{\gamma}) \mid \boldsymbol{R}, \boldsymbol{Y}, \boldsymbol{A}, \boldsymbol{Z}, \boldsymbol{X}\right\} \\
    &= \sum_{i=1}^{M} C_i^O \left\{ \log \left( \operatorname{expit}\left(\boldsymbol{Z}_{i}^{*} \gamma\right) \right)  +  \sum_{j=1}^{n_i} R_{ij} \log\left(\operatorname{expit}\left(\boldsymbol{X}_{i j}^{*} \boldsymbol{\eta}\right)\right) +  (1 - R_{ij})  \log\left(1 - \operatorname{expit}\left(\boldsymbol{X}_{i j}^{*} \boldsymbol{\eta}\right)  \right) \right\} + \\
    &\hspace{1cm}  (1 - w_i^{(\nu)})(1 - C_i^O) \log\left(1 - \operatorname{expit}\left(\boldsymbol{Z}_{i}^{*} \gamma\right)\right) + \\
    &\hspace{1cm} w_i^{(\nu)}(1 - C_i^O)  \left( \log \left( \operatorname{expit}\left(\boldsymbol{Z}_{i}^{*} \gamma\right) \right) + \sum_{j=1}^{n_i} \log \left(1 - \operatorname{expit} \left(\boldsymbol{X}_{i j}^{*}  \boldsymbol{\eta} \right)  \right)\right),
\end{align*}
where
\begin{align*}
    w_{i}^{(\nu)} 
    &= \mathbb{E}_{\boldsymbol{\eta}^{(\nu)}, \boldsymbol{\gamma}^{(\nu)}}\left[C_i \mid  \boldsymbol{R}_i, Y_i, A_i, \boldsymbol{Z}_i, \boldsymbol{X}_i\right] \\
    &= P(C_i = 1 \mid \boldsymbol{R}_i, Y_i, A_i, \boldsymbol{Z}_i, \boldsymbol{X}_i; \boldsymbol{\eta}^{(\nu)}, \boldsymbol{\gamma}^{(\nu)}) \\
    &= P(C_i = 1 \mid \boldsymbol{R}_i, A_i, \boldsymbol{Z}_i, \boldsymbol{X}_i; \boldsymbol{\eta}^{(\nu)}, \boldsymbol{\gamma}^{(\nu)})  \\
    &= P(C_i = 1 \mid \boldsymbol{R}_i \neq \boldsymbol{0}, A_i, \boldsymbol{Z}_i, \boldsymbol{X}_i; \boldsymbol{\eta}^{(\nu)}, \boldsymbol{\gamma}^{(\nu)})I(\boldsymbol{R}_i \neq \boldsymbol{0}) \\
    &\hspace{0.3cm} + P(C_i = 1 \mid \boldsymbol{R}_i = \boldsymbol{0}, A_i, \boldsymbol{Z}_i, \boldsymbol{X}_i; \boldsymbol{\eta}^{(\nu)}, \boldsymbol{\gamma}^{(\nu)})I(\boldsymbol{R}_i = \boldsymbol{0}) \\
    &= I(\boldsymbol{R}_i \neq \boldsymbol{0}) + P(C_i = 1 \mid \boldsymbol{R}_i = \boldsymbol{0}, A_i, \boldsymbol{Z}_i, \boldsymbol{X}_i; \boldsymbol{\eta}^{(\nu)}, \boldsymbol{\gamma}^{(\nu)})I(\boldsymbol{R}_i = \boldsymbol{0}) \\
    &= C_i^O + P(C_i = 1 \mid \boldsymbol{R}_i = \boldsymbol{0}, A_i, \boldsymbol{Z}_i, \boldsymbol{X}_i; \boldsymbol{\eta}^{(\nu)}, \boldsymbol{\gamma}^{(\nu)})(1 - C_i^O) \\
    &= C_i^O + (1 - C_i^O) \times \\
    &\hspace{0.4cm} \left\{\frac{P(\boldsymbol{R}_i = \boldsymbol{0}|C_i = 1, A_i, \boldsymbol{Z}_i, \boldsymbol{X}_i; \boldsymbol{\eta}^{(\nu)})P(C_i = 1 \mid A_i, \boldsymbol{Z}_i, \boldsymbol{X}_i; \boldsymbol{\gamma}^{(\nu)})}{P(\boldsymbol{R}_i = \boldsymbol{0}, C_i = 1 \mid A_i, \boldsymbol{Z}_i, \boldsymbol{X}_i;\boldsymbol{\gamma}^{(\nu)},\boldsymbol{\eta}^{(\nu)}) +P(\boldsymbol{R}_i = \boldsymbol{0}, C_i = 0 \mid A_i, \boldsymbol{Z}_i, \boldsymbol{X}_i; \boldsymbol{\gamma}^{(\nu)}) } \right\} \\
    &= C_i^O + (1 - C_i^O) \left\{ \frac{\prod_{j=1}^{n_i} \left(1 -  \operatorname{expit}(\boldsymbol{X}_{ij}^* \boldsymbol{\eta}^{(\nu)})\right) \operatorname{expit}(\boldsymbol{Z}_{i}^* \boldsymbol{\gamma}^{(\nu)})}{\prod_{j=1}^{n_i} \left( 1 -  \operatorname{expit}(\boldsymbol{X}_{ij}^* \boldsymbol{\eta}^{(\nu)})\right) \operatorname{expit}(\boldsymbol{Z}_{i}^* \boldsymbol{\gamma}^{(\nu)}) + 1 - \operatorname{expit}(\boldsymbol{Z}_{i}^* \boldsymbol{\gamma}^{(\nu)}) }\right\} \\
    &= C_i^O + (1 - C_i^O) \left\{\frac{\prod_{j=1}^{n_i} \left( 1 -  \operatorname{expit}(\boldsymbol{X}_{ij}^* \boldsymbol{\eta}^{(\nu)})\right) \operatorname{expit}(\boldsymbol{Z}_{i}^* \boldsymbol{\gamma}^{(\nu)})}{1 - \operatorname{expit}(\boldsymbol{Z}_{i}^* \boldsymbol{\gamma}^{(\nu)}) \left[ 1 -  \prod_{j=1}^{n_i} \left(1 -  \operatorname{expit}(\boldsymbol{X}_{ij}^* \boldsymbol{\eta}^{(\nu)})\right) \right]}\right\}
\end{align*}

\subsection{Implementation of the EM algorithm}
\label{appdx:D2}

The algorithm below summarizes the implementation of the EM algorithm.

\begin{algorithm}[H]
\linespread{1.5}\selectfont
\caption{Implementation of the EM algorithm}\label{alg:EM}
\begin{algorithmic}[1]
\State Initialize $\nu = 0$ and $\varepsilon = 10^{-6}$
\State Initialize $\boldsymbol{\gamma}^{(\nu)}$ and $\boldsymbol{\eta}^{(\nu)}$ by fitting cluster- and individual-level PS models based on $C_i^O$ 
\While{TRUE} 
\State Compute $w_i^{(\nu)}$ based on $(\boldsymbol{\gamma}^{(\nu)}, \boldsymbol{\eta}^{(\nu)})$
\State $\boldsymbol{\gamma}^{(\nu+1)}= \underset{\boldsymbol{\gamma}}{\operatorname{argmax}} \hspace{0.1cm}Q\left(\boldsymbol{\gamma}, \boldsymbol{\eta}^{(\nu)}, \boldsymbol{C}^{O}, \boldsymbol{R}, \boldsymbol{Y}, \boldsymbol{A}, \boldsymbol{Z}, \boldsymbol{X}\right)$ 
\State $\boldsymbol{\eta}^{(\nu+1)}= \underset{\boldsymbol{\eta}}{\operatorname{argmax}} \hspace{0.1cm}Q\left(\boldsymbol{\gamma}^{(\nu+1)}, \boldsymbol{\eta},  \boldsymbol{C}^{O}, \boldsymbol{R}, \boldsymbol{Y}, \boldsymbol{A}, \boldsymbol{Z}, \boldsymbol{X}\right)$ 
\If{$\| \ell_c\left(\boldsymbol{\gamma}^{(\nu)}, \boldsymbol{\eta}^{(\nu)}\right) - \ell_c \left(\boldsymbol{\gamma}^{(\nu + 1)}, \boldsymbol{\eta}^{(\nu + 1)}\right)\|^2  \leq \varepsilon$}
\State break
\EndIf
\State $\nu = \nu + 1$
\EndWhile
\State Output: $\hat{\boldsymbol{\eta}}^{EM} = \boldsymbol{\eta}^{(\nu+1)}$, $\hat{\boldsymbol{\gamma}}^{EM} = \boldsymbol{\gamma}^{(\nu+1)}$
\end{algorithmic}
\end{algorithm}

\renewcommand{\theequation}{E.\arabic{equation}}
\setcounter{equation}{0}
\section{Details of Simulations Setup and Results} 
\label{appdx:E}
\subsection{Descriptive summaries of the outcome and baseline covariates of the Treatment of  Anemia in Malaria-Endemic Ghana Study} \label{appdx:E1}

\newpage

\begin{table}\label{tab:s2}
 \centering
 \def\~{\hphantom{0}}
 \begin{minipage}{175mm}
  \caption{Descriptive summary statistics of outcomes ($Y_{ij}$), cluster-level covariates ($\boldsymbol{Z}_{i}$), and individual-level covariates ($\boldsymbol{X}_{ij}$)}
  \begin{tabular*}{\textwidth}{@{}l@{\extracolsep{\fill}}c@{\extracolsep{\fill}}c@{\extracolsep{\fill}}}
  \hline
  \hline
  & {Iron Group}  & {No Iron Group} \\ 
  \cline{2-2} \cline{3-3} \\ [-10pt]
  & {780 clusters, 967 children} & {772 clusters, 991 children}  \\ 
  \hline
  & \multicolumn{2}{c}{Descriptive statistics of outcomes}\\ [1pt]
 \cline{2-3} \\ [-6pt]
 Malaria, incidence/100 Child-Years  ($Y_{ij}$) & 60.9 & 65.1  \\
 No. of missing clusters (\%) & 12 (1.5) & 19 (2.5) \\
 No. of missing participants (\%) & 25 (2.7) & 29 (2.9) \\
 \hline
 & \multicolumn{2}{c}{Descriptive statistics of cluster-level covariates}\\ [1pt]
 \cline{2-3} \\ [-6pt]
 Median cluster size (range) & 1 (1 to 4) & 1 (1 to 5) \\
 Household education, yes  & 66.5 \%  & 67.6 \%  \\
 Wealth, high & 51.2 \% & 50.9 \%  \\
 Complementary foods $\leq$ 6 mo  & 87.8 \% & 87.2 \% \\
 Household size, mean (sd)  & 6.1 (2.8) & 6.0 (2.8) \\ 
 \hline
 & \multicolumn{2}{c}{Descriptive statistics of individual-level covariates}\\ [1pt]
 \cline{2-3} \\ [-6pt]
 Sex, male & 51.3 \% & 50.7 \% \\
 Consumed at least 1 iron-fortified product & & \\ [-3pt]
 in previous 7 d prior to enrollment & 1.9 \% & 2.8 \% \\
 Age in month, mean (sd) {[}range{]} & 19.5 (8.6) {[}6 to 46{]} & 19.4 (8.6) {[}6 to 36{]} \\
 Wasting z score, mean (sd) & -0.64 (0.99) & -0.62 (0.97) \\
 Stunted growth z score, mean (sd) & -0.90 (1.18) & -0.88 (1.23) \\
 Underweight z score, mean (sd) & -0.93 (0.99) & -0.91 (0.97) \\
 Hemoglobin level, g/dL, mean (sd) & 10.3 (1.3) & 10.3 (1.3) \\
 \hline
\end{tabular*}
\end{minipage}
\vspace*{-6pt}
\end{table}

\subsection{Data generating processes}\label{appdx:E2}

\subsubsection{$\beta_A = 0$}

We generated the covariates to match summary statistics of Table S2. The continuous covariates, \textit{wasting z score}, \textit{stunted growth z score}, and \textit{underweight z score} were simulated from the multivariate normal distribution with corresponding observed sample mean, variance, and an assumed common correlation of 0.1. \textit{Age} was simulated from the truncated normal distribution with observed sample mean, variance, and bounds. \textit{Hemoglobin level} was simulated from the normal distribution with the observed sample mean and variance. For binary covariates, \textit{household education}, \textit{wealth}, \textit{complementary foods $\leq$ 6 months}, and \textit{gender} were simulated from the Bernoulli distribution with the observed proportion. \textit{Consumed at least 1 iron-fortified product} was generated under the logistic-normal model \citep{donner2000design} with an assumed between-cluster standard deviation of 0.05. For count data, \textit{cluster size} and \textit{household size} were simulated from a discrete uniform distribution with corresponding range. 

The treatment assignment ($A_i$) was simulated from the Bernoulli distribution with probability $p = 0.5$. We treated the primary outcome $Y_{ij}$ (incidence of malaria per 100 child-year) as a continuous variable, which was generated as
\begin{align} \label{example_null:OM1}
    Y_{ij} = \beta_{I}^{*} + \beta_{A}^{*} A_i + \boldsymbol{Z}_{i}^{O}\boldsymbol{\beta}_{Z} + \boldsymbol{X}_{ij}^{O}\boldsymbol{\beta}_{X} +   A_i\boldsymbol{Z}_{i}^{O}\boldsymbol{\beta}_{AZ} + A_i\boldsymbol{X}_{ij}^{O}\boldsymbol{\beta}_{AX} + \delta_{i} + \varepsilon_i,
\end{align}
where $\boldsymbol{Z}_{i}^{O}$ contained \textit{household size}, \textit{household education}, and \textit{wealth}, and $\boldsymbol{X}_{ij}^{O}$ contained \textit{age}, \textit{wasting z score}, and \textit{stunted growth z score}. $\delta_{i} \sim N(0,\sigma^2_\delta)$ was the cluster random intercept and $\varepsilon_i \sim N(0,\sigma^2_\varepsilon)$ was the residual standard error. 

The true marginal mean model induced by marginalizing over $(\boldsymbol{Z}_{i}^{O},\boldsymbol{X}_{ij}^{O}, A_i\boldsymbol{Z}_{i}^{O},A_i\boldsymbol{X}_{ij}^{O})$, the random intercept, and residual standard error was  $\mathbb{E}[Y_{ij}|A_i] = \beta_{I} + \beta_{A} A_i$. The study concluded that daily use of MNP with iron did not increase the risk of malaria during the post-intervention period, so $(\beta_I^*, \beta_A^*, \boldsymbol{\beta}_Z^{\mathrm{T}}, \boldsymbol{\beta}_X^{\mathrm{T}}, \boldsymbol{\beta}_{AZ}^{\mathrm{T}}, \boldsymbol{\beta}_{AX}^{\mathrm{T}})$ were chosen such that the intercept ($\beta_I$) was 63.5 and the marginal treatment effect ($\beta_A$) was 0. . The parameter values were $\beta_I^* = 43.5$, $\beta_A^* = 0.5$, $\boldsymbol{\beta}_Z^{\mathrm{T}} = (0.5, 0.8, 0.7)$, $\boldsymbol{\beta}_X^{\mathrm{T}} = (0.8, 1.5, -0.5)$, $\boldsymbol{\beta}_{AZ}^{\mathrm{T}} = (0.3, -1.5, 0.2)$, $\boldsymbol{\beta}_{AX}^{\mathrm{T}} = (-0.1, 0.5, -1.1)$, which yielded $(\beta_I, \beta_A) = (63.5, 0)$.

Missing outcome data processes were induced through the following logistic models:
\begin{equation} \label{example_null:PS1}
    \text{logit}\left(\lambda_i\left(A_{i},\boldsymbol{Z}_{i}^{C}; \boldsymbol{\gamma}\right)\right) = \gamma_{I} + \gamma_{A} A_i + \boldsymbol{Z}_{i}^{C}\boldsymbol{\gamma}_{\boldsymbol{Z}} + A_{i}\boldsymbol{Z}_{i}^{C}\boldsymbol{\gamma}_{A\boldsymbol{Z}},
\end{equation}
\begin{equation} \label{example_null:PS2}
    \text{logit}\left(\phi_{ij}\left(A_{i}, Z_{i}^{I}, \boldsymbol{X}_{ij}^{I} \mid C_{i}=1 ; \boldsymbol{\eta}\right)\right) = \eta_{I} + \eta_{A} A_i + \eta_{Z}Z_{i}^{I} + \boldsymbol{X}_{ij}^{I}\boldsymbol{\eta}_{X} + A_i\boldsymbol{X}_{ij}^{I}\boldsymbol{\eta}_{AX}.
\end{equation}
$\boldsymbol{Z}_{i}^{C}$ included \textit{household size}, \textit{household education}, and \textit{wealth}. $Z_{i}^{I}$ included \textit{Complementary foods $\leq$ 6 mo}. $\boldsymbol{X}_{ij}^{I}$ included \textit{wasting z score} and \textit{stunted growth z score}. Around 2\% of the clusters and 2.8\% of the overall participants were missing in the original study. For illustration, we inflated the missingness by choosing $\boldsymbol{\gamma}$ and $\boldsymbol{\eta}$ such that 12\% of the clusters were missing and 30\% of the overall participants had missing outcomes. The parameter values for the PS were $\{\gamma_I = 1.10, \gamma_A = -0.29, \boldsymbol{\gamma_Z}^{\mathrm{T}} = (0.18, -0.29, -0.22), \boldsymbol{\gamma}_{A\boldsymbol{Z}}^{\mathrm{T}} = (0.26, -0.51, -0.36)\}$ and $\{\eta_I = 1.73, \eta_A = -0.36, \eta_Z^{\mathrm{T}} = 0.10, \boldsymbol{\eta}_{\boldsymbol{X}}^{\mathrm{T}} = (-0.11, 0.18), \boldsymbol{\eta}_{A\boldsymbol{X}}^{\mathrm{T}} = (-0.36, 0.41)\}$, yielding $P(C_i = 1| A_i, \boldsymbol{Z}_i; \boldsymbol{\gamma}) = 0.14$ and $P(R_{ij} = 1|C_i = 1, A_i, \boldsymbol{Z}_i, \boldsymbol{X}_i; \boldsymbol{\eta}) = 0.20$.

\subsubsection{$\beta_A = 1.5$}
The data generating processes for $\beta_A = 1.5$ were similar to those for $\beta_A = 0$. The treatment assignment ($A_i$) was simulated from the Bernoulli distribution with probability $p = 0.5$. For the covariates, each participant had a vector of individual-level covariates $\boldsymbol{X}_{ij} = (X_{ij}^{1}, X_{ij}^{2}, X_{ij}^{3}, X_{ij}^{4})$, where $(X_{ij}^{1}, X_{ij}^{2}, X_{ij}^{3}) $ were drawn from a multivariate normal distribution with standard error $(1, 1.2, 0.8)$ and common correlation of 0.1 and $X_{ij}^{4} = (X_{ij}^1)^{2}$. Each cluster had a vector of cluster-level covariates $\boldsymbol{Z}_{i} = (Z_{i}^{1}, Z_{i}^{2}, Z_{i}^{3}, Z_{i}^{4})$. $(Z_{i}^{1}, Z_{i}^{2})$ were summary statistics of  $(X_{ij}^{1}, X_{ij}^{2})$, i.e., $Z_{i}^{k} = \frac{\sum_{j=1}^{n_i} X_{ij}^{k}}{n_i}$ for $k = 1, 2$. $Z_i^3$ was simulated from a Poisson distribution with rate 1.2 and $Z_i^4 \sim N(1.2,1)$. We considered the following outcome model:
\begin{align} \label{example_alt:OM1}
    Y_{ij} = \beta_{I}^{*} + \beta_{A}^{*} A_i + \boldsymbol{Z}_{i}\boldsymbol{\beta}_{Z} + \boldsymbol{X}_{ij}\boldsymbol{\beta}_{X}  +  A_i\boldsymbol{Z}_{i}\boldsymbol{\beta}_{AZ} + A_i\boldsymbol{X}_{ij} \boldsymbol{\beta}_{AX} + \delta_{i} + \varepsilon_i,
\end{align}
where $Y_{ij}$ was a continuous outcome, $\delta_{i} \sim N(0,\sigma^2_\delta)$ was the cluster random intercept, and $\varepsilon_i \sim N(0,\sigma^2_\varepsilon)$ was the residual standard error. The true marginal mean model induced by marginalizing over $(\boldsymbol{Z}_{i},\boldsymbol{X}_{ij}, A_i\boldsymbol{Z}_{i}, A_i\boldsymbol{X}_{ij})$, the random intercept, and residual standard error was  $\mathbb{E}[Y_{ij}|A_i] = \beta_{I} + \beta_{A} A_i$. The parameter values were $\beta_I^* = 0$, $\beta_A^* = 1.5$, $\boldsymbol{\beta}_Z^{\mathrm{T}} = (2, -2.5, 1, -1)$, $\boldsymbol{\beta}_X^{\mathrm{T}} = (1, 1.2, 0.5, -0.5)$, $\boldsymbol{\beta}_{AZ}^{\mathrm{T}} = (0.8, -0.4, 1, -1)$, $\boldsymbol{\beta}_{AX}^{\mathrm{T}} = (0.5, 0.3, 1, -1)$, which yielded $(\beta_I, \beta_A) = (0, 1.5)$. 

Missing outcome data processes were induced through the following logistic models:
\begin{equation} \label{example_alt:PS1}
    \text{logit}\left(\lambda_i\left(A_{i},\boldsymbol{Z}_{i}^{C}; \boldsymbol{\gamma}\right)\right) = \gamma_{I} + \gamma_{A} A_i + \boldsymbol{Z}_{i}^{C}\boldsymbol{\gamma}_{\boldsymbol{Z}} + A_{i}\boldsymbol{Z}_{i}^{C}\boldsymbol{\gamma}_{A\boldsymbol{Z}},
\end{equation}
\begin{equation} \label{example_alt:PS2}
    \text{logit}\left(\phi_{ij}\left(A_{i}, \boldsymbol{Z}_{i}^{I}, \boldsymbol{X}_{ij}^{I} \mid C_{i}=1 ; \boldsymbol{\eta}\right)\right) = \eta_{I} + \eta_{A} A_i + \boldsymbol{Z}_{i}^{I}\boldsymbol{\eta}_{Z} + \boldsymbol{X}_{ij}^{I}\boldsymbol{\eta}_{X} + A_i\boldsymbol{Z}_{i}^{I}\boldsymbol{\eta}_{AZ} + A_i\boldsymbol{X}_{ij}^{I}\boldsymbol{\eta}_{AX},
\end{equation}
where $\boldsymbol{Z}_i^{C} = (Z_i^3, Z_i^4)$, $\boldsymbol{Z}_i^{I} = Z_i^3$, and $\boldsymbol{X}_{ij}^{I} = (X_{ij}^1, X_{ij}^2, X_{ij}^3, X_{ij}^4)$. The parameter values for the PS were $\{\gamma_I = 2.44$, $\gamma_A = 0.18$, $\gamma_{\boldsymbol{Z}} = (0.12, -0.39)$, $\gamma_{A\boldsymbol{Z}} = (-0.22, -0.29)\}$ and $\{\eta_I = 1.73$, $\eta_A = -0.22$, $\eta_Z = -0.16$, $\boldsymbol{\eta}_{\boldsymbol{X}} = (0.18, 0.26, 0.03, 0.18)$, $\eta_{AZ} = -0.05$, $\boldsymbol{\eta}_{A\boldsymbol{X}} = (0.18, 0.26, -0.22, -0.29)\}$, yielding $P(C_i = 1| A_i, \boldsymbol{Z}_i; \boldsymbol{\gamma}) = 0.14$ and $P(R_{ij} = 1|C_i = 1, A_i, \boldsymbol{Z}_i, \boldsymbol{X}_i; \boldsymbol{\eta}) = 0.23$.

\subsubsection{Varying Cluster Sizes and ICC values}

To examine the impact of varying cluster sizes, we considered $M= 1552$ with $n_i \sim DU(1,4)$, $n_i \sim DU(1,5)$, and $n_i = 3$ (corresponding to small cluster sizes with large number of clusters), and $M = 300$ with $n_i \sim DU(30, 50)$ (corresponding to large cluster sizes with medium number of clusters), where $DU$ is the discrete uniform distribution. The iron-Malaria study reported an intracluster correlation coefficient (ICC) of 0.0804 \citep{ratovoson2022proactive}. ICC was calculated as $\frac{\sigma^2_{\delta}}{\sigma^2_{\delta} + \sigma^2_\varepsilon}$. We set $\sigma^2_\varepsilon$ to 0.5, which yielded $\sigma^2_\delta \approx 0.043$. We also considered large ICC value (0.2) by setting $(\sigma_\epsilon^2, \sigma_\delta^2) = (5,1.25)$. To summarize, our simulation scenario consist of $ 2 \times 3 = 6$ combinations of different marginal treatment effects (null and alternative) and varying cluster sizes (3 levels). 

\subsection{Analysis Approaches}

We compared the following four analysis approaches. First, we carried out an unweighted CC-GEE analysis. Second, we applied the standard IPW-GEE method, where the PS was estimated by the unconditional logistics model with the same functional form as Model (\ref{example_null:PS2}) for the $\beta_A = 0$ setting and Model (\ref{example_alt:PS2}) for the $\beta_A = 1.5$ setting but ignored cluster-level missingness. Third, we employed the MIPW-GEE method, where both the cluster- and individual-level PS models were correctly specified. To estimate the parameters in the PS models, we fitted the standard logistic regression models based on $C_i^O$ (denoted as MIPW-GEE-no-EM) and also applied the EM algorithm (denoted as MIPW-GEE-EM). Lastly, we implemented our proposed MMR-GEE estimator by specifying $\mathcal{P}_1 = \{\lambda_{i}^{k}(\boldsymbol{\gamma}^{k}), k = 1, 2\}$ and $\mathcal{P}_2 = \{\phi_{ij}^{\ell}(\boldsymbol{\eta}^{\ell}), \ell = 1, 2\}$ with all parameters in $\mathcal{P}_1$ and $\mathcal{P}_2$ estimated by the EM algorithm. Both $\mathcal{P}_1$ and $\mathcal{P}_2$ contained one correctly specified and one misspecified models. 

For $\beta_A = 0$ setting, $\lambda_{i}^{1}(\boldsymbol{\gamma}^{1})$ and $\phi_{ij}^{1}(\boldsymbol{\eta}^{1})$ were the correctly specified models. $\lambda_{i}^{2}(\boldsymbol{\gamma}^{2})$ failed to include \textit{household size}, \textit{household education}, \textit{wealth}, and their interactions with the treatment indicator. $\phi_{ij}^{2}(\boldsymbol{\eta}^{2})$ 
failed to include \textit{wasting z score}, \textit{stunted growth z score}, and all treatment-covariates interaction terms. Moreover, both $\lambda_{i}^{2}(\boldsymbol{\gamma}^{2})$ and $\phi_{ij}^{2}(\boldsymbol{\eta}^{2})$ contained additional covariates, with $\lambda_{i}^{2}(\boldsymbol{\gamma}^{2})$ including \textit{complementary foods $\leq$ 6 months} and its treatment interaction term and $\phi_{ij}^{2}(\boldsymbol{\eta}^{2})$ including
\textit{age}, \textit{underweight z score}. For $\beta_A = 1.5$ setting, $\lambda_{i}^{1}(\boldsymbol{\gamma}^{1})$ and $\phi_{ij}^{1}(\boldsymbol{\eta}^{1})$ were the correctly specified PS  models. $\lambda_{i}^{2}(\boldsymbol{\gamma}^{2})$ failed to include $(Z_3, Z_4, AZ_3, AZ_4)$ and $\phi_{ij}^{2}(\boldsymbol{\eta}^{2})$ 
failed to include $(Z_3, X_1, X_3, X_4, AZ_3, AX_1, AX_3, AX_4)$. Furthermore, both $\lambda_{i}^{2}(\boldsymbol{\gamma}^{2})$ and $\phi_{ij}^{2}(\boldsymbol{\eta}^{2})$ contained additional covariates
$(Z_1, AZ_1)$.

\subsection{Simulation results}
Figures (\ref{fig:null2}) and (\ref{fig:null3}) are the empirical distribution of effect estimates under the null with $\{n_i = 3, M = 1552, ICC = 0.0804\}$ and $\{n_i \sim DU(30,50), M = 300, ICC = 0.2\}$, respectively. Figures (\ref{fig:alt1}),(\ref{fig:alt2}), and (\ref{fig:alt3}) are the empirical distribution of effect estimates under the alternative with $\{n_i \sim DU(1,4), M = 1552, ICC = 0.0804\}$, $\{n_i = 3, M = 1552, ICC = 0.0804\}$ and $\{n_i \sim DU(30,50), M = 300, ICC = 0.2\}$, respectively. 

\newpage
\begin{figure}
    \includegraphics[scale = 0.7]{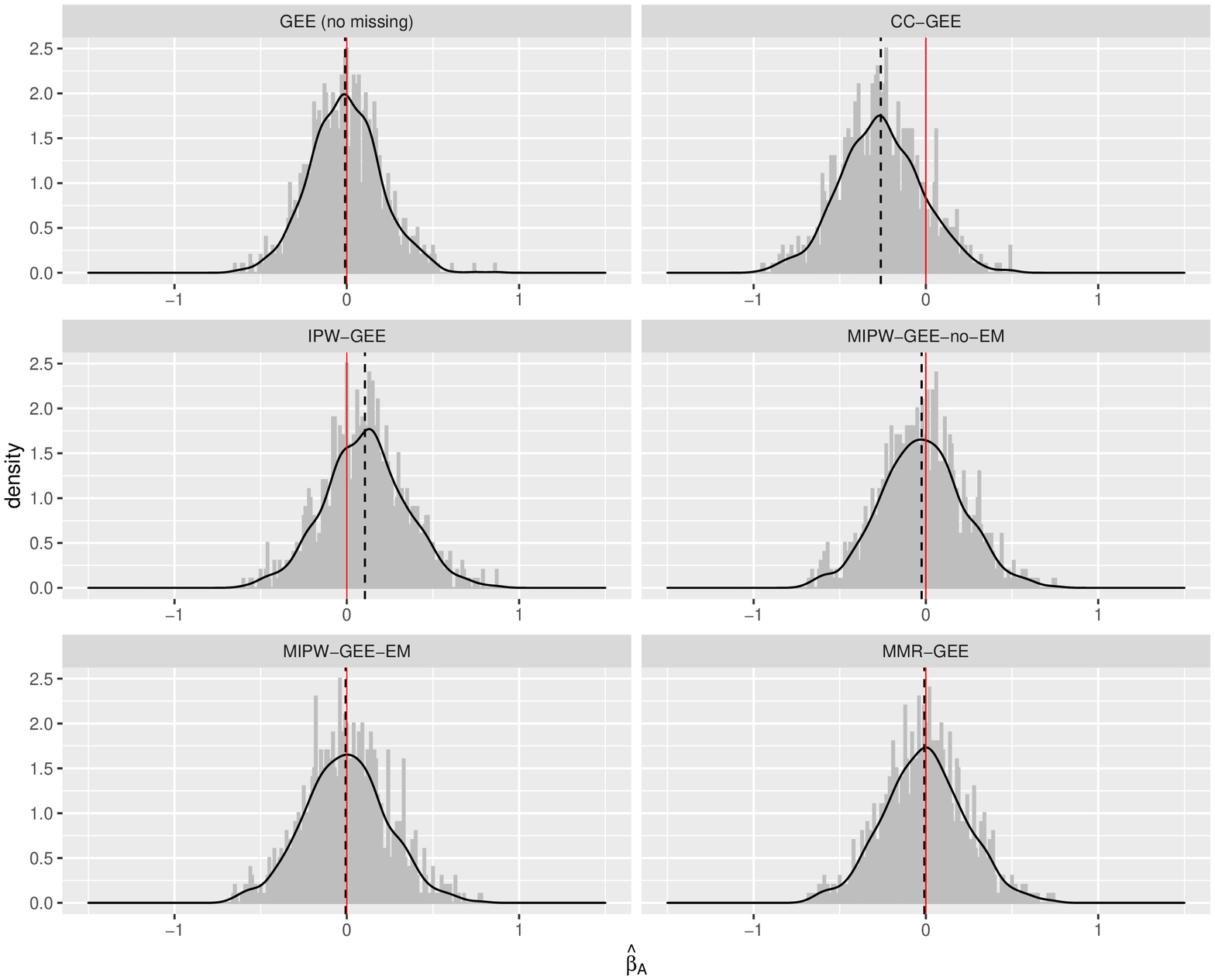}
    \caption{Empirical distribution of the estimated marginal treatment effect ($\hat{\beta}_A$) based on 1,000 replicates under $n_i =3, M = 1552, ICC = 0.08$ setting. The red line denotes the true marginal treatment effect (0). The dotted line denotes the empirical mean of the estimated marginal treatment effect.}
    \label{fig:null2}
\end{figure}

\newpage
\begin{figure}
    \includegraphics[scale = 0.7]{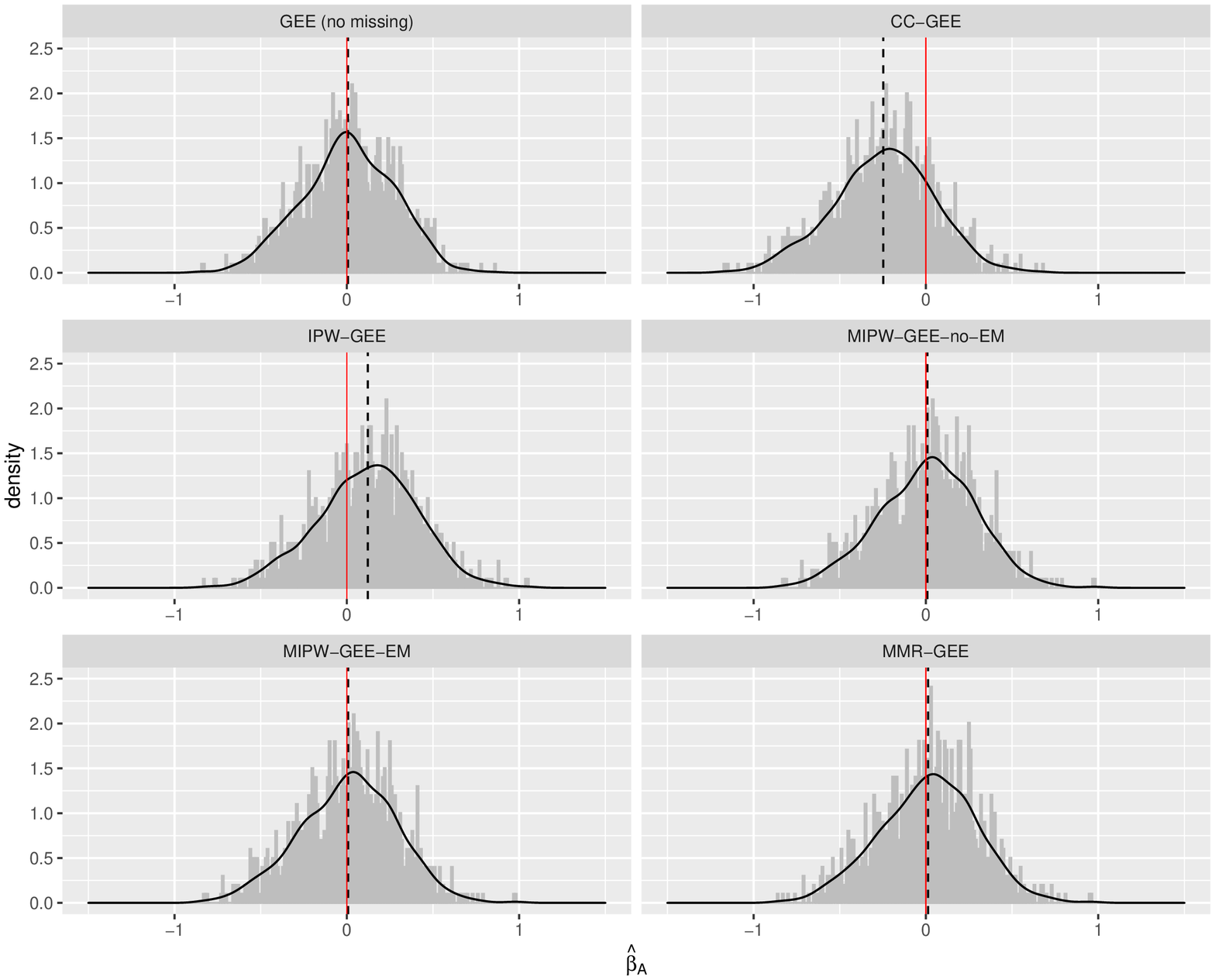}
    \caption{Empirical distribution of the estimated marginal treatment effect ($\hat{\beta}_A$) based on 1,000 replicates under $n_i \sim DU(30,50), M = 300, ICC = 0.2$ setting. The red line denotes the true marginal treatment effect (0). The dotted line denotes the empirical mean of the estimated marginal treatment effect.}
    \label{fig:null3}
\end{figure}

\newpage
\begin{figure}
    \centering
    \includegraphics[scale = 0.7]{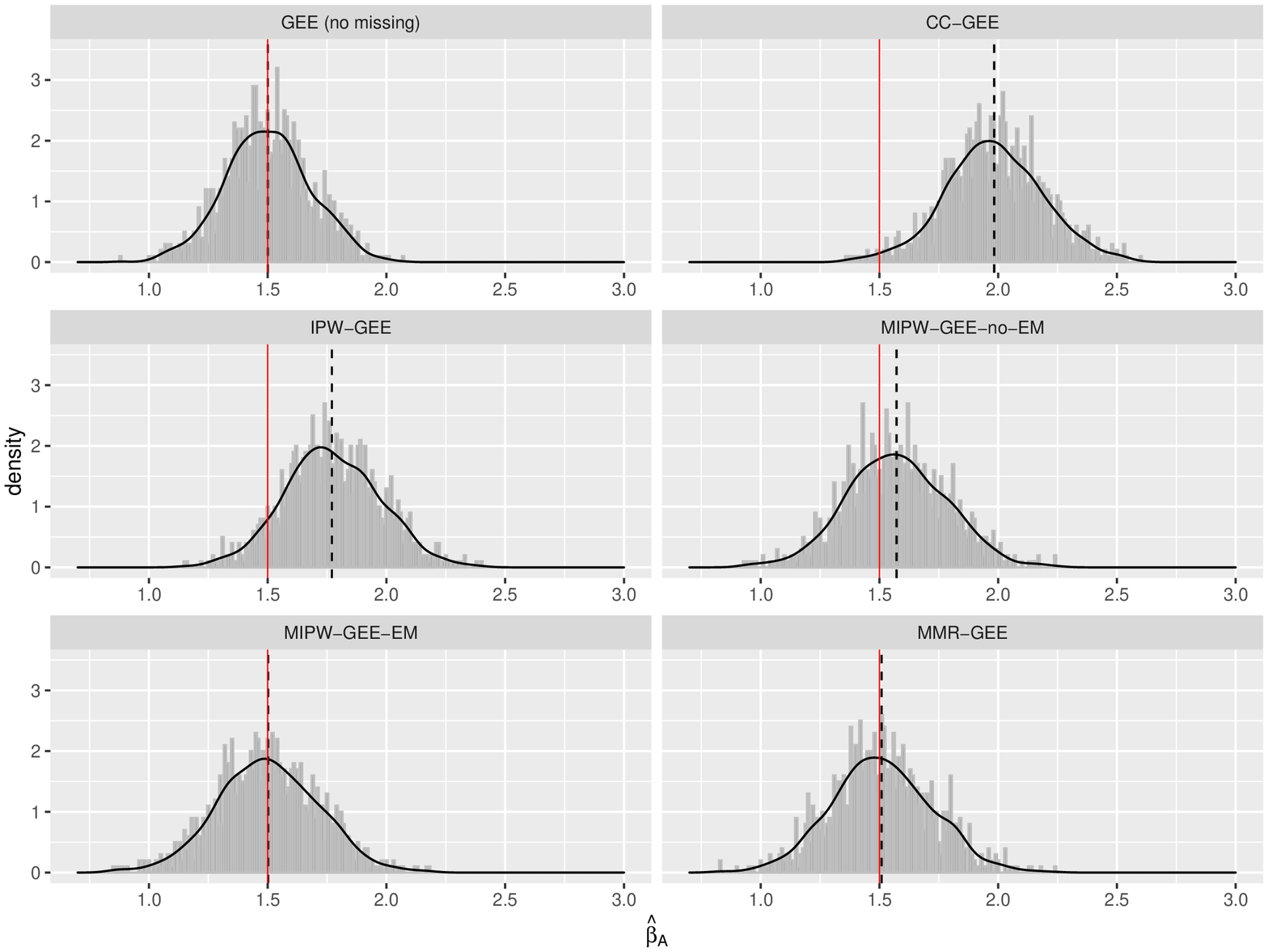}
    \caption{Empirical distribution of the estimated marginal treatment effect ($\hat{\beta}_A$) based on 1,000 replicates under $n_i \sim DU(1,4), M = 1552, ICC = 0.08$ setting. The red line denotes the true marginal treatment effect (1.5). The dotted line denotes the empirical mean of the estimated marginal treatment effect.}
    \label{fig:alt1}
\end{figure}

\newpage
\begin{figure}
    \centering
    \includegraphics[scale = 0.7]{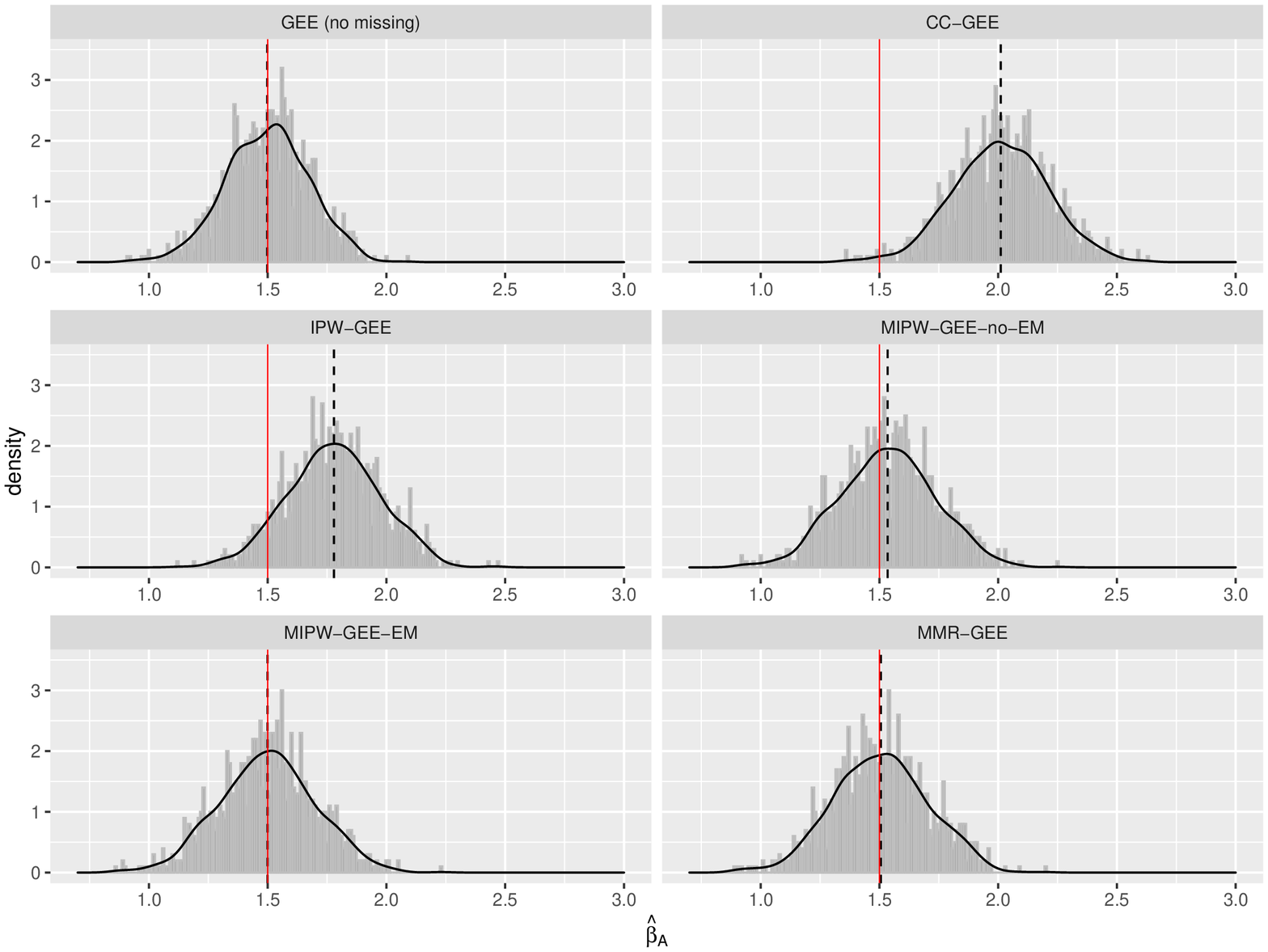}
    \caption{Empirical distribution of the estimated marginal treatment effect ($\hat{\beta}_A$) based on 1,000 replicates under $n_i = 3, M = 1552, ICC = 0.08$ setting. The red line denotes the true marginal treatment effect (1.5). The dotted line denotes the empirical mean of the estimated marginal treatment effect.}
    \label{fig:alt2}
\end{figure}

\newpage
\begin{figure}
    \centering
    \includegraphics[scale = 0.7]{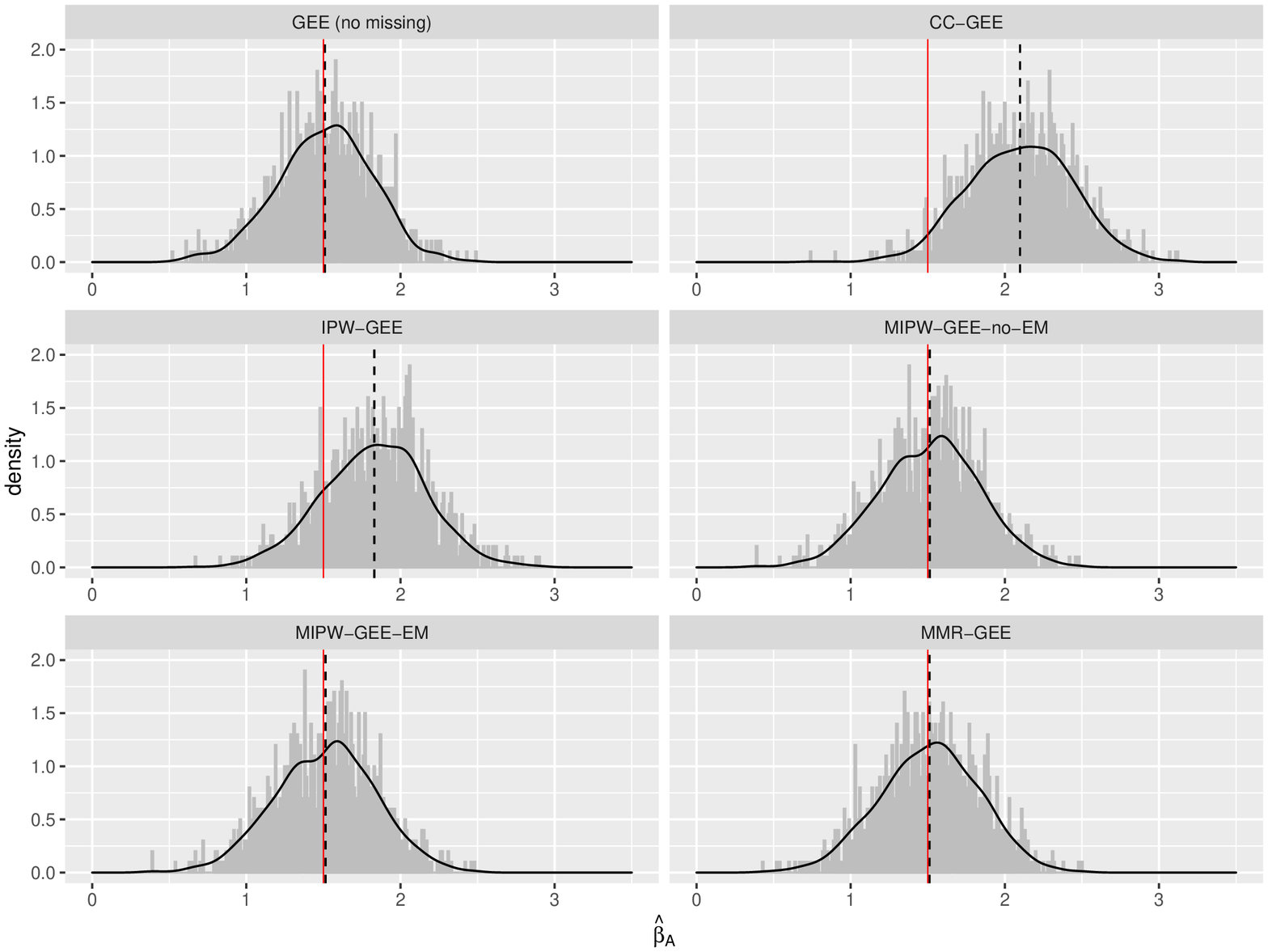}
    \caption{Empirical distribution of the estimated marginal treatment effect ($\hat{\beta}_A$) based on 1,000 replicates under $n_i \sim DU(30,50), M = 300, ICC = 0.2$ setting. The red line denotes the true marginal treatment effect (1.5). The dotted line denotes the empirical mean of the estimated marginal treatment effect.}
    \label{fig:alt3}
\end{figure}

\clearpage

\section{PS Model Fitting Results of the Pro-CCM Study}

\begin{table}[h!]
\label{tab:s3}
 \centering
 \def\~{\hphantom{0}}
 \begin{minipage}{175mm}
  \caption{Model fitting results for the PS}
  \begin{tabular*}{\textwidth}{@{}l@{\extracolsep{\fill}}c@{}c@{}c@{}c@{}c@{}c@{}}
  \hline
  \hline
  &  \multicolumn{6}{c}{p-value for association with} \\
  \cmidrule(lr){2-7} \\ [-15pt]
  &  \multicolumn{2}{c}{$P(R_{ijk} = 1)$\footnote{Wald test for $\theta_{X}$ and $\theta_{AX}$ in the regression $\text{logit}[P(R_{ijk} = 1)] = \theta_0 + \theta_A A_{i} + \theta_X X_{ijk} + \theta_{AX} A_i X_{ijk}$}} & \multicolumn{2}{c}{$P(C_{ij}^O = 1)$\footnote{Wald test for $\gamma_{X}$ and $\gamma_{AX}$ in the regression $\text{logit}[P(C_{ij}^O = 1)] = \gamma_I + \gamma_A A_{i} + \gamma_X X_{ij} + \gamma_{AX} A_i X_{ij}$}} & \multicolumn{2}{c}{$P(R_{ijk} = 1|C_{ij}^O = 1)$\footnote{Wald test for $\eta_{X}$ and $\eta_{AX}$ in the regression $\text{logit}[P(R_{ijk} = 1|C_{ij}^O = 1)] = \eta_I + \eta_A A_{i} + \eta_X X_{ijk} + \eta_{AX} A_i X_{ijk}$}} \\ 
  \cmidrule(lr){2-3}  \cmidrule(lr){4-5} \cmidrule(lr){6-7} \\ [-15pt] 
  &  $\theta_{X} = 0$ & $\theta_{AX} = 0$ & $\gamma_{X} = 0$ & $\gamma_{AX} = 0$ & $\eta_{X} = 0$ & $\eta_{AX} = 0$  \\ 
  \hline \\ [-10pt]
  Household-level covariates  \\ 
  \hline \\ [-10pt]   
 Household size & \textbf{\textless 0.01} & \textbf{\textless 0.01} & \textbf{\textless 0.01} & 0.89 & \textbf{\textless 0.01} & $\boldsymbol{0.01}$ \\
 \% of males & \textbf{\textless 0.01} & 0.53 & 0.56  & 0.92 & \textbf{\textless 0.01} & \textbf{0.01} \\
 Highest educ. & 0.14 & 0.92 & 0.22 & 0.69 & $\boldsymbol{0.02}$ & 0.15 \\
 IRS  & 0.88 & \textbf{\textless 0.01} & $\boldsymbol{0.05}$ & $\boldsymbol{0.01}$ & \textbf{\textless 0.01}  & 0.18\\
 \hline \\ [-10pt]
  Individual-level covariates & & &  \\ 
  \hline \\ [-10pt]       
 Sex, male& 0.12 & 0.15 &  & & \textbf{\textless 0.01} & 0.18 \\
 Age, median (sd) & \textbf{\textless 0.01} & \textbf{\textless 0.01} & & & \textbf{\textless 0.01} & 0.30  \\
 Primary school & 0.18 & 0.43 & & & \textbf{\textless 0.01} & 0.99 \\
 Secondary school & \textbf{\textless 0.01} & 0.20 & & & 0.28 &\textbf{\textless 0.01} \\
 High-level school & 0.71 & 0.06 & & & 0.42 & 0.18 \\
 Sleep in mosquito nets & 0.21 & \textbf{\textless 0.01}  & & & \textbf{\textless 0.01} & 0.60\\
 Sleep in the yard & 0.22 & 0.62 & & & 0.34 & 0.30\\
 \hline
\end{tabular*}
\end{minipage}
\vspace*{-6pt}
\end{table}

\bibliographystyle{plainnat}
\bibliography{supplementary.bib}